\pdfoutput=1
\documentclass[sigconf, screen=true, review=false, printacmref=false, printccs=false, printfolios=false, final]{acmart}
\usepackage{paralist}
\setcopyright{none}
\usepackage{nicefrac}
\definecolor{thelightblue}{RGB}{0,191,255}

\usepackage{nicefrac}
\newcommand{\mychoose}[2]{\left( \begin{smallmatrix} #1 \\ #2 \end{smallmatrix} \right)} 

\usepackage{mathtools}
\usepackage{blkarray, bigstrut} 

\usepackage{tikz}
\usepackage{verbatim}
\usetikzlibrary{arrows}
\usetikzlibrary{shapes,snakes}
\usetikzlibrary{decorations.pathmorphing} 
\usetikzlibrary{fit}					
\usetikzlibrary{backgrounds}	

\makeatletter
\global\let\tikz@ensure@dollar@catcode=\relax
\makeatother

\usepackage{subfigure}
\usepackage{xcolor}

\definecolor{thelightblue}{RGB}{0,191,255}
\definecolor{theblue}{RGB}{0,0,180}

\usepackage{blkarray}
\usepackage{bigstrut, booktabs}

\makeatletter
\renewcommand*\env@matrix[1][*\c@MaxMatrixCols c]{
\hskip -\arraycolsep
\let\@ifnextchar\new@ifnextchar
\array{#1}}
\makeatother

\usepackage{booktabs} 
\usepackage{enumitem}

\usepackage{subfigure}
\usepackage{rotating}

\usepackage{multirow}

\usepackage{color}
\usepackage{xcolor}
\definecolor{mydarkblue}{RGB}{0, 20, 159} 
\definecolor{mydarkblue}{rgb}{0,0.08,0.45} 

\DeclareSymbolFont{cmbrightop}{OT1}{cmbr}{m}{n}
\DeclareMathSymbol{\sfPsi}{\mathalpha}{cmbrightop}{9}

\usepackage{tikz}
\usepackage{verbatim}
\usetikzlibrary{arrows}
\usetikzlibrary{shapes,snakes}
\usetikzlibrary{decorations.pathmorphing} 
\usetikzlibrary{fit}					
\usetikzlibrary{backgrounds}	

\definecolor{gray}{RGB}{150,150,150}
\definecolor{theblue}{RGB}{0, 20, 159} 

\definecolor{myyellow}{RGB}{255,255,204}
\definecolor{myred}{RGB}{255,204,204}
\definecolor{myblue}{RGB}{0,200,255}
\definecolor{mygreen}{RGB}{80,220,80}

\newcommand{\etal}{\emph{et al.}}

\newcommand{\eg}{\emph{e.g.}}
\newcommand{\ie}{\emph{i.e.}}

\newcommand{\wrt}{\emph{w.r.t.}\ }

\newtheorem{thm}{Theorem}
\newtheorem{cor}{Corollary}

\newtheorem{Property}{Property}
\newtheorem{Proposition}{Proposition}

\newtheorem{Definition}{Definition}

\usepackage{balance}

\usepackage{epsfig}

\usepackage{graphicx}

\usepackage{amssymb}
\usepackage{amsmath}
\usepackage{amsfonts}

\usepackage{tabularx}
\usepackage{booktabs}
\usepackage{relsize}

\usepackage{numprint} 

\usepackage{setspace}
\graphicspath{{./}{./graphics/}}
\newcolumntype{H}{>{\setbox0=\hbox\bgroup}c<{\egroup}@{}}

\usepackage{algorithm}
\usepackage{algorithmicx}
\usepackage{algpseudocode}
\algblockdefx[parallel]{ParFor}{EndPar}[1][]{$\textbf{parallel for}$ #1 $\textbf{do}$}{$\textbf{end parallel}$}
\algrenewcommand{\alglinenumber}[1]{\fontsize{6.5}{7}\selectfont#1}
\algtext*{EndPar}

\algblockdefx[parallel]{parfor}{endpar}[1][]{$\textbf{parallel for}$ #1 $\textbf{do}$}{$\textbf{end parallel}$}
\algrenewcommand{\alglinenumber}[1]{\scriptsize#1:}

\algblockdefx[parallel]{parallelfor}{parallelend}
[1][]{\textbf{parallel for} #1}
{\textbf{end parallel}}

\usepackage{nicefrac}

\algnotext{EndFor}
\algnotext{EndIf}
\algnotext{EndProcedure}

\newcommand{\setAlgFontSize}{\fontsize{8pt}{9pt}\selectfont}

\newcommand{\multilinenospaceD}[1]{\State \parbox[t]{\dimexpr0.96 \linewidth-\algorithmicindent}{\begin{spacing}{1.1}\setAlgFontSize#1\strut \end{spacing}}}

\usepackage[notheorems]{rr-math}
\newcommand{\hash}{\ensuremath{c}} 

\definecolor{red}{RGB}{0,0,0} 

\begin{document}
\title{Heterogeneous Network Motifs}

\settopmatter{authorsperrow=4}
\author{Ryan A. Rossi}
\orcid{1234-5678-9012-3456}
\affiliation{
\institution{Adobe Research}
}
\author{Nesreen K. Ahmed}
\affiliation{
\institution{Intel Labs}
}
\author{Aldo Carranza}
\affiliation{
\institution{Stanford University}
}
\author{David Arbour}
\affiliation{
\institution{Adobe Research}
}
\author{Anup Rao}
\affiliation{
\institution{Adobe Research}
}
\author{Sungchul Kim}
\affiliation{
\institution{Adobe Research}
}
\author{Eunyee Koh}
\affiliation{
\institution{Adobe Research}
}
\email{}

\renewcommand{\shortauthors}{R.~A.~Rossi et al.}

\begin{abstract}
Many real-world applications give rise to \emph{large heterogeneous networks} where nodes and edges can be of any arbitrary type (\eg, user, web page, location).
Special cases of such heterogeneous graphs include homogeneous graphs, bipartite, k-partite, signed, labeled graphs, among many others.
In this work, we generalize the notion of network motifs to heterogeneous networks.
In particular, small induced typed subgraphs called \emph{typed graphlets} (heterogeneous network motifs) are introduced and shown to be the fundamental building blocks of complex heterogeneous networks.
Typed graphlets are a powerful generalization of the notion of graphlet (network motif) to heterogeneous networks as they capture both the induced subgraph of interest and the types associated with the nodes in the induced subgraph.
To address this problem, we propose a fast, parallel, and space-efficient framework for counting typed graphlets in large networks.
We discover the existence of non-trivial combinatorial relationships between lower-order ($k\!-\!1$)-node typed graphlets and leverage them for deriving many of the \emph{$k$-node typed graphlets} in $o(1)$ constant time.
Thus, we avoid explicit enumeration of those typed graphlets.
Notably, the time complexity matches the best untyped graphlet counting algorithm.
The experiments demonstrate the effectiveness of the proposed framework in terms of runtime, space-efficiency, parallel speedup, and scalability as it is able to handle large-scale networks.
\end{abstract}



\maketitle

\section{Introduction} \label{sec:intro}
Higher-order connectivity patterns such as small induced subgraphs called graphlets (network motifs)\footnote{The terms graphlet, network motif, and induced subgraph are used interchangeably.} are known to be the fundamental building blocks of simple homogeneous networks~\cite{Milo2002} and are essential for modeling and understanding the fundamental components of these networks~\cite{pgd,pgd-kais,benson2016higher}.
Furthermore, graphlets are important for many predictive and descriptive modeling application tasks~\cite{zhang-image-categorization-via-graphlets,vishwanathan2010graph,shervashidze2009efficient,Milo2002,prvzulj2004modeling,milenkovic2008uncovering,hayes2013graphlet,ahmed17streams,lichtenwalter2012vertex} such as 
image processing and computer vision~\cite{zhang-image-categorization-via-graphlets,zhang2013probabilistic}, 
network alignment~\cite{koyuturk2006pairwise,prvzulj2007biological,milenkovic2008uncovering,crawford2015great}, 
classification~\cite{vishwanathan2010graph,shervashidze2009efficient}, 
visualization and sensemaking~\cite{pgd,pgd-kais}, 
dynamic network analysis~\cite{kovanen2011temporal,hulovatyy2015exploring}, 
community detection~\cite{radicchi2004defining,palla2005uncovering,solava2012graphlet,benson2016higher}, 
role discovery~\cite{ahmed17aaai,role2vec}, 
anomaly detection~\cite{noble2003graph,akoglu2015graph}, 
and link prediction~\cite{Rossi2018a}. 

\definecolor{typeOneColor}{RGB}{8,81,156} 
\definecolor{typeTwoColor}{RGB}{222,45,38} 
\definecolor{typeThreeColor}{RGB}{49,163,84} 
\definecolor{typeFourColor}{RGB}{117,107,177} 

\makeatletter
\global\let\tikz@ensure@dollar@catcode=\relax
\makeatother
\tikzstyle{every node}=[font=\large,line width=1.5pt]
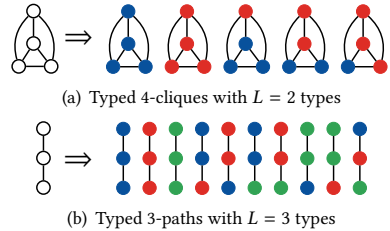
\begin{figure}[t!]

\subfigure[Typed 4-cliques with $L=2$ types]{
\scalebox{0.70}{
\scalebox{0.28}{
\begin{tikzpicture}[-,>=latex,auto,node distance=2.2cm,thick,
main node/.style={circle,draw=black,fill=white,draw,font=\sffamily\Huge\bfseries,text=black,minimum width=0.9cm, line width=1mm},
]

\node[main node] (1) {}; 
\node[main node] (2) [right of=1] {}; 
\node[main node] (3) [above right of=1, left=1.5pt] {}; 
\node[main node] (4) [above of=3] {}; 

\tikzstyle{LabelStyle}=[below=3pt]
\path[every node/.style={font=\sffamily}] 
	(1) edge [bend left,line width=1.0mm] node[above right] {} (4)
	(2) edge [bend right,line width=1.0mm] node[above right] {} (4)
(1) edge [line width=1.0mm, left] node [above left] {} (2) 
	(1)  edge [line width=1.0mm, right] node[below right] {} (3)
(2) edge [line width=1.0mm, left] node[below left] {} (3)
	(3) edge [line width=1.0mm,right] node[above right] {} (4);
\end{tikzpicture}
}
\hspace{-1.4mm}
\scalebox{0.28}{
\begin{tikzpicture}[-,>=latex,auto,node distance=2.0cm,thick,
main node/.style={circle,draw=black,fill=white,draw,font=\sffamily\Huge\bfseries,text=black,minimum width=0.9cm},
white node/.style={draw=white,draw,font=\sffamily\Huge\bfseries,text=black,minimum width=0.9cm}
]
\node[white node] (1) {};
\node[white node] (2) [below of=1,left=1.5pt] {};
\node[white node] (3) [below of=2] {};
\node[white node] (4) [right of=2, left=60pt, below=10pt, above=-0.05pt] {\vspace{-2mm}\fontsize{56}{56}\selectfont $\Rightarrow$};
\end{tikzpicture}
}
\hspace{0mm}
\scalebox{0.28}{
\begin{tikzpicture}[-,>=latex,auto,node distance=2.2cm,thick,
main node/.style={circle,draw=black,fill=black,draw,font=\sffamily\Huge\bfseries,text=white,minimum width=0.9cm},
white/.style={circle,draw=white,fill=white,draw,font=\sffamily,text=white,minimum width=0.001cm},
typeOne node/.style={circle,draw=typeOneColor,fill=typeOneColor,draw,font=\sffamily\Huge\bfseries,text=white,minimum width=0.9cm},
typeTwo node/.style={circle,draw=typeTwoColor,fill=typeTwoColor,draw,font=\sffamily\Huge\bfseries,text=white,minimum width=0.9cm},
]

\node[typeOne node] (1) {}; 
\node[typeOne node] (2) [right of=1] {}; 
\node[typeOne node] (3) [above right of=1, left=1.5pt] {}; 
\node[typeOne node] (4) [above of=3] {}; 

\tikzstyle{LabelStyle}=[below=3pt]
\path[every node/.style={font=\sffamily}] 
	(1) edge [bend left,line width=1.0mm] node[above right] {} (4)
	(2) edge [bend right,line width=1.0mm] node[above right] {} (4)
(1) edge [line width=1.0mm, left] node [above left] {} (2) 
	(1)  edge [line width=1.0mm, right] node[below right] {} (3)
(2) edge [line width=1.0mm, left] node[below left] {} (3)
	(3) edge [line width=1.0mm,right] node[above right] {} (4);
\end{tikzpicture}
}
\hspace{0.3mm}
\scalebox{0.28}{
\begin{tikzpicture}[-,>=latex,auto,node distance=2.2cm,thick,
main node/.style={circle,draw=black,fill=black,draw,font=\sffamily\Huge\bfseries,text=white,minimum width=0.9cm},
white/.style={circle,draw=white,fill=white,draw,font=\sffamily,text=white,minimum width=0.001cm},
typeOne node/.style={circle,draw=typeOneColor,fill=typeOneColor,draw,font=\sffamily\Huge\bfseries,text=white,minimum width=0.9cm},
typeTwo node/.style={circle,draw=typeTwoColor,fill=typeTwoColor,draw,font=\sffamily\Huge\bfseries,text=white,minimum width=0.9cm},
]

\node[typeTwo node] (1) {}; 
\node[typeTwo node] (2) [right of=1] {}; 
\node[typeTwo node] (3) [above right of=1, left=1.5pt] {}; 
\node[typeTwo node] (4) [above of=3] {}; 

\tikzstyle{LabelStyle}=[below=3pt]
\path[every node/.style={font=\sffamily}] 
	(1) edge [bend left,line width=1.0mm] node[above right] {} (4)
	(2) edge [bend right,line width=1.0mm] node[above right] {} (4)
(1) edge [line width=1.0mm, left] node [above left] {} (2) 
	(1)  edge [line width=1.0mm, right] node[below right] {} (3)
(2) edge [line width=1.0mm, left] node[below left] {} (3)
	(3) edge [line width=1.0mm,right] node[above right] {} (4);
\end{tikzpicture}
}
\hspace{0.3mm}
\scalebox{0.28}{
\begin{tikzpicture}[-,>=latex,auto,node distance=2.2cm,thick,
main node/.style={circle,draw=black,fill=black,draw,font=\sffamily\Huge\bfseries,text=white,minimum width=0.9cm},
white/.style={circle,draw=white,fill=white,draw,font=\sffamily,text=white,minimum width=0.001cm},
typeOne node/.style={circle,draw=typeOneColor,fill=typeOneColor,draw,font=\sffamily\Huge\bfseries,text=white,minimum width=0.9cm},
typeTwo node/.style={circle,draw=typeTwoColor,fill=typeTwoColor,draw,font=\sffamily\Huge\bfseries,text=white,minimum width=0.9cm},
]

\node[typeOne node] (1) {}; 
\node[typeOne node] (2) [right of=1] {}; 
\node[typeOne node] (3) [above right of=1, left=1.5pt] {}; 
\node[typeTwo node] (4) [above of=3] {}; 

\tikzstyle{LabelStyle}=[below=3pt]
\path[every node/.style={font=\sffamily}] 
	(1) edge [bend left,line width=1.0mm] node[above right] {} (4)
	(2) edge [bend right,line width=1.0mm] node[above right] {} (4)
(1) edge [line width=1.0mm, left] node [above left] {} (2) 
	(1)  edge [line width=1.0mm, right] node[below right] {} (3)
(2) edge [line width=1.0mm, left] node[below left] {} (3)
	(3) edge [line width=1.0mm,right] node[above right] {} (4);
\end{tikzpicture}
}
\hspace{0.3mm}
\scalebox{0.28}{
\begin{tikzpicture}[-,>=latex,auto,node distance=2.2cm,thick,
main node/.style={circle,draw=black,fill=black,draw,font=\sffamily\Huge\bfseries,text=white,minimum width=0.9cm},
white/.style={circle,draw=white,fill=white,draw,font=\sffamily,text=white,minimum width=0.001cm},
typeOne node/.style={circle,draw=typeOneColor,fill=typeOneColor,draw,font=\sffamily\Huge\bfseries,text=white,minimum width=0.9cm},
typeTwo node/.style={circle,draw=typeTwoColor,fill=typeTwoColor,draw,font=\sffamily\Huge\bfseries,text=white,minimum width=0.9cm},
]

\node[typeOne node] (1) {}; 
\node[typeOne node] (2) [right of=1] {}; 
\node[typeTwo node] (3) [above right of=1, left=1.5pt] {}; 
\node[typeTwo node] (4) [above of=3] {}; 

\tikzstyle{LabelStyle}=[below=3pt]
\path[every node/.style={font=\sffamily}] 
	(1) edge [bend left,line width=1.0mm] node[above right] {} (4)
	(2) edge [bend right,line width=1.0mm] node[above right] {} (4)
(1) edge [line width=1.0mm, left] node [above left] {} (2) 
	(1)  edge [line width=1.0mm, right] node[below right] {} (3)
(2) edge [line width=1.0mm, left] node[below left] {} (3)
	(3) edge [line width=1.0mm,right] node[above right] {} (4);
\end{tikzpicture}
}
\hspace{0.3mm}
\scalebox{0.28}{
\begin{tikzpicture}[-,>=latex,auto,node distance=2.2cm,thick,
main node/.style={circle,draw=black,fill=black,draw,font=\sffamily\Huge\bfseries,text=white,minimum width=0.9cm},
white/.style={circle,draw=white,fill=white,draw,font=\sffamily,text=white,minimum width=0.001cm},
typeOne node/.style={circle,draw=typeOneColor,fill=typeOneColor,draw,font=\sffamily\Huge\bfseries,text=white,minimum width=0.9cm},
typeTwo node/.style={circle,draw=typeTwoColor,fill=typeTwoColor,draw,font=\sffamily\Huge\bfseries,text=white,minimum width=0.9cm},
]

\node[typeOne node] (1) {}; 
\node[typeTwo node] (2) [right of=1] {}; 
\node[typeTwo node] (3) [above right of=1, left=1.5pt] {}; 
\node[typeTwo node] (4) [above of=3] {}; 

\tikzstyle{LabelStyle}=[below=3pt]
\path[every node/.style={font=\sffamily}] 
	(1) edge [bend left,line width=1.0mm] node[above right] {} (4)
	(2) edge [bend right,line width=1.0mm] node[above right] {} (4)
(1) edge [line width=1.0mm, left] node [above left] {} (2) 
	(1)  edge [line width=1.0mm, right] node[below right] {} (3)
(2) edge [line width=1.0mm, left] node[below left] {} (3)
	(3) edge [line width=1.0mm,right] node[above right] {} (4);
\end{tikzpicture}
}
\vspace{-15mm}
\label{fig:typed-4-cliques-with-2-types}
}
} 

\vspace{-2mm}
\subfigure[\vspace{-2mm}Typed 3-paths with $L=3$ types]{
\scalebox{0.70}{
\scalebox{0.28}{

\begin{tikzpicture}[-,>=latex,auto,node distance=2.0cm,thick,
main node/.style={circle,draw=black,fill=white,draw,font=\sffamily\Huge\bfseries,text=black,minimum width=0.9cm, line width=1mm},
]

\node[main node] (1) {}; 
\node[main node] (2) [below of=1] {}; 
\node[main node] (3) [below of=2] {}; 

\tikzstyle{LabelStyle}=[below=3pt]
\path[every node/.style={font=\sffamily}] 
(1) edge [line width=1.0mm, left] node [above left] {} (2) 
(2) edge [line width=1.0mm, left] node[below left] {} (3);
\end{tikzpicture}
}
\scalebox{0.28}{

\begin{tikzpicture}[-,>=latex,auto,node distance=2.0cm,thick,
main node/.style={circle,draw=black,fill=white,draw,font=\sffamily\Huge\bfseries,text=black,minimum width=0.9cm},
white node/.style={draw=white,draw,font=\sffamily\Huge\bfseries,text=black,minimum width=0.9cm}
]
\node[white node] (1) {};
\node[white node] (2) [below of=1] {};
\node[white node] (3) [below of=2] {};
\node[white node] (4) [right of=2, left=30pt, below=10pt, above=0.05pt] {\vspace{2mm}\fontsize{56}{56}\selectfont $\Rightarrow$};
\end{tikzpicture}
}
\hspace{2mm}
\scalebox{0.28}{

\begin{tikzpicture}[-,>=latex,auto,node distance=2.0cm,thick,
typeOne node/.style={circle,draw=typeOneColor,fill=typeOneColor,draw,font=\sffamily\Huge\bfseries,text=white,minimum width=0.9cm},
typeTwo node/.style={circle,draw=typeTwoColor,fill=typeTwoColor,draw,font=\sffamily\Huge\bfseries,text=white,minimum width=0.9cm},
typeThree node/.style={circle,draw=typeThreeColor,fill=typeThreeColor,draw,font=\sffamily\Huge\bfseries,text=white,minimum width=0.9cm},
]

\node[typeOne node] (1) {}; 
\node[typeOne node] (2) [below of=1] {}; 
\node[typeOne node] (3) [below of=2] {}; 

\tikzstyle{LabelStyle}=[below=3pt]
\path[every node/.style={font=\sffamily}] 
(1) edge [line width=1.0mm, left] node [above left] {} (2) 
(2) edge [line width=1.0mm, left] node[below left] {} (3);
\end{tikzpicture}
}
\label{fig:typed-3-path-homo-typeOne-3colors}
\hspace{0.3mm}
\scalebox{0.28}{
\begin{tikzpicture}[-,>=latex,auto,node distance=2.0cm,thick,
typeOne node/.style={circle,draw=typeOneColor,fill=typeOneColor,draw,font=\sffamily\Huge\bfseries,text=white,minimum width=0.9cm},
typeTwo node/.style={circle,draw=typeTwoColor,fill=typeTwoColor,draw,font=\sffamily\Huge\bfseries,text=white,minimum width=0.9cm},
typeThree node/.style={circle,draw=typeThreeColor,fill=typeThreeColor,draw,font=\sffamily\Huge\bfseries,text=white,minimum width=0.9cm},
]

\node[typeTwo node] (1) {}; 
\node[typeTwo node] (2) [below of=1] {}; 
\node[typeTwo node] (3) [below of=2] {}; 

\tikzstyle{LabelStyle}=[below=3pt]
\path[every node/.style={font=\sffamily}] 
(1) edge [line width=1.0mm, left] node [above left] {} (2) 
(2) edge [line width=1.0mm, left] node[below left] {} (3);
\end{tikzpicture}
}
\hspace{0.3mm}
\scalebox{0.28}{
\begin{tikzpicture}[-,>=latex,auto,node distance=2.0cm,thick,
typeOne node/.style={circle,draw=typeOneColor,fill=typeOneColor,draw,font=\sffamily\Huge\bfseries,text=white,minimum width=0.9cm},
typeTwo node/.style={circle,draw=typeTwoColor,fill=typeTwoColor,draw,font=\sffamily\Huge\bfseries,text=white,minimum width=0.9cm},
typeThree node/.style={circle,draw=typeThreeColor,fill=typeThreeColor,draw,font=\sffamily\Huge\bfseries,text=white,minimum width=0.9cm},
]

\node[typeThree node] (1) {}; 
\node[typeThree node] (2) [below of=1] {}; 
\node[typeThree node] (3) [below of=2] {}; 

\tikzstyle{LabelStyle}=[below=3pt]
\path[every node/.style={font=\sffamily}] 
(1) edge [line width=1.0mm, left] node [above left] {} (2) 
(2) edge [line width=1.0mm, left] node[below left] {} (3);
\end{tikzpicture}
}
\hspace{0.3mm}
\scalebox{0.28}{

\begin{tikzpicture}[-,>=latex,auto,node distance=2.0cm,thick,
typeOne node/.style={circle,draw=typeOneColor,fill=typeOneColor,draw,font=\sffamily\Huge\bfseries,text=white,minimum width=0.9cm},
typeTwo node/.style={circle,draw=typeTwoColor,fill=typeTwoColor,draw,font=\sffamily\Huge\bfseries,text=white,minimum width=0.9cm},
typeThree node/.style={circle,draw=typeThreeColor,fill=typeThreeColor,draw,font=\sffamily\Huge\bfseries,text=white,minimum width=0.9cm},
]

\node[typeOne node] (1) {}; 
\node[typeOne node] (2) [below of=1] {}; 
\node[typeTwo node] (3) [below of=2] {}; 

\tikzstyle{LabelStyle}=[below=3pt]
\path[every node/.style={font=\sffamily}] 
(1) edge [line width=1.0mm, left] node [above left] {} (2) 
(2) edge [line width=1.0mm, left] node[below left] {} (3);
\end{tikzpicture}
}
\hspace{0.3mm}
\scalebox{0.28}{

\begin{tikzpicture}[-,>=latex,auto,node distance=2.0cm,thick,
typeOne node/.style={circle,draw=typeOneColor,fill=typeOneColor,draw,font=\sffamily\Huge\bfseries,text=white,minimum width=0.9cm},
typeTwo node/.style={circle,draw=typeTwoColor,fill=typeTwoColor,draw,font=\sffamily\Huge\bfseries,text=white,minimum width=0.9cm},
typeThree node/.style={circle,draw=typeThreeColor,fill=typeThreeColor,draw,font=\sffamily\Huge\bfseries,text=white,minimum width=0.9cm},
]

\node[typeTwo node] (1) {}; 
\node[typeTwo node] (2) [below of=1] {}; 
\node[typeOne node] (3)  [below of=2] {}; 

\tikzstyle{LabelStyle}=[below=3pt]
\path[every node/.style={font=\sffamily}] 
(1) edge [line width=1.0mm, left] node [above left] {} (2) 
(2) edge [line width=1.0mm, left] node[below left] {} (3);
\end{tikzpicture}
}
\hspace{0.3mm}
\scalebox{0.28}{

\begin{tikzpicture}[-,>=latex,auto,node distance=2.0cm,thick,
typeOne node/.style={circle,draw=typeOneColor,fill=typeOneColor,draw,font=\sffamily\Huge\bfseries,text=white,minimum width=0.9cm},
typeTwo node/.style={circle,draw=typeTwoColor,fill=typeTwoColor,draw,font=\sffamily\Huge\bfseries,text=white,minimum width=0.9cm},
typeThree node/.style={circle,draw=typeThreeColor,fill=typeThreeColor,draw,font=\sffamily\Huge\bfseries,text=white,minimum width=0.9cm},
]

\node[typeOne node] (1) {}; 
\node[typeOne node] (2) [below of=1] {}; 
\node[typeThree node] (3)  [below of=2] {}; 

\tikzstyle{LabelStyle}=[below=3pt]
\path[every node/.style={font=\sffamily}] 
(1) edge [line width=1.0mm, left] node [above left] {} (2) 
(2) edge [line width=1.0mm, left] node[below left] {} (3);
\end{tikzpicture}
}
\hspace{0.3mm}
\scalebox{0.28}{

\begin{tikzpicture}[-,>=latex,auto,node distance=2.0cm,thick,
typeOne node/.style={circle,draw=typeOneColor,fill=typeOneColor,draw,font=\sffamily\Huge\bfseries,text=white,minimum width=0.9cm},
typeTwo node/.style={circle,draw=typeTwoColor,fill=typeTwoColor,draw,font=\sffamily\Huge\bfseries,text=white,minimum width=0.9cm},
typeThree node/.style={circle,draw=typeThreeColor,fill=typeThreeColor,draw,font=\sffamily\Huge\bfseries,text=white,minimum width=0.9cm},
]

\node[typeTwo node] (1) {}; 
\node[typeTwo node] (2) [below of=1] {}; 
\node[typeThree node] (3)  [below of=2] {}; 

\tikzstyle{LabelStyle}=[below=3pt]
\path[every node/.style={font=\sffamily}] 
(1) edge [line width=1.0mm, left] node [above left] {} (2) 
(2) edge [line width=1.0mm, left] node[below left] {} (3);
\end{tikzpicture}
}
\hspace{0.3mm}
\scalebox{0.28}{

\begin{tikzpicture}[-,>=latex,auto,node distance=2.0cm,thick,
typeOne node/.style={circle,draw=typeOneColor,fill=typeOneColor,draw,font=\sffamily\Huge\bfseries,text=white,minimum width=0.9cm},
typeTwo node/.style={circle,draw=typeTwoColor,fill=typeTwoColor,draw,font=\sffamily\Huge\bfseries,text=white,minimum width=0.9cm},
typeThree node/.style={circle,draw=typeThreeColor,fill=typeThreeColor,draw,font=\sffamily\Huge\bfseries,text=white,minimum width=0.9cm},
]

\node[typeThree node] (1) {}; 
\node[typeThree node] (2) [below of=1] {}; 
\node[typeOne node] (3)  [below of=2] {}; 

\tikzstyle{LabelStyle}=[below=3pt]
\path[every node/.style={font=\sffamily}] 
(1) edge [line width=1.0mm, left] node [above left] {} (2) 
(2) edge [line width=1.0mm, left] node[below left] {} (3);
\end{tikzpicture}
}
\hspace{0.3mm}
\scalebox{0.28}{

\begin{tikzpicture}[-,>=latex,auto,node distance=2.0cm,thick,
typeOne node/.style={circle,draw=typeOneColor,fill=typeOneColor,draw,font=\sffamily\Huge\bfseries,text=white,minimum width=0.9cm},
typeTwo node/.style={circle,draw=typeTwoColor,fill=typeTwoColor,draw,font=\sffamily\Huge\bfseries,text=white,minimum width=0.9cm},
typeThree node/.style={circle,draw=typeThreeColor,fill=typeThreeColor,draw,font=\sffamily\Huge\bfseries,text=white,minimum width=0.9cm},
]

\node[typeThree node] (1) {}; 
\node[typeThree node] (2) [below of=1] {}; 
\node[typeTwo node] (3)  [below of=2] {}; 

\tikzstyle{LabelStyle}=[below=3pt]
\path[every node/.style={font=\sffamily}] 
(1) edge [line width=1.0mm, left] node [above left] {} (2) 
(2) edge [line width=1.0mm, left] node[below left] {} (3);
\end{tikzpicture}
}
\hspace{0.3mm}
\scalebox{0.28}{

\begin{tikzpicture}[-,>=latex,auto,node distance=2.0cm,thick,
typeOne node/.style={circle,draw=typeOneColor,fill=typeOneColor,draw,font=\sffamily\Huge\bfseries,text=white,minimum width=0.9cm},
typeTwo node/.style={circle,draw=typeTwoColor,fill=typeTwoColor,draw,font=\sffamily\Huge\bfseries,text=white,minimum width=0.9cm},
typeThree node/.style={circle,draw=typeThreeColor,fill=typeThreeColor,draw,font=\sffamily\Huge\bfseries,text=white,minimum width=0.9cm},
]

\node[typeOne node] (1) {}; 
\node[typeTwo node] (2) [below of=1] {}; 
\node[typeThree node] (3)  [below of=2] {}; 

\tikzstyle{LabelStyle}=[below=3pt]
\path[every node/.style={font=\sffamily}] 
(1) edge [line width=1.0mm, left] node [above left] {} (2) 
(2) edge [line width=1.0mm, left] node[below left] {} (3);
\end{tikzpicture}
}
\label{fig:typed-3-path-with-3-types}
}
} 

\vspace{-4mm}
\caption{
Examples of heterogeneous network motifs
}
\label{fig:typed-motifs-3colors}
\vspace{-4mm}
\end{figure}

However, such (untyped) graphlets are \emph{unable} to capture the rich (typed) connectivity patterns in more complex networks such as those that are heterogeneous (which includes signed, labeled, bipartite, k-partite, and attributed graphs, among others).
In heterogeneous graphs, nodes and edges can be of different types and explicitly modeling such types is crucial.
Such heterogeneous graphs (networks) arise ubiquitously in the natural world where nodes and edges of multiple types are observed, \eg, 
between humans~\cite{kong2013inferring-heterSoc}, 
neurons~\cite{bullmore2009complex,bassett2006small}, 
routers and autonomous systems (ASes)~\cite{rossi2013topology}, 
web pages~\cite{yin2009exploring}, 
devices \& sensors~\cite{eagle2006reality}, 
infrastructure (roads, airports, power stations)~\cite{powergrid}, 
economies~\cite{schweitzer2009economic}, 
vehicles (cars, satelites, UAVs)~\cite{hung2008mobility-heter}, 
and information in general~\cite{yu2014personalized,sun2011pathsim,rossi16collective-factor}.

In this work, we generalize the notion of network motifs to heterogeneous networks.
In particular, we introduce the notion of a small induced typed subgraph called \emph{typed graphlet}.\footnote{The terms typed, colored, labeled, and heterogeneous graphlet (network motif) are used interchangeably.}
Typed graphlets are a powerful generalization of the notion of graphlet as they capture both the induced subgraph of interest and the types associated to the nodes in the induced subgraph (Figure~\ref{fig:typed-motifs-3colors}).
These small induced typed subgraphs are the fundamental \emph{building blocks of rich heterogeneous networks}.
Typed graphlets naturally capture the higher-order typed connectivity patterns in bipartite, k-partite, signed, labeled, k-star, attributed graphs, and more generally heterogeneous networks (Figure~\ref{fig:heter-special-cases}).
As such, typed graphlets are useful for a wide variety of predictive and descriptive modeling applications in these rich complex networks.

Despite their fundamental and practical importance, counting typed graphlets in large graphs remains a challenging and unsolved problem.
To address this problem, we propose a fast, parallel, and space-efficient framework for counting typed graphlets in large networks.
The time complexity is provably optimal as it matches the best untyped graphlet counting algorithm.
Using non-trivial combinatorial relationships between lower-order ($k\!-\!1$)-node typed graphlets, we derive equations that allow us to compute many of the $k$-node typed graphlet counts in $o(1)$ constant time.
Thus, we avoid explicit enumeration of many typed graphlets by simply computing the count exactly using the discovered combinatorial relationships.
For every edge, we count a few typed graphlets and obtain the exact counts of the remaining typed graphlets in $o(1)$ constant time.
Furthermore, we store only the nonzero typed graphlet counts for every edge.
To better handle large-scale heterogeneous networks with an arbitrary number of types, we propose an efficient parallel algorithm for typed graphlet counting that scales almost linearly as the number of processing units increase.
While we primarily focus on the more challenging problem of local typed graphlet counting, we also demonstrate how to solve the global typed graphlet counting problem.
As an aside, the proposed approach is the first such method that counts colored graphlets on the edges (or more generally, between any arbitrary pair of nodes).
Nevertheless, the proposed framework overcomes the computational issues with previous work that severely limited its use in practice, \eg, all existing methods are only able to compute typed graphlets on extremely small and unrealistic graphs with hundreds or thousands of nodes/edges.

Theoretically, we show that typed graphlets are more powerful and encode more information than untyped graphlets.
In addition, we theoretically demonstrate the worst-case time and space complexity of the proposed framework.
It is notably much faster and significantly more space-efficient than the existing state-of-the-art algorithms.
In fact, the time complexity of the proposed approach is shown to be equivalent to the best untyped graphlet counting algorithm.
Furthermore, unlike previous work, we derive many of the typed graphlets directly in $o(1)$ constant time using counts of lower-order $(k\!-\!1)$-node typed graphlets.

Empirically, the proposed approach is shown to be orders of magnitude faster than existing state-of-the-art methods. 
In particular, we observe between 89 and 10,981 times speedup in runtime performance compared to the best existing method.
Notably, on graphs of even moderate size (thousands of nodes/edges), existing approaches fail to finish in a reasonable amount of time (24 hours).
In terms of space, the proposed approach uses between 42x and 776x less space than existing methods.
We also demonstrate the parallel scaling of the parallel algorithm and observe nearly linear speedups as the number of processing units increases.
In addition to real-world graphs from a wide range of domains, we also show results on a number of synthetically generated graphs from a variety of synthetic graph models.
Finally, we demonstrate the utility of typed graphlets for exploratory network analysis and understanding using a variety of well-known networks 
and make a number of new and important observations/findings.

Compared to the homogeneous/untyped graphlet counting problem (which has found many important applications~\cite{koyuturk2006pairwise,prvzulj2007biological,vishwanathan2010graph,shervashidze2009efficient,solava2012graphlet,benson2016higher,ahmed17aaai,role2vec,noble2003graph,akoglu2015graph,Rossi2018a}), 
heterogeneous graphlets are vastly more powerful containing a significant amount of additional information. 
We show this formally using information theory and demonstrate the importance of heterogeneous graphlets empirically using real-world graphs.
In real-world graphs, we observe that only a handful of the possible heterogeneous graphlets actually occur.
Furthermore, we discover that 99\% of the overall counts in a number of real-world networks are from only a few such heterogeneous graphlets.
This observation indicates a power-law relationship between the counts of the various heterogeneous graphlets.
The rare heterogeneous graphlets (\ie, those that rarely occur) also contain a lot of useful information and these forbidden heterogeneous motifs may indicate anomalies/outliers or simply unique structural behaviors that are fundamentally important but extremely difficult to identify using traditional methods.
Moreover, the heterogeneous graphlets found to be important are easily interpretable and provide key insights into the structure and underlying phenomena governing the formation of the complex network that would otherwise be hidden using traditional methods.

\begin{figure}[t!]
\vspace{4mm}
\begin{minipage}{0.54\linewidth}
\centering
\hspace{-2mm}
\includegraphics[width=0.90\linewidth]{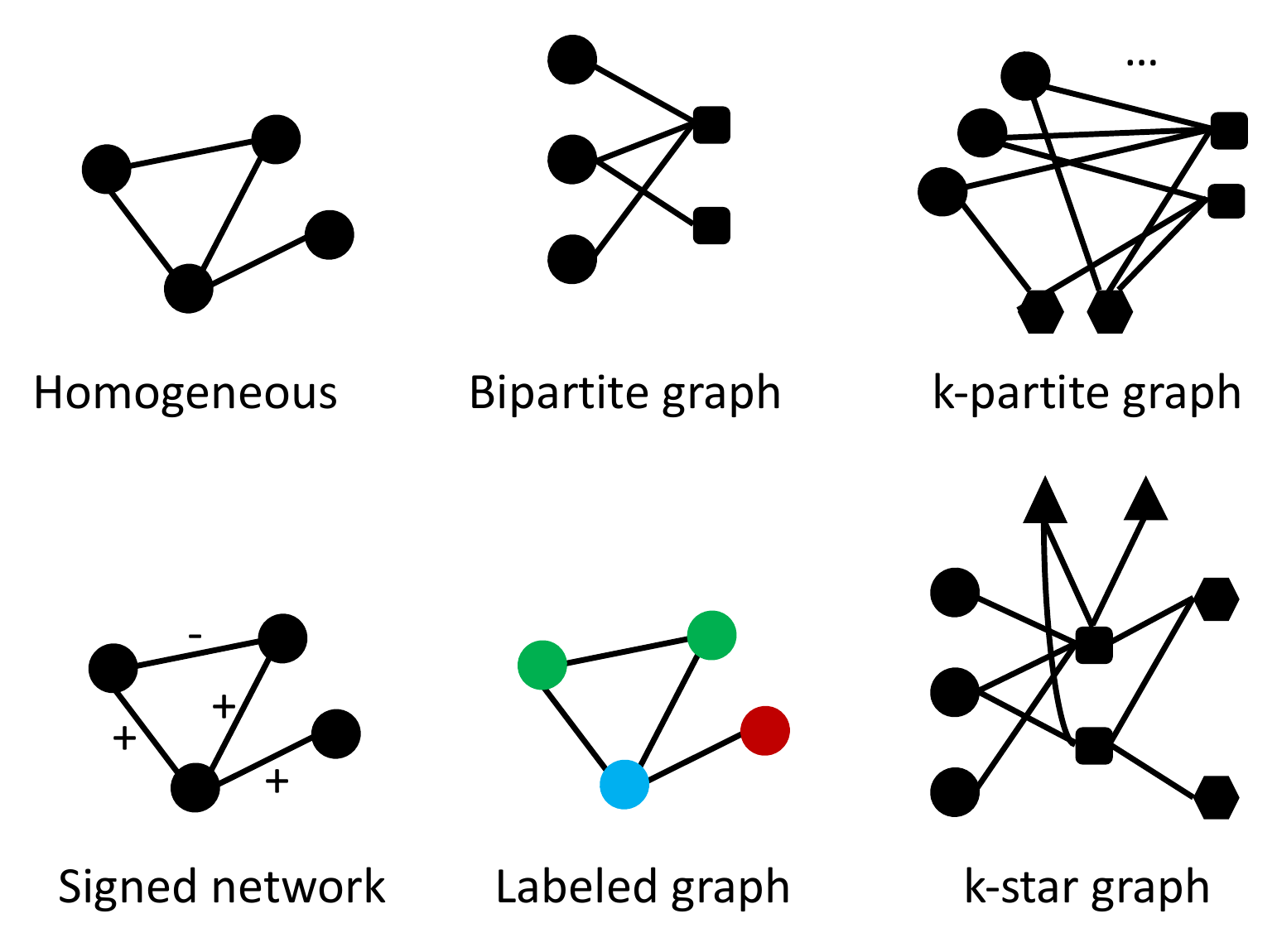} 
\end{minipage}
\hfill
\begin{minipage}{0.45\linewidth}
\centering
\scalebox{0.80}{
\small
\fontsize{8.5}{9}\selectfont
\begin{tabularx}{1.05\linewidth}{rH cc HH}
\toprule
\textbf{Graph Type} &&  $|\mathcal{T}_V|$ & $|\mathcal{T}_E|$ \\
\midrule
\textsc{Homogeneous} && $1$ & $1$ \\
\textsc{Bipartite} && $2$ & $1$ \\
\textsc{K-partite} && $k$ & $k-1$ \\
\textsc{Signed} && $1$ & $2$ \\
\textsc{Labeled} && $k$ & $\ell$ \\
\textsc{Star} && $k$ & $k-1$ \\
\bottomrule
\end{tabularx}
}
\end{minipage}
\caption{\emph{Heterogeneous network motifs} are useful for a wide variety of graphs.
These are only a few examples of the different graphs that are naturally supported by the proposed framework.
}
\label{fig:heter-special-cases}
\end{figure}

This work generalizes the notion of network motif to heterogeneous networks and 
describes a framework for finding all such \emph{heterogeneous network motifs}.
The proposed framework has the following desired properties:

\smallskip
\begin{compactenum}[$\bullet$]
\item \textbf{Fast}: 
The approach is fast for large graphs by leveraging \emph{non-trivial combinatorial relationships} to derive many of the typed graphlets in $o(1)$ constant time. 
Thus, the worst-case time complexity matches that of the best untyped graphlet algorithm.
As shown in Table~\ref{table:runtime-perf}, the approach is orders of magnitude faster than the state-of-the-art. 

\item \textbf{Space-Efficient}: The approach is space-efficient by hashing and storing only the typed motif counts that appear on a given edge.

\item \textbf{Parallel}: 
The typed graphlet counting approach lends itself to an efficient and scalable parallel implementation.
We observe near-linear parallel scaling results in Section~\ref{sec:exp-parallel-scaling}. 
This indicates the feasibility of the approach for large real-world networks.

\item \textbf{Scalable for Large Networks}:
The proposed approach is scalable for large heterogeneous networks.
In particular, the approach scales nearly linearly as the size of the graph increases.

\item \textbf{Effectiveness}: Typed graphlet counting is shown to be effective for understanding and mining large heterogeneous networks from a variety of different application domains.
\end{compactenum}
\smallskip

\begin{table}[t!]
\caption{Summary of notation. Matrices are bold upright roman letters; vectors are bold lowercase letters.}
\vspace{-3mm}
\renewcommand{\arraystretch}{1.25} 
\scalebox{0.90}{
\centering 
\fontsize{8}{8.5}\selectfont
\setlength{\tabcolsep}{6pt} 
\label{table:notation}
\hspace*{-2.5mm}
\begin{tabularx}{1.10\linewidth}{@{}r X@{}} 
\toprule
$G$ & graph \\ 
$V(G)$ & node set of $G$ \\
$E(G)$ & edge set of $G$ \\
$H,F$ & graphlet of $G$ \\
$I_G(H)$ & set of unique instances of $H$ in $G$ \\

$N, M$ & number of nodes $N = |V|$ and edges $M = |E|$ in the graph \\
$K$ & size of a network motif ($\#$ nodes) \\
$L$ & number of types (\ie, colors, labels) \\

$\mathcal{H}$ & set of all untyped motifs in $G$ \\
$\mathcal{H}_T$ & set of all typed motifs in $G$ \\

$T$ & \# of different typed motifs $T = |\mathcal{H}_T|$ observed in $G$ with $L$ types \\ 
$T_{\max}$ & total \# of possible typed motifs with $L$ types, hence $T \leq T_{\max}$ \\ 
$T_H$ & \# of different typed motifs for a particular motif $H \in \mathcal{H}$ \\

$\mathcal{T}_V$ & set of node types in $G$ \\
$\mathcal{T}_E$ & set of edge types in $G$ \\
$\phi$ & type function $\phi : V \rightarrow \mathcal{T}_V$ \\
$\xi$ & type function $\xi : E \rightarrow \mathcal{T}_E$ \\

$\vt$ & $K$-dimensional type vector $\vt = \big[\,\! \phi_{w_1} \;\! \cdots \; \phi_{w_K} \,\!\big]$ \\

$f_{ij}(H, \vt)$ & \# of instances of motif $H$ that contain nodes $i$ and $j$ with type vector $\vt$ \\

$\mathbb{F}$ & an arbitrary typed motif hash function (Section~\ref{sec:typed-motif-hash-function}) \\
$\hash$ & hash value from $\mathbb{F}$ representing a unique id for an arbitrary typed motif \\

$\Delta$ & maximum degree of a node in $G$ \\
$\Gamma_{i}^{t}$ & set of neighbors of node $i$ with type $t$ \\
$d_{i}^{t}$ & degree of node $i$ with type $t$, $d_{i}^{t} = |\Gamma_{i}^{t}|$ \\
$T_{ij}^{t}$ & set of nodes of type $t$ that form typed triangles with $i$ and $j$ \\
$S_{i}^{t}$, $S_{j}^{t}$ & set of nodes of type $t$ that form typed 3-node stars centered at $i$ (or $j$) \\
$I$ & set of non-adjacent nodes with respect to a pair of nodes $(i,j)$ \\

$\mathcal{M}$ & set of typed motif ids that appear in $G$ \\
$\mathcal{M}_{ij}$ & set of typed motif ids for a given pair of nodes $(i,j)$ \\
$\mathcal{X}_{ij}$ & nonzero (typed-motif, count) pairs for edge $(i,j) \in E$ \\

$\Psi$ & hash table for checking whether a node is connected to $i$ or $j$ and its ``relationship'' (\eg, $\lambda_{1}$, $\lambda_{2}$, $\lambda_{3}$) in constant time \\

\bottomrule
\end{tabularx}}
\end{table}

\section{Heterogeneous Network Motifs}
\label{sec:typed-network-motifs}
This section introduces a generalization of graphlets (network motifs) called \emph{heterogeneous network motifs} (or simply \emph{typed graphlets}).
In particular, Section~\ref{sec:heter-graph-model} describes the heterogeneous graph model and highlights a number of important properties of it.
Next, Section~\ref{sec:typed-graphlet-formulation} formally introduces the notion of \emph{typed graphlet} proposed in this work and discusses its implications.
Finally, Section~\ref{sec:special-cases-heter-generalization} discusses a few important special cases of heterogeneous graphs that are naturally supported by the proposed framework in Section~\ref{sec:framework}.

\subsection{Heterogeneous Graph Model} \label{sec:heter-graph-model}
In this section, we introduce a general heterogeneous graph model and discuss different problem settings and graph types where typed graphlets are likely to be useful (Figure~\ref{fig:heter-special-cases}). 
We use the following heterogeneous graph formulation:
\begin{Definition}[Heterogeneous network] \label{def:heter-network}
A heterogeneous network is defined as $G=(V,E)$ consisting of a set of node objects $V$ and a set of edges $E$ connecting the nodes in $V$.
A heterogeneous network also has a \emph{node type mapping function} $\,\phi : V \rightarrow \mathcal{T}_V$ 
and an \emph{edge type mapping function} defined as $\,\xi : E \rightarrow \mathcal{T}_E$ 
where $\mathcal{T}_V$ and $\mathcal{T}_E$ denote the set of node object types and edge types, respectively.
The type of node $i$ is denoted as $\phi_i$ 
whereas the type of edge $e = (i,j) \in E$ is denoted as $\xi_{ij}=\xi_e$.
\end{Definition}\noindent
A homogeneous graph is a special case of a heterogeneous graph where $|\mathcal{T}_V|=|\mathcal{T}_E|=1$. 
Other special cases of the proposed framework include 
bipartite graphs, 
signed networks with $\{+,-\}$ types,
and more generally labeled graphs.
As such, in general, a heterogeneous network can be represented as a series of matrices and tensors that are coupled (\ie, the tensors and matrices all share at least one mode with each other)~\cite{acar2011all,rossi16collective-factor}.
A few other special cases are discussed further in Section~\ref{sec:special-cases-heter-generalization}.
See Table~\ref{table:notation} for a summary of key notation.

\subsection{Typed Graphlet Generalization} \label{sec:typed-graphlet-formulation}
In this section, we introduce a more general notion of graphlet called \emph{typed graphlet} that naturally extends to both homogeneous and general heterogeneous networks.
We use $G$ to represent a graph and $H$ or $F$ to represent graphlets.

\subsubsection{Untyped Graphlets}
We begin by defining untyped graphlets for graphs with a single type.

\begin{Definition}[\sc Untyped Graphlet] \label{def:graphlet}
An untyped graphlet $H$ is a connected induced subgraph of $G$.
\end{Definition} \noindent

Given a graphlet in some graph, it may be the case that we can find other topologically identical ``appearances" of this structure in that graph. 
We call these ``appearances" \textit{graphlet instances}.

\begin{Definition}[\scshape Untyped Graphlet Instance]\label{def:graphlet-instance}
An instance of an untyped graphlet $H$ in graph $G$ is an untyped graphlet $F$ in $G$ that is isomorphic to $H$.
\end{Definition}

\subsubsection{Typed Graphlets}
In heterogeneous graphs, nodes/edges can be of many different types and so explicitly and jointly modeling such types is essential (Figure~\ref{fig:heter-special-cases}).
In this work, we introduce the notion of a \emph{typed graphlet} that explicitly captures both the connectivity pattern of interest and the types.
Notice that typed graphlets are a generalization of graphlets to heterogeneous networks.

\begin{Definition}[\scshape Typed Graphlet]\label{def:typed-graphlet}
A typed graphlet of a graph $G=(V,E,\phi,\xi)$ is a connected induced heterogeneous subgraph $H=(V',E',\phi',\xi')$ of $G$ such that
\begin{compactenum}
\item $(V',E')$ is a graphlet of $(V,E)$,
\item $\phi'=\phi|_{V'}$, that is, $\phi'$ is the restriction of $\phi$ to $V'$,
\item $\xi'=\xi|_{E'}$, that is, $\xi'$ is the restriction of $\xi$ to $E'$.
\end{compactenum}
\end{Definition}\noindent
The terms typed graphlet, colored graphlet, and heterogeneous network motif (graphlet) are used interchangeably.

We can consider the presence of topologically identical ``appearances" of a typed graphlet in a graph.
\begin{Definition}[\scshape Typed Graphlet Instance]\label{def:typed-graphlet-instance}
An instance of a typed graphlet $H=(V',E',\phi',\xi')$ of graph $G$ is a typed graphlet $F=(V'',E'',\phi'',\xi'')$ of $G$ such that
\begin{compactenum}
\item $(V'',E'')$ is isomorphic to $(V',E')$, 
\item $\mathcal{T}_{V''}=\mathcal{T}_{V'}$ and $\mathcal{T}_{E''}=\mathcal{T}_{E'}$, that is, the multisets of node and edge types are correspondingly equal.
\end{compactenum}
The set of unique typed graphlet instances of $H$ in $G$ is denoted as $I_G(H)$.
\end{Definition}\noindent
Comparing the above definitions of graphlet and typed graphlet, we see at first glance that typed graphlets are nontrivial extensions of their homogeneous counterparts.
The ``position'' of an edge (node) in a typed graphlet is often topologically important, \eg, an edge at the end of the 4-path (Figure~\ref{fig:4-path-edge-orbit-Si}) vs. an edge at the center of a 4-path (Figure~\ref{fig:4-path-center-orbit}).
These topological differences of a typed graphlet are called (automorphism) \emph{typed orbits} since they take into account ``symmetries'' between edges (nodes) of a graphlet.
Typed graphlet orbits are a generalization of (homogeneous) graphlet orbits~\cite{prvzulj2007biological}.

\subsection{Number of Typed Graphlets}
\noindent
For a single $K$-node untyped motif (\eg, $K$-clique), the number of \emph{typed motifs} with $L$ types is:
\begin{equation} \label{eq:num-typed-graphlets-for-motif}
\left( \binom{L}{K} \right) = \binom{L+K-1}{K}
\end{equation}\noindent
where $L=$ number of types (colors) and $K=$ size of the network motif ($\#$ of nodes).
Table~\ref{table:typed-graphlets-example} shows the number of \emph{typed network motifs} that arise from a single motif $H \in \mathcal{H}$ of size $K\in \{2,\ldots,4\}$ nodes as the number of types varies from $L=1,2,\ldots,9$.
Notice that Table~\ref{table:typed-graphlets-example} is for a single $K$-node motif $H \in \mathcal{H}$ and therefore the total number of typed network motifs for all $K$-node motifs is a multiple of the amounts given in Table~\ref{table:typed-graphlets-example}.
For instance, the total number of typed motif orbits with 4 nodes that arise from 7 types is $10 \cdot 210 = 2100$ since there are 10 connected 4-node homogeneous motif orbits.
See Figure~\ref{fig:typed-motifs-3colors} for other examples.
Unlike homogeneous motifs, it is obviously impossible to show all the heterogeneous motifs counted by the proposed approach since it works for general heterogeneous graphs with any arbitrary number of types $L$ and structure.

\begin{table}[h!]
\renewcommand{\arraystretch}{0.95} 
\caption{Number of \emph{typed graphlets} (for a single untyped graphlet) as the size of the graphlet ($K$) and types ($L$) varies.}
\label{table:typed-graphlets-example}
\vspace{-2mm}
\scalebox{0.90}{
\begin{tabularx}{1.0\linewidth}{HX XXXXX XXX X}
\toprule
&& \multicolumn{9}{c}{\bf Types $L$} \\
\cmidrule(l{3pt}r{7pt}){3-11}

&& \textbf{1} & \textbf{2} & \textbf{3} & \textbf{4}  & \textbf{5} & \textbf{6}  & \textbf{7} & \textbf{8} & \textbf{9}  \\
\midrule
& \textbf{K=2} & 1  & 3  & 6  & 10  & 15  & 21  & 28 & 36  & 45 \\
\multirow{2}{*}{\bf K} &
\textbf{K=3} & 1  & 4  & 10  & 20  & 35  & 56  & 84   & 120 & 165 \\
& \textbf{K=4} & 1  & 5  & 15 & 35  & 70  & 126  & 210   & 330 & 495  \\
\bottomrule
\end{tabularx}
}
\end{table}

\subsection{Generalization to Other Graphs} \label{sec:special-cases-heter-generalization}
The proposed notion of \emph{heterogeneous network motifs} can be used for applications on bipartite, k-partite, signed, labeled, attributed, and more generally heterogeneous networks.
A few examples of such graphs are shown in Figure~\ref{fig:heter-special-cases}.
The proposed framework naturally handles general heterogeneous graphs with arbitrary structure and an arbitrary number of types.
It is straightforward to see that homogeneous, bipartite, k-partite, signed, labeled, and star graphs are all special cases of heterogeneous graphs.
Therefore, the framework for deriving heterogeneous motifs can easily support such networks.
See Figure~\ref{fig:heter-special-cases} for a few of the popular special cases of heterogeneous graphs that are naturally supported in the proposed framework.
In the case of attributed graphs with $D>1$ attributes/features, the attributes of a node or edge can be mapped to types using any arbitrary approach such as~\cite{role2vec} or~\cite{wl}.
For instance, given an attribute vector $\vx \in \RR^{D}$ (which may contain intrinsic attributes or structural features), we derive a type $y$ via $f : \vx \rightarrow y$ where $f$ is a function that maps $\vx$ to a single type $y$.
Notice that attributed graphs are a generalization of labeled graphs.
Another simple approach is to count typed motifs for each of the $D$ attributes.

\section{Framework} \label{sec:framework}
This section describes the general framework for counting typed graphlets.
The typed graphlet framework can be used for counting typed graphlets locally for every edge in $G$ as well as the global typed graphlet counting problem (Definition~\ref{def:global-typed-graphlet-counting}) that focuses on computing the total frequency of a subset or all typed graphlets up to a certain size. 
This paper mainly focuses on the harder local typed graphlet counting problem defined as:
\begin{Definition}[Local Typed Graphlet Counting] \label{def:local-typed-graphlet-counting}
Given a graph $G$ and an edge $(i, j) \in E$, 
the local typed graphlet counting problem is to find the set of all typed graphlets that contain nodes $i$ and $j$ and their corresponding frequencies.
This work focuses on computing typed graphlet counts for every edge in $G$.
\end{Definition}
This problem is obviously more difficult than the global typed graphlet counting problem.
Nevertheless, we discuss the global typed graphlet counting problem in Section~\ref{sec:global-typed-graphlet-counting} and mention a number of improvements and optimization's that can be leveraged if we are only concerned with the total frequency of the different typed graphlets in $G$.

Counting \emph{typed k-node motifs} consists of two parts.
The first is the actual structure of the typed k-node motif and the second part is the type configuration of the motif.
In this work, we propose combinatorial relationships for typed graphlets and use them to derive the counts of many typed graphlet directly in $o(1)$ constant time.
Thus, the proposed approach avoids explicit enumeration of the nodes involved in those typed graphlets.
Algorithm~\ref{alg:typed-motifs-exact} shows the general approach for counting typed network motifs.
Note that we do not make any restriction or assumption on the number of node or edge types. 
The algorithm naturally handles heterogeneous graphs with arbitrary number of types and structure.
See Table~\ref{table:notation} for a summary of notation.

{
\algblockdefx[parallel]{parfor}{endpar}[1][]{$\textbf{parallel for}$ #1 $\textbf{do}$}{$\textbf{end parallel}$}
\algrenewcommand{\alglinenumber}[1]{\fontsize{6.5}{7}\selectfont#1}
\begin{figure}[h!]
\begin{center}
\begin{algorithm}[H]
\caption{\,Heterogeneous Network Motifs}
\label{alg:typed-motifs-exact}
\begin{spacing}{1.1}
\fontsize{8}{9}\selectfont
\begin{algorithmic}[1]
\Require a graph $G=(V,E,\Phi,\xi)$ 
\smallskip

\State Initialize $\Psi$ to all zeros
\label{algline:init-arrays}

\parfor[{\bf each} $(i,j) \in E$ in order] \label{algline:graphlet-for}
	
	\State Set $\vx = \mathbf{0}$ to a vector of all zeros \emph{and} $\mathcal{X}_{ij} = \emptyset$
	
	\State $\mathcal{M}_{ij} = \emptyset$ \Comment{set of typed motif ids occurring between $i$ and $j$}
	
\smallskip
	\For {$k \in \Gamma_i$} 
\label{algline:graphlet-create-hash}
	    \textbf{if} $k \not= j$ \textbf{then} $\Psi(k) = \lambda_1$ \label{algline:graphlet-mark-neigh-of-v}
	\EndFor
	
	\For {$k \in \Gamma_j$} \label{algline:triangles-and-wedges}
		\If{$k = i$} \textbf{continue} \EndIf \label{algline:graphlet-tri-w=v} 
		\If{$\Psi(k) = \lambda_1$} \label{algline:triangle} \Comment{triangle motif}
			\State $T_{ij} \leftarrow T_{ij} \cup \{k\}$, $|T_{ij}^{\phi_k}\!|=|T_{ij}^{\phi_k}\!|+1$, and set $\Psi(k)=\lambda_3$ \label{algline:triangle-found}  
			
\State $\langle \vx, \mathcal{M}_{ij} \rangle = \textsc{Update}(\vx,\mathcal{M}_{ij},\mathbb{F}(g_2, \Phi_i, \Phi_j, \Phi_{k}, 0))$  \label{algline:typed-tri-motif-hash}

		\Else \Comment{\emph{typed 3-path} centered at node $j$} \label{algline:otherwise-3-path-centered-at-j}
			\State $S_{j} \leftarrow S_{j} \cup \{k\}$,  $|S_{j}^{\phi_k}\!|=|S_{j}^{\phi_k}\!|+1$, and set $\Psi(k)=\lambda_2$ \label{algline:3-path-centered-at-j} 
		    
		    \State $\langle \vx, \mathcal{M}_{ij} \rangle = \textsc{Update}(\vx,\mathcal{M}_{ij},\mathbb{F}(g_1, \Phi_i, \Phi_j, \Phi_{k}, 0))$ \label{algline:typed-3-path-centered-j-motif-hash}
		    
		\EndIf
	\EndFor
	
	\For {$k \in (\Gamma_i \setminus T_{ij})$} \label{algline:S_i} \Comment{Set of 3-path nodes centered at $i$}
			\State $S_{i} \leftarrow S_{i} \cup \{k\}$ and $|S_{i}^{\phi_k}\!|=|S_{i}^{\phi_k}\!|+1$ \label{algline:3-path-centered-at-i} 

		    \State $\langle \vx, \mathcal{M}_{ij} \rangle = \textsc{Update}(\vx,\mathcal{M}_{ij},\mathbb{F}(g_1, \Phi_i, \Phi_j, \Phi_{k}, 0))$ \label{algline:typed-3-path-centered-i-motif-hash}
		    
	\EndFor
	
	\State Given $S_i$ and $S_j$, derive typed path-based motifs via Algorithm~\ref{alg:typed-path-based-motifs-exact} \label{algline:main-alg-call-typed-path-based-motifs-exact}
	\smallskip
	
	\State Given $T_{ij}$, derive typed triangle-based motifs via Algorithm~\ref{alg:typed-triangle-based-motifs-exact} \label{algline:main-alg-call-typed-triangle-based-motifs-exact}
\smallskip

	\For{$t, t^{\prime} \in \{1,\ldots, L\}$ such that $t \leq t^{\prime}$} \label{algline:main-alg-for-type-pair}
		\multilinenospaceD{Derive remaining typed graphlet orbits in constant time via Eq.~\ref{eq:typed-4-path-center-orbit}-\ref{eq:typed-4-chordal-cycle-center-orbit} and update counts $\vx$ and set of motifs $\mathcal{M}_{ij}$ if needed}
		\label{algline:main-alg-derive-remaining-typed-graphlet-orbits-constant-time}
	\EndFor
	
\For{$\hash \in \mathcal{M}_{ij}$} \Comment{unique typed graphlets between node $i$ and $j$} \label{algline:for-each-unique-motif-id-for-given-edge}
		\State $\mathcal{X}_{ij} = \mathcal{X}_{ij} \cup \{(\hash,\vx_{\hash})\}$ \Comment{store nonzero typed graphlet counts} 
			\label{algline:add-sparse-typed-motif-id-and-count-pairs}
\EndFor

	\State Reset $\Psi$ to all zero  \label{algline:unmark-graphlet-create-hash}
	
\endpar
\State Merge all local typed motifs found by each worker to obtain $\mathcal{M}$ 
\State {\bf return} $\mathcal{X}$ \emph{and} the set of motifs $\mathcal{M}$ occurring in $G$
\smallskip
\end{algorithmic}
\end{spacing}
\vspace{-0.mm}
\end{algorithm}
\end{center}
\vspace{-4.mm}
\end{figure}
}

\subsection{Counting Typed 3-Node Motifs} \label{sec:algorithm-3-node-typed-motifs}
We begin by introducing the notion of a typed neighborhood \emph{and} typed degree of a node.
These are then used as a basis for deriving all typed 3-node motif counts in worst-case $\mathcal{O}(\Delta)$ time (Theorem~\ref{lem:time-complexity-3-node-graphlets}).
\begin{Definition}[Typed Neighborhood]\label{def:typed-neighborhood}
Given an arbitrary node $i$ in $G$, the \emph{typed neighborhood} $\Gamma_{i}^{t}$ is the set of nodes with type $t$ that are reachable by following edges originating from $i$ within $1$-hop distance.
More formally, 
\begin{equation}\label{eq:typed-neighborhood}
\Gamma_{i}^{t} = \{ j \in V \, | \, (i,j) \in E \wedge \phi_j=t \}
\end{equation}\noindent
Intuitively, a node $j \in \Gamma_{i}^{t}$ iff there exists an edge $(i,j) \in E$ between node $i$ and $j$ and the type of node $j$ denoted as $\phi_j$ is $t$.
\end{Definition}\noindent
\begin{Definition}[Typed Degree] \label{def:typed-degree}
The \emph{typed-degree} $d_{i}^{t}$ of node $i$ with type $t$ is defined as $d_{i}^{t} = |\Gamma_{i}^{t}|$ where $d_{i}^{t}$ is the number of nodes connected to node $i$ with type $t$.
\end{Definition}

Using these notions as a basis, we can define $S_{i}^{t}$, $S_{j}^{t}$, and $T_{ij}^{t}$ for $t=1,\ldots,L$ (Figure~\ref{fig:typed-lower-order-sets}).
Obtaining these sets is equivalent to computing all $3$-node typed motif counts.
These sets are all defined with respect to a given edge $(i,j) \in E$ between node $i$ and $j$ with types $\phi_i$ and $\phi_j$.
Since typed graphlets are counted for each edge $(i,j) \in E$, the types $\phi_i$ and $\phi_j$ are fixed ahead of time. 
Thus, there is only one remaining type to select for $3$-node typed motifs.
\begin{Definition}[Typed Triangle Nodes] \label{def:typed-triangle-node}
Given an edge $(i,j) \in E$ between node $i$ and $j$ with types $\phi_i$ and $\phi_j$, 
let $T_{ij}^{t}$ denote the set of nodes of type $t$ that complete a typed triangle with node $i$ and $j$ defined as:
\begin{align} \label{eq:typed-triangle}
T_{ij}^{t} = \Gamma_{i}^{t} \cap \Gamma_j^{t}
\end{align}\noindent
where $|T_{ij}^{t}|$ denotes the number of nodes that form triangles with node $i$ and $j$ of type $t$.
Furthermore, every node $k \in T_{ij}^{t}$ is of type $t$ and thus completes a typed triangle with node $i$ and $j$ consisting of types $\phi_i$, $\phi_j$, and $\phi_k=t$.
\end{Definition}
\begin{Definition}[Typed 3-Star Nodes Centered at $i$] \label{def:typed-3-star-node-centered-at-i}
Given an edge $(i,j) \in E$ between node $i$ and $j$ with types $\phi_i$ and $\phi_j$.
Let $S_{i}^{t}$ denote the set of nodes of type $t$ that form 3-node stars (or equivalently 3-node paths) centered at node $i$ (and not including $j$).
More formally, 
\begin{align} \label{eq:3-node-typed-star-centered-at-node-i}
S_{i}^{t} &= \big\lbrace k \in (\Gamma_{i}^{t} \setminus \{j\}) \; \big| \; k \notin \Gamma_{j}^{t} \big\rbrace \\
&= \Gamma_{i}^{t} \setminus \big(\Gamma_{j}^{t} \cup \{j\}\big) = \Gamma_{i}^{t} \setminus T_{ij}^{t}
\end{align}\noindent
where $|S_{i}^{t}|$ denotes the number of nodes of type $t$ that form 3-stars centered at node $i$ (not including $j$).
\end{Definition}
\begin{Definition}[Typed 3-Star Nodes Centered at $j$] \label{def:typed-3-star-node-centered-at-j}
Given an edge $(i,j) \in E$ between node $i$ and $j$ with types $\phi_i$ and $\phi_j$, 
let $S_{j}^{t}$ denote the set of nodes of type $t$ that form 3-node stars centered at node $j$ (and not including $i$) defined formally as:
\begin{align} \label{eq:3-node-typed-star-centered-at-node-j}
S_{j}^{t} &= \big\lbrace k \in (\Gamma_{j}^{t} \setminus \{i\}) \; \big| \; k \notin \Gamma_{i}^{t} \big\rbrace \\
&= \Gamma_{j}^{t} \setminus \big(\Gamma_{i}^{t} \cup \{i\}\big) = \Gamma_{j}^{t} \setminus T_{ij}^{t}
\end{align}\noindent
where $|S_{j}^{t}|$ denotes the number of nodes of type $t$ that form 3-stars centered at node $j$ (not including $i$).
\end{Definition}
\begin{Property} \label{prop:relationship-between-typed-sets-and-untyped-sets}
\begin{align} 
T_{ij} = \bigcup_{t=1}^{L} T_{ij}^{t},\quad\;\;\;
S_{i} = \bigcup_{t=1}^{L} S_{i}^{t},\quad\;\;\;
S_{j} = \bigcup_{t=1}^{L} S_{j}^{t}
\end{align}
\end{Property}
This property is shown in Figure~\ref{fig:typed-lower-order-sets}.
These lower-order 3-node typed motif counts are used to derive many higher-order typed motif counts in $o(1)$ constant time (Section~\ref{sec:combinatorial-relationships}).
\begin{Definition}[Typed 3-Stars (for an edge)] \label{def:typed-3-star-for-edge}
Given an edge $(i,j) \in E$ between node $i$ and $j$ with types $\phi_i$ and $\phi_j$, 
the number of typed 3-node stars that contain $(i,j) \in E$ with types $\phi_i$, $\phi_j$, $t$ is: 
\begin{align} \label{eq:3-node-typed-star-for-edge}
|S_{ij}^{t}| = |S_{i}^{t}|+|S_{j}^{t}|
\end{align}\noindent
where $|S_{ij}^{t}|$ denotes the number of typed 3-stars that contain nodes $i$ and $j$ with types $\phi_i$, $\phi_j$, $t$.
\end{Definition}
Moreover, the number of typed triangles centered at $(i,j) \in E$ with types $\phi_i$, $\phi_j$, $t$ is simply $|T_{ij}^{t}|$ (Definition~\ref{def:typed-triangle-node}) whereas the number of typed 3-node stars that contain $(i,j) \in E$ with types $\phi_i$, $\phi_j$, $t$ is $|S_{ij}^{t}| = |S_{i}^{t}|+|S_{j}^{t}|$ (Definition~\ref{def:typed-3-star-for-edge}).
We do not need to actually store the sets $S_{i}^{t}$, $S_{j}^{t}$, and $T_{ij}^{t}$ for every type $t = 1, \ldots, L$.
We only need to store the \emph{size/cardinality} of the sets (as shown in Algorithm~\ref{alg:typed-motifs-exact}). 
For convenience, we denote the size of those sets as $|S_{i}^{t}|$, $|S_{j}^{t}|$, and $|T_{ij}^{t}|$ for all $t=1,\ldots,L$, respectively.

{
\algblockdefx[parallel]{parfor}{endpar}[1][]{$\textbf{parallel for}$ #1 $\textbf{do}$}{$\textbf{end parallel}$}
\algrenewcommand{\alglinenumber}[1]{\fontsize{6.5}{7}\selectfont#1}
\begin{figure}[h!]
\begin{center}
\begin{algorithm}[H]
\caption{\,
Heterogeneous Path-based Network Motifs
}
\label{alg:typed-path-based-motifs-exact}
\begin{spacing}{1.2}
\fontsize{8}{9}\selectfont
\begin{algorithmic}[1]
\Require a graph $G=(V,E,\Phi,\xi)$, 
an edge $(i,j)$, 
sets of nodes $S_i$ and $S_j$ that form 3-paths centered at $i$ and $j$, respectively, 
a typed motif count vector $\vx$ for $(i,j)$, 
and 
set $\mathcal{M}_{ij}$ of unique typed motifs for $i$ and $j$.
\smallskip

	\For{{\bf each} $w_k \in S_i$ in {\bf order} $w_1, w_2, \cdots$ of $S_{i}$}
	\label{algline:Si-three-paths-centered-at-i}
	    
	    \For{$w_r \in \Gamma_{w_k}$}
	        \State \textbf{if} $(w_r=i) \vee (w_r=j)$ \textbf{continue}

\If{$w_r \not\in (\Gamma_i \cup \Gamma_j)$} \Comment{$\Psi(w_r)=0$, then \emph{4-path-edge} orbit} 

\State $\langle \vx, \mathcal{M}_{ij} \rangle = \textsc{Update}(\vx,\mathcal{M}_{ij},\mathbb{F}(g_3, \Phi_i, \Phi_j, \Phi_{w_k}, \Phi_{w_r}))$

\ElsIf{$w_r \in S_i \wedge w_r \leq w_k$} 
\Comment{$\Psi(w_r)$=$\lambda_1$, tailed-tri (tail)} 
\label{algline:tailed-tri-edge-tail-edge-orbit-Si}

		   \State $\langle \vx, \mathcal{M}_{ij} \rangle = \textsc{Update}(\vx,\mathcal{M}_{ij},\mathbb{F}(g_7, \Phi_i, \Phi_j, \Phi_{w_k}, \Phi_{w_r}))$

\EndIf
	    \EndFor 
	    
	\EndFor \label{algline:Si-three-paths-centered-at-i-endfor}
	
	\For{{\bf each} $w_k \in S_j$ in {\bf order} $w_1, w_2, \cdots$ of $S_{j}$} \label{algline:Sj-three-paths-centered-at-j}
	    
	    \For{$w_r \in \Gamma_{w_k}$}
	        \State \textbf{if} $(w_r=i) \vee (w_r=j)$ \textbf{continue}

\If{$w_r \not\in (\Gamma_i \cup \Gamma_j)$} \Comment{$\Psi(w_r)=0$, \emph{typed 4-path-edge} orbit}
\State $\langle \vx, \mathcal{M}_{ij} \rangle = \textsc{Update}(\vx,\mathcal{M}_{ij},\mathbb{F}(g_3, \Phi_i, \Phi_j, \Phi_{w_k}, \Phi_{w_r}))$

\ElsIf{$w_r \in S_j \wedge w_r \leq w_k$} 
\Comment{$\Psi(w_r)$=$\lambda_2$, \emph{tailed-tri} (tail)} 
\label{algline:tailed-tri-edge-tail-edge-orbit-Sj}
\State $\langle \vx, \mathcal{M}_{ij} \rangle = \textsc{Update}(\vx,\mathcal{M}_{ij},\mathbb{F}(g_7, \Phi_i, \Phi_j, \Phi_{w_k}, \Phi_{w_r}))$

\ElsIf{$w_r \in S_i$} \label{algline:typed-4-cycle} \Comment{$\Psi(w_r)=\lambda_1$, \emph{typed 4-cycle}}

\State $\langle \vx, \mathcal{M}_{ij} \rangle = \textsc{Update}(\vx,\mathcal{M}_{ij},\mathbb{F}(g_6, \Phi_i, \Phi_j, \Phi_{w_k}, \Phi_{w_r}))$
\EndIf	
	    \EndFor
	    
\EndFor \label{algline:Sj-three-paths-centered-at-j-endfor}
	
\State {\bf return} set of typed motifs $\mathcal{M}_{ij}$ occurring between $i$ and $j$ and $\vx$
\smallskip
\end{algorithmic}
\end{spacing}
\vspace{-1.mm}
\end{algorithm}
\end{center}
\vspace{-8.mm}
\end{figure}
}

{
\algblockdefx[parallel]{parfor}{endpar}[1][]{$\textbf{parallel for}$ #1 $\textbf{do}$}{$\textbf{end parallel}$}
\algrenewcommand{\alglinenumber}[1]{\fontsize{6.5}{7}\selectfont#1}
\begin{figure}[h!]
\begin{center}
\begin{algorithm}[H]
\caption{\,
Heterogeneous Triangle-based Network Motifs
}
\label{alg:typed-triangle-based-motifs-exact}
\begin{spacing}{1.2}
\fontsize{8}{9}\selectfont
\begin{algorithmic}[1]
\Require a graph $G=(V,E,\Phi,\xi)$, 
an edge $(i,j)$, 
set of nodes $T_{ij}$ that form triangles with $i$ and $j$,
sets of nodes $S_i$ and $S_j$ that form 3-paths centered at $i$ and $j$, respectively, 
a typed motif count vector $\vx$ for $(i,j)$, 
and 
set $\mathcal{M}_{ij}$ of unique typed motifs for $i$ and $j$.
\smallskip

	\For{{\bf each} $w_k \in T_{ij}$ in {\bf order} $w_1, w_2, \cdots$ of $T_{ij}$}\label{algline:T-triangles} 

	    \For{$w_r \in \Gamma_{w_k}$}

\If{$w_r \in T_{ij} \wedge w_r \leq w_k$} \Comment{$\Psi(w_r)=\lambda_3$, \emph{typed 4-clique}} \label{algline:typed-4-clique}
\State $\langle \vx, \mathcal{M}_{ij} \rangle = \textsc{Update}(\vx,\mathcal{M}_{ij},\mathbb{F}(g_{12}, \Phi_i, \Phi_j, \Phi_{w_k}, \Phi_{w_r}))$

\ElsIf{$w_r \! \in \!(S_i \!\cup \!S_j)$} 
\Comment{$\Psi(w_r)$=$\lambda_1\!\!\vee\!\!\lambda_2$, \emph{chordal-cycle-edge}}
\label{algline:typed-4-chordal-cycle-edge-orbit}
\State \textbf{if} $(w_r=i) \vee (w_r=j)$ \textbf{continue}
\State $\langle \vx, \mathcal{M}_{ij} \rangle = \textsc{Update}(\vx,\mathcal{M}_{ij},\mathbb{F}(g_{10}, \Phi_i, \Phi_j, \Phi_{w_k}, \Phi_{w_r}))$

\ElsIf{$w_r \not\in (\Gamma_i \cup \Gamma_j)$} \Comment{\emph{tailed-tri-center}}
\label{algline:typed-4-tailed-tri-center-orbit}
\State \textbf{if} $(w_r=i) \vee (w_r=j)$ \textbf{continue}
\State $\langle \vx, \mathcal{M}_{ij} \rangle = \textsc{Update}(\vx,\mathcal{M}_{ij},\mathbb{F}(g_{8}, \Phi_i, \Phi_j, \Phi_{w_k}, \Phi_{w_r}))$
\EndIf
\EndFor
	\EndFor
\State {\bf return} set of typed motifs $\mathcal{M}_{ij}$ occurring between $i$ and $j$ and $\vx$
\smallskip
\end{algorithmic}
\end{spacing}
\vspace{-1.mm}
\end{algorithm}
\end{center}
\vspace{-8.mm}
\end{figure}
}

\makeatletter
\global\let\tikz@ensure@dollar@catcode=\relax
\makeatother
\tikzstyle{every node}=[font=\large,line width=1.5pt]
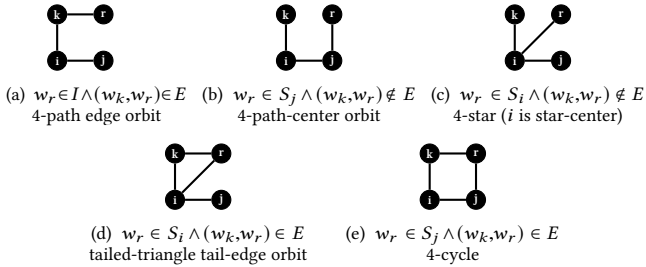
\begin{figure}[t!]
\centering
\begin{center}

\subfigure[$w_r\!\!\in \! I \!\wedge\! (w_k,\!\!w_r) \!\!\in\! E$ \centerline{4-path edge orbit}]{
\scalebox{0.28}{
\centering
\begin{tikzpicture}[-,>=latex,auto,node distance=2.2cm,thick,main node/.style={circle,draw=black,fill=black,draw,font=\sffamily\Huge\bfseries,text=white,
minimum width=0.9cm
}]

\node[main node] (1) {$\mathbf{i}$};
\node[main node] (2) [right of=1] {$\mathbf{j}$};
\node[main node] (3) [above of=1] {$\mathbf{k}$};
\node[main node] (4) [above of=2] {$\mathbf{r}$};
\node[main node] (5) [above right of=2,fill=white,draw=white,text=white]{};
\node[main node] (6) [above left of=1,fill=white,draw=white,text=white]{};
\node[main node] (8) [above right of=5,fill=white,draw=white,text=white]{};

\tikzstyle{LabelStyle}=[below=3pt]
\path[every node/.style={font=\sffamily}] 
(1) edge [line width=1.0mm, left] node [above left] {} (2) 
	(1)  edge [line width=1.0mm, right] node[below right] {} (3)
	(3) edge [line width=1.0mm,right] node[above right] {} (4);
\end{tikzpicture}
}
\label{fig:4-path-edge-orbit-Si}
}
\hfill
\subfigure[$w_r \!\!\in \!S_j \!\wedge\!  (w_k,\!\!w_r) \!\!\not\in\! E$ \centerline{4-path-center orbit}]{
\scalebox{0.28}{
\centering
\begin{tikzpicture}[-,>=latex,auto,node distance=2.2cm,thick,main node/.style={circle,draw=black,fill=black,draw,font=\sffamily\Huge\bfseries,text=white,
minimum width=0.9cm
}]

\node[main node] (1) {$\mathbf{i}$};
\node[main node] (2) [right of=1] {$\mathbf{j}$};
\node[main node] (3) [above of=2] {$\mathbf{r}$};
\node[main node] (4) [above of=1] {$\mathbf{k}$};
\node[main node] (5) [above right of=2,fill=white,draw=white,text=white]{};
\node[main node] (6) [above left of=1,fill=white,draw=white,text=white]{};
\node[main node] (7) [above left of=6,fill=white,draw=white,text=white]{};
\node[main node] (8) [above right of=5,fill=white,draw=white,text=white]{};

\tikzstyle{LabelStyle}=[below=3pt]
\path[every node/.style={font=\sffamily}] 
(1) edge [line width=1.0mm, left] node [above left] {} (2) 
	(2)  edge [line width=1.0mm, right] node[below right] {} (3)
	(1) edge [line width=1.0mm,right] node[above right] {} (4);
\end{tikzpicture}
}
\label{fig:4-path-center-orbit}
}
\hfill
\subfigure[$w_r \!\in \!S_i \!\wedge\! (w_k,\!\!w_r) \!\not\in\! E$ \centerline{4-star ($i$ is star-center)}]{
\scalebox{0.28}{
\centering
\begin{tikzpicture}[-,>=latex,auto,node distance=2.2cm,thick,main node/.style={circle,draw=black,fill=black,draw,font=\sffamily\Huge\bfseries,text=white,
minimum width=0.9cm
}]

\node[main node] (1) {$\mathbf{i}$};
\node[main node] (2) [right of=1] {$\mathbf{j}$};
\node[main node] (3) [above of=1] {$\mathbf{k}$};
\node[main node] (4) [above of=2] {$\mathbf{r}$};
\node[main node] (5) [above right of=2,fill=white,draw=white,text=white]{};
\node[main node] (6) [above left of=1,fill=white,draw=white,text=white]{};
\node[main node] (7) [above left of=6,fill=white,draw=white,text=white]{};
\node[main node] (8) [above right of=5,fill=white,draw=white,text=white]{};

\tikzstyle{LabelStyle}=[below=3pt]
\path[every node/.style={font=\sffamily}] 
(1) edge [line width=1.0mm, left] node [above left] {} (2) 
	(1)  edge [line width=1.0mm, right] node[below right] {} (3)
	(1) edge [line width=1.0mm,right] node[above right] {} (4);
\end{tikzpicture}
}
\label{fig:4-star-node-i-is-star-center}
}

\vspace{-4mm}
\subfigure[$w_r \!\!\in \! S_i \!\wedge\! (w_k,\!\!w_r) \!\!\in\! E$ \centerline{tailed-triangle tail-edge orbit}]{
\scalebox{0.28}{
\centering
\begin{tikzpicture}[-,>=latex,auto,node distance=2.2cm,thick,main node/.style={circle,draw=black,fill=black,draw,font=\sffamily\Huge\bfseries,text=white,
minimum width=0.9cm
}]

\node[main node] (1) {$\mathbf{i}$};
\node[main node] (2) [right of=1] {$\mathbf{j}$};
\node[main node] (3) [above of=1] {$\mathbf{k}$};
\node[main node] (4) [above of=2] {$\mathbf{r}$};
\node[main node] (5) [above right of=2,fill=white,draw=white,text=white]{};
\node[main node] (6) [above left of=1,fill=white,draw=white,text=white]{};
\node[main node] (7) [above left of=6,fill=white,draw=white,text=white]{};
\node[main node] (8) [above right of=5,fill=white,draw=white,text=white]{};

\tikzstyle{LabelStyle}=[below=3pt]
\path[every node/.style={font=\sffamily}] 
(1) edge [line width=1.0mm, left] node [above left] {} (2) 
	(1)  edge [line width=1.0mm, right] node[below right] {} (3)
(1) edge [line width=1.0mm, left] node[below left] {} (4)
	(3) edge [line width=1.0mm,right] node[above right] {} (4);
\end{tikzpicture}
}
\label{fig:tailed-triangle-tail-orbit-Si}
}
\hspace{4mm}
\subfigure[$w_r \!\in \!S_j \!\wedge\!  (w_k,\!\!w_r) \!\in\! E$ \centerline{4-cycle}]{
\scalebox{0.28}{
\centering
\begin{tikzpicture}[-,>=latex,auto,node distance=2.2cm,thick,main node/.style={circle,draw=black,fill=black,draw,font=\sffamily\Huge\bfseries,text=white,
minimum width=0.9cm
}]

\node[main node] (1) {$\mathbf{i}$};
\node[main node] (2) [right of=1] {$\mathbf{j}$};
\node[main node] (3) [above of=2] {$\mathbf{r}$};
\node[main node] (4) [above of=1] {$\mathbf{k}$};
\node[main node] (5) [above right of=2,fill=white,draw=white,text=white]{};
\node[main node] (6) [above left of=1,fill=white,draw=white,text=white]{};
\node[main node] (7) [above left of=6,fill=white,draw=white,text=white]{};
\node[main node] (8) [above right of=5,fill=white,draw=white,text=white]{};

\tikzstyle{LabelStyle}=[below=3pt]
\path[every node/.style={font=\sffamily}] 
(1) edge [line width=1.0mm, left] node [above left] {} (2) 
	(2)  edge [line width=1.0mm, right] node[below right] {} (3)
(3) edge [line width=1.0mm, left] node[below left] {} (4)
	(1) edge [line width=1.0mm,right] node[above right] {} (4);
\end{tikzpicture}
}
\label{fig:4-cycle}
}

\end{center}

\vspace{-4mm}
\caption{Path-based 4-node orbits derived from lower-order ($k\!-\!1$)-node graphlets in Algorithm~\protect\ref{alg:typed-path-based-motifs-exact}.
For each motif orbit above, we provide the equation for counting the motif orbit.
Since $w_k \in S_{i}$ holds in (a)-(e), it was removed from the equations above for brevity.
}
\label{fig:typed-path-based-4-node-motifs}
\vspace{-3mm}
\end{figure}

\subsection{Counting Typed 4-Node Motifs} \label{sec:algorithm-4-node-typed-motifs}
To derive $k$-node typed graphlets, the framework leverages the lower-order ($k\!-\!1$)-node typed graphlets.
Therefore, $4$-node typed graphlets are derived by leveraging the sets $T_{ij} = \Gamma_i \cup \Gamma_j$, 
$S_j = \Gamma_j \setminus T_{ij}$,
and $S_i = \Gamma_i \setminus T_{ij}$ 
computed from the lower-order $3$-node typed graphlets along with the set $I$ of non-adjacent nodes $\wrt$ $(i,j) \in E$ defined formally as follows: 
\begin{align}\label{eq:indep-node-set}
I &= V \setminus (\Gamma_i \cup \Gamma_j \cup \{i,j\}) \\
&= V \setminus (T_{ij} \cup S_i \cup S_j \cup \{i,j\}). \nonumber
\end{align}
\begin{Property} \label{prop:set-I}
\begin{equation}
|V| = |I| + |S_i| + |S_j| + |T_{ij}| + |\{i,j\}|
\end{equation}\noindent
\end{Property}\noindent
The proof is straightforward by Eq.~\ref{eq:indep-node-set} and applying 
the principle of inclusion-exclusion~\cite{pgd-kais}.

\definecolor{gray}{RGB}{100,100,100}
\definecolor{darkgray}{RGB}{150,150,150}

\definecolor{theblue}{RGB}{0, 20, 159} 
\definecolor{thelightblue}{RGB}{0,191,255}
\definecolor{thelightred}{RGB}{255,191,0}
\definecolor{thecrimson}{RGB}{	153, 0, 0}

\makeatletter
\global\let\tikz@ensure@dollar@catcode=\relax
\makeatother
\tikzstyle{every node}=[font=\large,line width=1.5pt]
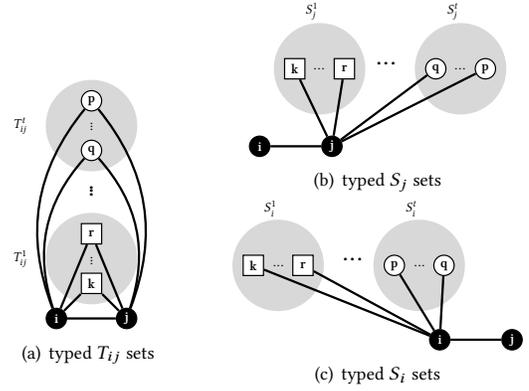
\begin{figure}[t!]
\centering

\tikzstyle{background-type1}=[circle,
fill=gray!25, 
inner sep=0.3cm,
rounded corners=4mm]

\tikzstyle{background-type2}=[circle,
fill=gray!25,
inner sep=0.3cm,
rounded corners=4mm]

\tikzstyle{background-type3}=[rectangle,
fill=green!15,
inner sep=0.3cm,
rounded corners=4mm]

\tikzstyle{background-type4}=[rectangle,
fill=gray!25,
inner sep=0.3cm,
rounded corners=4mm]                                                

\tikzstyle{background}=[rectangle,
fill=purple!10,
inner sep=0.3cm,
rounded corners=4mm]

\tikzstyle{background-white}=[rectangle,
fill=white,
inner sep=0cm,
rounded corners=0mm]                                                

\begin{minipage}[b]{0.38\linewidth}
\centering

\subfigure[typed $T_{ij}$ sets]{
\scalebox{0.3}{
\centering
\begin{tikzpicture}[-,>=latex,auto,node distance=2.2cm,thick,main node/.style={circle,draw=black,fill=black,draw,font=\sffamily\Huge\bfseries,text=white,
},
type0 node/.style={circle,draw=black,fill=white,draw,font=\sffamily\Huge\bfseries,text=black,minimum width=0.9cm,minimum height=0.9cm},
type1 node/.style={rectangle,draw=black,fill=white,draw,text=black,font=\sffamily\Huge\bfseries,minimum height=0.9cm},
type2 node/.style={diamond,draw=black,fill=white,draw,font=\sffamily\Huge\bfseries,text=black,minimum width=0.9cm,minimum height=0.9cm},
type3 node/.style={circle,draw=gray!80,fill=white,draw,font=\sffamily\Huge\bfseries,text=black},
type4 node/.style={circle,draw=green!80,fill=white,draw,font=\sffamily\Huge\bfseries,text=black},
white node/.style={circle,draw=white,fill=white,draw,font=\sffamily\Huge\bfseries,text=white,minimum width=0.1cm},
green node/.style={circle,draw=gray!25,fill=gray!25,draw,font=\sffamily\Huge\bfseries,text=black,minimum width=0.1cm},
blue node/.style={circle,draw=gray!25,fill=gray!25,draw,font=\sffamily\Huge\bfseries,text=black,minimum width=0.1cm},
text=black,minimum width=0.9cm,font=\sffamily\Huge\bfseries]

\node[main node] (1) {$\mathbf{i}$};
\node[type1 node] (3) [above right of=1] {$\mathbf{k}$};
\node[main node] (2) [below right of=3] {$\mathbf{j}$};
\node[type1 node] (4) [above of=3] {$\mathbf{r}$};
\node[white node] (5) [above of=4,below=15pt] {};
\node[type0 node] (6) [above of=5] {$\mathbf{q}$};
\node[type0 node] (7) [above of=6] {$\mathbf{p}$};
\node[green node] (10) [above of=3,below=18pt] {$\vdots$};
\node[blue node] (11) [above of=6,below=16pt] {$\vdots$};

\tikzstyle{LabelStyle}=[below=3pt]
\path[every node/.style={font=\sffamily}] 
(1) edge [line width=1.0mm, left] node [above left] {} (2) 
	(1)  edge [line width=1.0mm, right] node[below right] {} (3)
(2) edge [line width=1.0mm, left] node[below left] {} (3)
(1) edge [line width=1.0mm, left] node[below left] {} (4)
	(2) edge [line width=1.0mm,right] node[above right] {} (4)
	    (1) edge [bend left, line width=1.0mm, left] node[below left] {} (6)
	(2) edge [bend right, line width=1.0mm,right] node[above right] {} (6)
	    (1) edge [bend left, line width=1.0mm, left] node[below left] {} (7)
	(2) edge [bend right,line width=1.0mm,right] node[above right] {} (7);
	
\begin{pgfonlayer}{background}
\node [background-type1,
fit=(3) (4), font=\sffamily\Huge\bfseries,label=left:\Huge $T_{ij}^{1}\;\;\;\;$] {};
\node [background-type2, fit=(6) (7),font=\sffamily\Huge\bfseries,label=left:\Huge $\Huge T_{ij}^{t}\;\;\;\;$] {};
\node [background-white, fit=(5),font=\sffamily\Huge\bfseries,label=mid:\Huge {\fontsize{32}{32}\selectfont $\mathrm{\vdots}$}] {};
\end{pgfonlayer}
\end{tikzpicture}
}
\label{fig:typed-triangle-nodes}
}
\vspace{3mm}
\end{minipage}
\begin{minipage}[b]{0.50\linewidth}
\centering
\subfigure[typed $S_j$ sets]{
\scalebox{0.3}{
\centering
\begin{tikzpicture}[-,>=latex,auto,node distance=2.2cm,thick,main node/.style={circle,draw=black,fill=black,draw,font=\sffamily\Huge\bfseries,text=white,
},
type0 node/.style={circle,draw=black,fill=white,draw,font=\sffamily\Huge\bfseries,text=black,minimum width=0.9cm,minimum height=0.9cm},
type1 node/.style={rectangle,draw=black,fill=white,draw,text=black,font=\sffamily\Huge\bfseries,minimum height=0.9cm},
type2 node/.style={diamond,draw=black,fill=white,draw,font=\sffamily\Huge\bfseries,text=black,minimum width=0.9cm,minimum height=0.9cm},
type3 node/.style={circle,draw=gray!80,fill=white,draw,font=\sffamily\Huge\bfseries,text=black},
type4 node/.style={circle,draw=green!80,fill=white,draw,font=\sffamily\Huge\bfseries,text=black},
white node/.style={circle,draw=white,fill=white,draw,font=\sffamily\Huge\bfseries,text=white,minimum width=0.1cm},
green node/.style={circle,draw=gray!25,fill=gray!25,draw,font=\sffamily\Huge\bfseries,text=black,minimum width=0.1cm},
blue node/.style={circle,draw=gray!25,fill=gray!25,draw,font=\sffamily\Huge\bfseries,text=black,minimum width=0.1cm},
text=black,minimum width=0.9cm,font=\sffamily\Huge\bfseries]

\node[main node] (1) {$\mathbf{i}$};
\node[type1 node] (3) [above right of=1, above=40pt] {$\mathbf{k}$};
\node[main node] (2) [right of=1,right=15pt] {$\mathbf{j}$};
\node[type1 node] (4) [right of=3] {$\mathbf{r}$};
\node[white node] (5) [right of=4, left=5pt] {};
\node[type0 node] (6) [right of=5] {$\mathbf{q}$};
\node[type0 node] (7) [right of=6] {$\mathbf{p}$};
\node[green node] (10) [right of=3,left=20pt] {$...$};
\node[blue node] (11) [right of=6,left=20pt] {$...$};

\tikzstyle{LabelStyle}=[below=3pt]
\path[every node/.style={font=\sffamily}] 
(1) edge [line width=1.0mm, left] node [above left] {} (2) 
(2) edge [line width=1.0mm, left] node[below left] {} (3)
	(2) edge [line width=1.0mm,right] node[above right] {} (4)
	(2) edge [line width=1.0mm,right] node[above right] {} (6)
	(2) edge [line width=1.0mm,right] node[above right] {} (7);
	
\begin{pgfonlayer}{background}
\node [background-type1,
fit=(3) (4), font=\sffamily\Huge\bfseries,label=above:\Huge $S_{j}^{1}\;\;\;\;$] {};
\node [background-type2, fit=(6) (7),font=\sffamily\Huge\bfseries,label=above:\Huge $\Huge S_{j}^{t}\;\;\;\;$] {};
\node [background-white, fit=(5),font=\sffamily\Huge\bfseries,label=mid:\Huge {\fontsize{32}{32}\selectfont $\mathrm{...}$}] {};
\end{pgfonlayer}
\end{tikzpicture}
}
\label{fig:typed-3-path-nodes-centered-at-j}
}

\vspace{-3mm}
\subfigure[typed $S_i$ sets]{
\scalebox{0.3}{
\centering
\begin{tikzpicture}[-,>=latex,auto,node distance=2.2cm,thick,main node/.style={circle,draw=black,fill=black,draw,font=\sffamily\Huge\bfseries,text=white,
},
type0 node/.style={circle,draw=black,fill=white,draw,font=\sffamily\Huge\bfseries,text=black,minimum width=0.9cm,minimum height=0.9cm},
type1 node/.style={rectangle,draw=black,fill=white,draw,text=black,font=\sffamily\Huge\bfseries,minimum height=0.9cm},
type2 node/.style={diamond,draw=black,fill=white,draw,font=\sffamily\Huge\bfseries,text=black,minimum width=0.9cm,minimum height=1.2cm},
type3 node/.style={circle,draw=gray!80,fill=white,draw,font=\sffamily\Huge\bfseries,text=black},
type4 node/.style={circle,draw=green!80,fill=white,draw,font=\sffamily\Huge\bfseries,text=black},
white node/.style={circle,draw=white,fill=white,draw,font=\sffamily\Huge\bfseries,text=white,minimum width=0.1cm},
green node/.style={circle,draw=gray!25,fill=gray!25,draw,font=\sffamily\Huge\bfseries,text=black,minimum width=0.1cm},
blue node/.style={circle,draw=gray!25,fill=gray!25,draw,font=\sffamily\Huge\bfseries,text=black,minimum width=0.1cm},
text=black,minimum width=0.9cm,font=\sffamily\Huge\bfseries]

\node[main node] (1) {$\mathbf{i}$};
\node[type0 node] (3) [above left of=1, above right=40pt] {$\mathbf{q}$};
\node[main node] (2) [right of=1,right=15pt] {$\mathbf{j}$};
\node[type0 node] (4) [left of=3] {$\mathbf{p}$};
\node[white node] (5) [left of=4, right=5pt] {};
\node[type1 node] (6) [left of=5] {$\mathbf{r}$};
\node[type1 node] (7) [left of=6] {$\mathbf{k}$};
\node[green node] (10) [left of=3,right=20pt] {$...$};
\node[blue node] (11) [left of=6,right=20pt] {$...$};

\tikzstyle{LabelStyle}=[below=3pt]
\path[every node/.style={font=\sffamily}] 
(1) edge [line width=1.0mm, left] node [above left] {} (2) 
(1) edge [line width=1.0mm, left] node[below left] {} (3)
	(1) edge [line width=1.0mm,right] node[above right] {} (4)
	(1) edge [line width=1.0mm,right] node[above right] {} (6)
	(1) edge [line width=1.0mm,right] node[above right] {} (7);
	
\begin{pgfonlayer}{background}
\node [background-type1,
fit=(3) (4), font=\sffamily\Huge\bfseries,label=above:\Huge $S_{i}^{t}\;\;\;\;$] {};
\node [background-type2, fit=(6) (7),font=\sffamily\Huge\bfseries,label=above:\Huge $\Huge S_{i}^{1}\;\;\;\;$] {};
\node [background-white, fit=(5),font=\sffamily\Huge\bfseries,label=mid:\Huge {\fontsize{32}{32}\selectfont $\mathrm{...}$}] {};
\end{pgfonlayer}
\end{tikzpicture}
}
\label{fig:typed-3-path-nodes-centered-at-i}
}
\hspace{6mm}
\end{minipage}
\hfill

\vspace{-3mm}
\caption{Typed lower-order sets used to derive many higher-order graphlets in constant time.
Note node $i$ and $j$ can be of arbitrary types.
}
\label{fig:typed-lower-order-sets}
\vspace{-3mm}
\end{figure}

\subsubsection{General Principle for Counting Typed Graphlets} \label{sec:general-principle}
We now introduce a general typed graphlet formulation.
Let $f_{ij}(H, \vt)$ denote the number of distinct k-node typed graphlet orbits of $H$ with the type vector $\vt$ that contain edge $(i,j) \in E$ 
and have properties $P,Q \in \{S_i, S_j, T_{ij}, I\}$ defined formally as:
\begin{align} \label{eq:general-typed-graphlet-formulation} 
f_{ij}(H, \vt) = \Big|\Big\{ \{i,j,w_k,w_r\} \,\big|\,
& w_k \!\in P \wedge 
w_r \!\in Q \wedge \\ \nonumber
&\mathbb{I}\{(w_k, w_r) \!\in E\} \wedge
w_r \not= w_k \wedge \\  \nonumber
&\vt = \big[ \phi_i \;\; \phi_j \;\; \phi_{w_k} \; \phi_{w_r}\big]
\Big\}\Big|
\end{align}\noindent
where $\mathbb{I}\{(w_k, w_r) \!\in E\} = 1$ if $(w_k, w_r) \!\in E$ holds and $0$ otherwise (\ie, $\mathbb{I}\{(w_k, w_r) \!\in E\} = 0$ if $(w_k, w_r) \!\not\in E$).
For clarity and simplicity, $(w_k,w_r) \in E$ or $(w_k,w_r) \not\in E$ is sometimes used (\eg, Table~\ref{table:typed-graphlet-equations}) as opposed to $\mathbb{I}\{(w_k, w_r) \!\in E\}=1$ or $\mathbb{I}\{(w_k, w_r) \!\in E\}=0$.

The equations for deriving every typed graphlet orbit are provided in Table~\ref{table:typed-graphlet-equations}.
Notice that all typed graphlets with $k$-nodes are formulated with respect to the node sets $\{S_i,S_j,T_{ij},I\}$ derived from the typed graphlets with ($k\!-\!1$)-nodes.
Hence, the higher-order typed graphlets with order $k$ are derived from the lower-order ($k\!-\!1$)-node typed graphlets.

We classify typed motifs as path-based or triangle-based.
Typed path-based motifs are the typed 4-node motifs derived from the sets $S_i$ and $S_j$ of nodes that form 3-node paths centered at node $i$ and $j$, respectively.
Conversely, typed triangle-based motifs are the typed 4-node motifs derived from the set $T_{ij}$ of nodes that form triangles (3-cliques) with node $i$ and $j$.
Naturally, typed path-based motifs are the least dense (motifs with fewest edges) whereas the typed triangle-based motifs are the most dense.
Recall $T_{ij} = \Gamma_i \cap \Gamma_j$, $S_j = \Gamma_j \setminus T_{ij}$, $S_i= \Gamma_i \setminus T_{ij}$, and $I = V \setminus (\Gamma_i \cup \Gamma_j \cup \{i,j\}) = V \setminus (T_{ij} \cup S_i \cup S_j \cup \{i,j\})$.
Therefore, $|\Gamma_i| \geq |T_{ij}|$ and $|\Gamma_i| \geq |S_i|$.
Further, if $|\Gamma_i| = |T_{ij}|$ then $|S_i|=0$ and conversely if $|\Gamma_i| = |S_{i}|$ then $|T_{ij}|=0$.

In this work, we derive equations using new non-trivial combinatorial relationships between lower-order typed ($k\!-\!1$)-graphlets that allow us to derive many of the $k$-node typed graphlets in $o(1)$ constant time.
Notably, we make no assumptions about the number of types, their distribution among the nodes and edges, or any other additional information.
On the contrary, the framework is extremely general for arbitrary heterogeneous graphs (see Figure~\ref{fig:heter-special-cases} for a number of popular special cases covered by the framework).
In addition, we also avoid a lot of computations by symmetry breaking techniques, and other conditions to avoid unnecessary work.
Typed graphlets that are computed in constant time include
typed 4-path-center orbit,
typed 4-node stars,
typed chordal-cycle-center orbit,
and typed tailed-triangle-tri-edge orbit. 
Notice that two of the typed tailed-triangle orbits including the tri-center and tail-edge orbit are 
essentially derived for free while computing the typed graphlets using $S_i$ (typed 4-cycles) and $S_j$.

\subsubsection{Directed Typed Graphlets}
The approach is also straightforward to adapt for directed typed motifs.
In particular, we simply replace $\Gamma_i^{t}$ with $\Gamma_i^{t,+}$ and $\Gamma_i^{t,-}$ for typed out-neighbors and typed in-neighbors, respectively.
Thus, we also have $T_{ij}^{t,+}$, $T_{ij}^{t,-}$, $S_j^{t,+}$, $S_j^{t,-}$, $S_i^{t,+}$, and $S_i^{t,-}$.
Now it is just a matter of enumerating all combinations of these sets with the out/in-neighbor sets as well.
That is, we essentially have two additional versions of Algorithm~\ref{alg:typed-motifs-exact} and Algorithm~\ref{alg:typed-path-based-motifs-exact}-\ref{alg:typed-triangle-based-motifs-exact} for each in and out set (\wrt to the main for loop).
The other trivial modification is to ensure each directed typed motif is assigned a unique id (this is the same modification required for typed orbits).
The time and space complexity remains the same since all we did is split the set of neighbors (and the other sets) into two smaller sets by partitioning the nodes in $\Gamma_i^{t}$ into $\Gamma_i^{t,+}$ and $\Gamma_i^{t,-}$.
Similarly, for $T_{ij}^{t}$, $S_j^{t}$, and $S_i^{t}$.

\begin{table*}[t!]
\caption{
Typed Graphlet Orbit Equations.
All typed graphlet orbits with $4$-nodes are formulated with respect to the node sets $\{S_i,S_j,T_{ij},I\}$ derived from the typed $3$-node graphlets.
Recall $T_{ij} = \Gamma_i \cap \Gamma_j$, $S_j = \Gamma_j \setminus T_{ij}$, $S_i= \Gamma_i \setminus T_{ij}$, and $I = V \setminus (\Gamma_i \cup \Gamma_j \cup \{i,j\}) = V \setminus (T_{ij} \cup S_i \cup S_j \cup \{i,j\})$.
Therefore, $|\Gamma_i| \geq |T_{ij}|$ and $|\Gamma_i| \geq |S_i|$.
Further, if $|\Gamma_i| = |T_{ij}|$ then $|S_i|=0$ and conversely if $|\Gamma_i| = |S_{i}|$ then $|T_{ij}|=0$.
In all cases, $w_r \not= w_k$. 
Note $\rho(H)$ is density.
}
\vspace{-3mm}
\centering 
\fontsize{8}{8.5}\selectfont
\setlength{\tabcolsep}{3pt} 
\label{table:typed-graphlet-equations}
\hspace*{-2.5mm}
\begin{tabularx}{1.00\linewidth}{rl c c XH@{}} 
\toprule
\textsc{Motif} $H$ & 
\textsc{Orbit} & $|E(H)|$ & $\rho(H)$ & 
$
f_{ij}(H, \vt) = \Big|\Big\{ \{i,j,w_k,w_r\} \,\big|\,
w_k \!\in P \wedge 
w_r \!\in Q \wedge 
\mathbb{I}\{(w_k, w_r) \!\in E\} \wedge
w_r \not= w_k \wedge
\vt = \big[ \phi_i \;\; \phi_j \;\; \phi_{w_k} \; \phi_{w_r}\big]
\Big\}\Big|
$
\\
\midrule

\textbf{4-path} & \text{edge} & 3 & 0.50 &
$
f_{ij}(g_3, \vt) = \Big|\Big\{ \{i,j,w_k,w_r\} \,\big|\,
w_k \!\in S_i \wedge 
w_r \!\in I \wedge 
(w_k, w_r) \!\in E 
\wedge
\vt = \big[ \phi_i \;\; \phi_j \;\; \phi_{w_k} \; \phi_{w_r}\big]
\Big\}\Big|
$
\\

& \text{center} & 3 & 0.50 &
$
f_{ij}(g_4, \vt) = \Big|\Big\{ \{i,j,w_k,w_r\} \,\big|\,
w_k \!\in \!S_j \wedge w_r \!\in \!S_i \wedge  (w_k,\!w_r) \!\not\in\! E
\wedge
\vt = \big[ \phi_i \;\; \phi_j \;\; \phi_{w_k} \; \phi_{w_r}\big]
\Big\}\Big|
$
\\
\midrule

\textbf{4-star} &  & 3 & 0.50 &
$
f_{ij}(g_5, \vt) = \Big|\Big\{ \{i,j,w_k,w_r\} \,\big|\,
w_k \!\in S_i \wedge 
w_r \!\in S_i \wedge 
(w_k, w_r) \!\not\in E 
\wedge
\vt = \big[ \phi_i \;\; \phi_j \;\; \phi_{w_k} \; \phi_{w_r}\big]
\Big\}\Big|
$
\\
\midrule

\textbf{4-cycle} & & 4 & 0.67 &
$
f_{ij}(g_6, \vt) = \Big|\Big\{ \{i,j,w_k,w_r\} \,\big|\,
w_k \!\in S_j \wedge 
w_r \!\in S_i \wedge 
(w_k, w_r) \!\in E 
\wedge
\vt = \big[ \phi_i \;\; \phi_j \;\; \phi_{w_k} \; \phi_{w_r}\big]
\Big\}\Big|
$
\\
\midrule

\textbf{tailed-triangle} 
& \text{tail-edge} 
& 4 & 0.67 &
$
f_{ij}(g_7, \vt) = \Big|\Big\{ \{i,j,w_k,w_r\} \,\big|\,
w_k \!\in S_i \wedge 
w_r \!\in S_i \wedge 
w_r \not= w_k \wedge 
(w_k, w_r) \!\in E 
\wedge
\vt = \big[ \phi_i \;\; \phi_j \;\; \phi_{w_k} \; \phi_{w_r}\big]
\Big\}\Big|
$
\\

& \text{center} 
& 4 & 0.67 &
$
f_{ij}(g_8, \vt) = \Big|\Big\{ \{i,j,w_k,w_r\} \,\big|\,
w_k \!\in T_{ij} \wedge 
w_r \!\in I \wedge
(w_k, w_r) \!\in E 
\wedge
\vt = \big[ \phi_i \;\; \phi_j \;\; \phi_{w_k} \; \phi_{w_r}\big]
\Big\}\Big|
$
\\

& \text{tri-edge} 
& 4 & 0.67 &
$
f_{ij}(g_9, \vt) = \Big|\Big\{ \{i,j,w_k,w_r\} \,\big|\,
w_k \!\in T_{ij} \wedge 
w_r \!\in S_i \wedge 
(w_k, w_r) \!\not\in E 
\wedge
\vt = \big[ \phi_i \;\; \phi_j \;\; \phi_{w_k} \; \phi_{w_r}\big]
\Big\}\Big|
$
\\
\midrule

\textbf{chordal-cycle} 
& \text{edge} & 5 & 0.83 &
$
f_{ij}(g_{10}, \vt) = \Big|\Big\{ \{i,j,w_k,w_r\} \,\big|\,
w_k \!\in T_{ij} \wedge 
w_r \!\in \big(S_i \cup S_j\big) \wedge 
w_r \not= w_k \wedge 
(w_k, w_r) \!\in E 
\wedge
\vt = \big[ \phi_i \;\; \phi_j \;\; \phi_{w_k} \; \phi_{w_r}\big]
\Big\}\Big|
$
\\

& \text{center} & 5 & 0.83 &
$
f_{ij}(g_{11}, \vt) = \Big|\Big\{ \{i,j,w_k,w_r\} \,\big|\,
w_k \!\in T_{ij} \wedge 
w_r \!\in T_{ij} \wedge 
w_r \not= w_k \wedge 
(w_k, w_r) \!\not\in E 
\wedge
\vt = \big[ \phi_i \;\; \phi_j \;\; \phi_{w_k} \; \phi_{w_r}\big]
\Big\}\Big|
$
\\
\midrule

\textbf{4-clique} 
& & 6 & 1.00 &
$
f_{ij}(g_{12}, \vt) = \Big|\Big\{ \{i,j,w_k,w_r\} \,\big|\,
w_k \!\in T_{ij} \wedge 
w_r \!\in T_{ij} \wedge 
w_r \not= w_k \wedge 
(w_k, w_r) \!\in E 
\wedge
\vt = \big[ \phi_i \;\; \phi_j \;\; \phi_{w_k} \; \phi_{w_r}\big]
\Big\}\Big|
$
\\

\bottomrule
\end{tabularx}
\end{table*}

\makeatletter
\global\let\tikz@ensure@dollar@catcode=\relax
\makeatother
\tikzstyle{every node}=[font=\large,line width=1.5pt]
\begin{figure}[t!]
\centering
\begin{center}

\subfigure[$w_r \!\!\in \! I \!\,\wedge\,\! (w_k,\!\!w_r) \!\!\in\! E$ \centerline{tailed-triangle tri-center orbit}]{
\scalebox{0.3}{
\centering
\begin{tikzpicture}[-,>=latex,auto,node distance=2.2cm,thick,main node/.style={circle,draw=black,fill=black,draw,font=\sffamily\Huge\bfseries,text=white,
minimum width=0.9cm
}]

\node[main node] (1) {$\mathbf{i}$};
\node[main node] (3) [above right of=1] {$\mathbf{k}$};
\node[main node] (2) [below right of=3] {$\mathbf{j}$};
\node[main node] (4) [above of=3] {$\mathbf{r}$};
\node[main node] (5) [above right of=2,fill=white,draw=white,text=white]{};
\node[main node] (6) [above left of=1,fill=white,draw=white,text=white]{};

\tikzstyle{LabelStyle}=[below=3pt]
\path[every node/.style={font=\sffamily}] 
(1) edge [line width=1.0mm, left] node [above left] {} (2) 
	(1)  edge [line width=1.0mm, right] node[below right] {} (3)
(2) edge [line width=1.0mm, left] node[below left] {} (3)
	(3) edge [line width=1.0mm,right] node[above right] {} (4);
\end{tikzpicture}
}
\label{fig:tailed-triangle-tri-center-orbit}
}
\hspace{4mm}
\subfigure[$w_r \!\!\in\! S_i \! \wedge \! (w_k,\!w_r) \!\! \not\in \! E$ \centerline{tailed-triangle $\text{tri-edge orbit}$}]{
\scalebox{0.3}{
\centering
\begin{tikzpicture}[-,>=latex,auto,node distance=2.2cm,thick,main node/.style={circle,draw=black,fill=black,draw,font=\sffamily\Huge\bfseries,text=white,
minimum width=0.9cm
}]

\node[main node] (1) {$\mathbf{i}$};
\node[main node] (3) [above right of=1] {$\mathbf{k}$};
\node[main node] (2) [below right of=3] {$\mathbf{j}$};
\node[main node] (4) [above left of=1] {$\mathbf{r}$};
\node[main node] (5) [above right of=2,fill=white,draw=white,text=white]{};
\node[main node] (6) [left of=4,fill=white,draw=white,text=white]{};

\tikzstyle{LabelStyle}=[below=3pt]
\path[every node/.style={font=\sffamily}] 
(1) edge [line width=1.0mm, left] node [above left] {} (2) 
	(1)  edge [line width=1.0mm, right] node[below right] {} (3)
(2) edge [line width=1.0mm, left] node[below left] {} (3)
	(1) edge [line width=1.0mm,right] node[above right] {} (4);
\end{tikzpicture}
}
\label{fig:tailed-tri-tri-edge-orbit-Si}
}

\subfigure[$\!w_r \!\!\in \! S_i \! \wedge \! (w_k,\!\!w_r) \!\!\in \! E$ 
chordal-cycle edge orbit]{ 
\scalebox{0.3}{
\centering
\begin{tikzpicture}[-,>=latex,auto,node distance=2.2cm,thick,
main node/.style={circle,draw=black,fill=black,draw,font=\sffamily\Huge\bfseries,text=white,minimum width=0.9cm},
white/.style={circle,draw=white,fill=white,draw,font=\sffamily,text=white,minimum width=0.001cm},
]

\node[main node] (1) {$\mathbf{i}$};
\node[main node] (3) [above right of=1] {$\mathbf{k}$};
\node[main node] (2) [below right of=3] {$\mathbf{j}$};
\node[main node] (4) [above of=3] {$\mathbf{r}$};
\node[main node] (5) [above right of=2,fill=white,draw=white,text=white,right=3pt]{};
\node[main node] (6) [above left of=1,fill=white,draw=white,text=white,left=3pt]{};

\tikzstyle{LabelStyle}=[below=3pt]
\path[every node/.style={font=\sffamily}] 
	(1) edge [bend left,line width=1.0mm] node[above right] {} (4)
(1) edge [line width=1.0mm, left] node [above left] {} (2) 
	(1)  edge [line width=1.0mm, right] node[below right] {} (3)
(2) edge [line width=1.0mm, left] node[below left] {} (3)
	(3) edge [line width=1.0mm,right] node[above right] {} (4);
\end{tikzpicture}
}
\label{fig:chordal-cycle-edge-orbit-Si}
}
\hfill
\subfigure[$w_r \!\!\in \!T_{ij} \! \wedge \! (w_k,\!\!w_r) \!\!\not\in \! E$ \centerline{chordal-cycle-center orbit}]{
\scalebox{0.3}{
\centering
\begin{tikzpicture}[-,>=latex,auto,node distance=2.2cm,thick,
main node/.style={circle,draw=black,fill=black,draw,font=\sffamily\Huge\bfseries,text=white,minimum width=0.9cm},
white/.style={circle,draw=white,fill=white,draw,font=\sffamily,text=white,minimum width=0.001cm},
]

\node[main node] (1) {$\mathbf{i}$};
\node[main node] (3) [above right of=1] {$\mathbf{k}$};
\node[main node] (2) [below right of=3] {$\mathbf{j}$};
\node[main node] (4) [above of=3] {$\mathbf{r}$};
\node[main node] (5) [above right of=2,fill=white,draw=white,text=white]{};
\node[main node] (6) [above left of=1,fill=white,draw=white,text=white]{};
\node[white] (7) [above left of=6,fill=white,draw=white,text=white,right=20pt]{};
\node[white] (8) [above right of=5,fill=white,draw=white,text=white,left=20pt]{};

\tikzstyle{LabelStyle}=[below=3pt]
\path[every node/.style={font=\sffamily}] 
	(1) edge [bend left,line width=1.0mm] node[above right] {} (4)
	(2) edge [bend right,line width=1.0mm] node[above right] {} (4)
(1) edge [line width=1.0mm, left] node [above left] {} (2) 
	(1)  edge [line width=1.0mm, right] node[below right] {} (3)
(2) edge [line width=1.0mm, left] node[below left] {} (3);
\end{tikzpicture}
}
\label{fig:chordal-cycle-center-orbit}
}
\hfill
\subfigure[$w_r \!\in \! T_{ij} \,\!\! \wedge \! (w_k,\!\!w_r) \!\!\in \! E$ \centerline{4-clique}]{
\scalebox{0.3}{
\centering
\begin{tikzpicture}[-,>=latex,auto,node distance=2.2cm,thick,
main node/.style={circle,draw=black,fill=black,draw,font=\sffamily\Huge\bfseries,text=white,minimum width=0.9cm},
white/.style={circle,draw=white,fill=white,draw,font=\sffamily,text=white,minimum width=0.001cm},
]

\node[main node] (1) {$\mathbf{i}$};
\node[main node] (3) [above right of=1] {$\mathbf{k}$};
\node[main node] (2) [below right of=3] {$\mathbf{j}$};
\node[main node] (4) [above of=3] {$\mathbf{r}$};
\node[main node] (5) [above right of=2,fill=white,draw=white,text=white]{};
\node[main node] (6) [above left of=1,fill=white,draw=white,text=white]{};
\node[white] (7) [above left of=6,fill=white,draw=white,text=white,right=20pt]{};
\node[white] (8) [above right of=5,fill=white,draw=white,text=white,left=20pt]{};

\tikzstyle{LabelStyle}=[below=3pt]
\path[every node/.style={font=\sffamily}] 
	(1) edge [bend left,line width=1.0mm] node[above right] {} (4)
	(2) edge [bend right,line width=1.0mm] node[above right] {} (4)
(1) edge [line width=1.0mm, left] node [above left] {} (2) 
	(1)  edge [line width=1.0mm, right] node[below right] {} (3)
(2) edge [line width=1.0mm, left] node[below left] {} (3)
	(3) edge [line width=1.0mm,right] node[above right] {} (4);
\end{tikzpicture}
}
\label{fig:4-clique}
}
\end{center}

\vspace{-3mm}
\caption{Triangle-based 4-node orbits derived from lower-order ($k\!-\!1$)-node graphlets in Algorithm~\protect\ref{alg:typed-triangle-based-motifs-exact}.
For each motif orbit above, we provide the equation used to count it.
Since $w_k \in T_{ij}$ holds for all triangle-based 4-node orbits, it was removed from the equations above for brevity.
}
\label{fig:triangle-based-4-node-motifs}
\vspace{-3mm}
\end{figure}
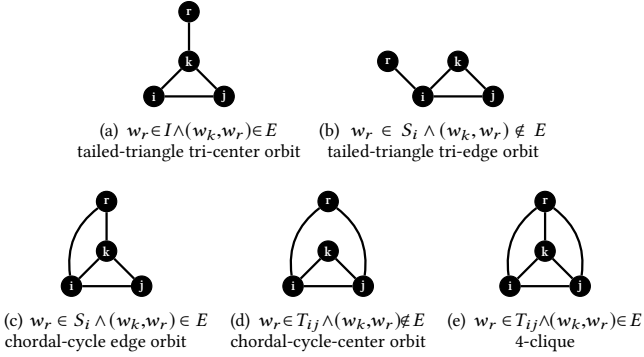

\subsection{Combinatorial Relationships} \label{sec:combinatorial-relationships}
Using new combinatorial relationships between lower-order typed graphlets, 
we derive all remaining typed graphlet orbits in $o(1)$ constant time via Eq.~\ref{eq:typed-4-path-center-orbit}-\ref{eq:typed-4-chordal-cycle-center-orbit} (See Line~\ref{algline:main-alg-for-type-pair}-\ref{algline:main-alg-derive-remaining-typed-graphlet-orbits-constant-time}).
Since we derive all typed graphlet counts for a given edge $(i,j) \in E$ between node $i$ and $j$, we already have two types $\phi_i$ and $\phi_j$.
Thus, these types are fixed ahead of time.
In the case of 4-node typed graphlets, there are two remaining types that need to be selected.
Notice that for typed graphlet orbits, we must solve $\frac{L(L-1)}{2}+L$ equations in the worst-case.
The counts of all remaining typed graphlets are derived in $o(1)$ constant time using the counts of the lower-order ($k\!-\!1$)-node typed graphlets and a few other counts from the $k$-node typed graphlets.
After deriving the exact count of each remaining graphlet with types $\phi_i$, $\phi_j$, $t$, and $t^{\prime}$ for every $t, t^{\prime} \in \{1, \ldots, L\}$ such that $t\leq t^{\prime}$ (Line~\ref{algline:main-alg-for-type-pair}-\ref{algline:main-alg-derive-remaining-typed-graphlet-orbits-constant-time}), if such count is nonzero, we compute a motif hash $\hash = \mathbb{F}(g, \phi_i, \phi_j, t, t^{\prime})$ for graphlet orbit $g$, set $\mathcal{M}_{ij} \leftarrow \mathcal{M}_{ij} \cup \{\hash\}$, and then set the count of that typed motif in $\vx_{\hash}$ to the count derived in constant $o(1)$ time.

\medskip\noindent\textbf{Typed 4-Path Center Orbit Count (Figure~\ref{fig:4-path-center-orbit})}: 
To count the typed 4-path center orbits for a given edge $(i,j) \in E$ with types $\phi_i$ and $\phi_j$, we simply select the remaining two types denoted as $t$ and $t^{\prime}$ to obtain the 4-dimensional type vector $\vt = \big[\, \phi_i \;\, \phi_j \;\; t \;\; t^{\prime} \,\big]$ and derive the count in $o(1)$ constant time as follows:
\begin{equation} \label{eq:typed-4-path-center-orbit}
f_{ij}(g_4,\vt) = 
\begin{cases}
(|S_{i}^{t}| \cdot |S_{j}^{t}|) - f_{ij}(g_{6}, \vt) 							 										 & \text{if } t=t^{\prime} \\[5pt]
(|S_{i}^{t}| \cdot |S_{j}^{t^{\prime}}|) + (|S_{i}^{t^{\prime}}| \cdot |S_{j}^{t}|) - f_{ij}(g_{6}, \vt)	 & \text{otherwise} \\[2pt]
\end{cases}
\end{equation}\noindent
where $f_{ij}(g_{6},\vt)$ is the typed 4-cycle count for edge $(i,j) \in E$ with type vector $\vt$.

\medskip\noindent\textbf{Typed 4-Star Count (Figure~\ref{fig:4-star-node-i-is-star-center})}: 
To count the typed 4-stars for a given edge $(i,j) \in E$ with types $\phi_i$ and $\phi_j$, we simply select the remaining two types denoted as $t$ and $t^{\prime}$ to obtain the 4-dimensional type vector $\vt = \big[\, \phi_i \;\, \phi_j \;\; t \;\; t^{\prime} \,\big]$.
We derive the typed 4-star counts with the type vector $\vt$ for edge $(i,j) \in E$ in constant time as follows: 
\begin{equation} \label{eq:typed-4-star}
f_{ij}(g_5,\vt) = 
\begin{cases}
\mychoose{|S_{i}^{t}|}{2} + \mychoose{|S_{j}^{t}|}{2} - f_{ij}(g_{7}, \vt) 							 & \text{if } t=t^{\prime} \\
(|S_{i}^{t}| \cdot |S_{i}^{t^{\prime}}|) + (|S_{j}^{t}| \cdot |S_{j}^{t^{\prime}}|) - f_{ij}(g_{7}, \vt)	 & \text{otherwise} \\[2pt]
\end{cases}
\end{equation}\noindent
where $f_{ij}(g_{7},\vt)$ is the tailed-triangle tail-edge orbit count for edge $(i,j) \in E$ with type vector $\vt$.
The only path-based typed graphlet containing a triangle is the tailed-triangle tail-edge orbit shown in Figure~\ref{fig:tailed-triangle-tail-orbit-Si}. 
As an aside, Figure~\ref{fig:tailed-triangle-tail-orbit-Si} is for the tailed-triangle tail-edge orbits centered at node $i$; however, the other tailed-triangle tail-edge orbit centered at node $j$ is also computed.
Observe that this is the only orbit needed to derive the typed 4-stars in constant time as shown in Eq.~\ref{eq:typed-4-star}.

\medskip\noindent\textbf{Typed Tailed-Triangle Tri-Edge Orbit Count (Figure~\ref{fig:tailed-tri-tri-edge-orbit-Si})}: 
\begin{equation} \label{eq:typed-4-tailed-triangle-triangle-edge-orbit}
f_{ij}(g_9,\vt) = 
\begin{cases}
\big(|T_{ij}^{t}| \cdot (|S_{i}^{t}| + |S_{j}^{t}|)\big) - f_{ij}(g_{10}, \vt) 							 			 & \text{if } t=t^{\prime} \\[5pt]
\big(|T_{ij}^{t}| \cdot (|S_{i}^{t^{\prime}}| + |S_{j}^{t^{\prime}}|)\big) \; + &  \text{otherwise} \\
\big(|T_{ij}^{t^{\prime}}| \cdot (|S_{i}^{t}| + |S_{j}^{t}|)\big) - f_{ij}(g_{10}, \vt)	 & \\[2pt]
\end{cases}
\end{equation}\noindent
where $f_{ij}(g_{10},\vt)$ is the chordal-cycle edge orbit count for edge $(i,j) \in E$ with type vector $\vt$.

\medskip\noindent\textbf{Typed Chordal-Cycle Center Orbit Count (Figure~\ref{fig:chordal-cycle-center-orbit})}: 
\begin{equation} \label{eq:typed-4-chordal-cycle-center-orbit}
f_{ij}(g_{11},\vt) = 
\begin{cases}
\mychoose{|T_{ij}^{t}|}{2} - f_{ij}(g_{12}, \vt) 				 			 	& \text{if } t=t^{\prime} \\[5pt]
\big(|T_{ij}^{t}| \cdot |T_{ij}^{t^{\prime}}|\big) - f_{ij}(g_{12}, \vt)	 	& \text{otherwise} \\[2pt]
\end{cases}
\end{equation}\noindent
where $f_{ij}(g_{12},\vt)$ is the typed 4-clique count for edge $(i,j) \in E$ with type vector $\vt$.

\medskip\noindent\textbf{Discussion}:
Many other combinatorial relationships can be derived in a similar fashion.
The equations shown above are the ones required such that the worst-case time complexity matches that of the best known untyped graphlet counting algorithm (See Section~\ref{sec:complexity-analysis}).

\subsection{From Typed Orbits to Graphlets}
Counts of the \emph{typed graphlets} for each edge $(i,j) \in E$ can be derived from the \emph{typed graphlet orbits} using the following equations:
\begin{align}
& f_{ij}(h_{3}, \vt) = f_{ij}(g_{3}, \vt) + f_{ij}(g_{4}, \vt) \\
& f_{ij}(h_{4}, \vt) = f_{ij}(g_{5}, \vt) \\
& f_{ij}(h_{5}, \vt) = f_{ij}(g_{6}, \vt) \\
& f_{ij}(h_{6}, \vt) = f_{ij}(g_{7}, \vt) + f_{ij}(g_{8}, \vt) + f_{ij}(g_{9}, \vt) \\
& f_{ij}(h_{7}, \vt) = f_{ij}(g_{10}, \vt) + f_{ij}(g_{11}, \vt) \\
& f_{ij}(h_{8}, \vt) = f_{ij}(g_{12}, \vt) 
\end{align}\noindent
where $h_{}$ is the graphlet without considering the orbit (Table~\ref{table:typed-graphlet-equations}).

{
\algblockdefx[parallel]{parfor}{endpar}[1][]{$\textbf{parallel for}$ #1 $\textbf{do}$}{$\textbf{end parallel}$}
\algrenewcommand{\alglinenumber}[1]{\fontsize{6.5}{7}\selectfont#1}
\begin{figure}[h!]
\begin{center}
\begin{algorithm}[H]
\caption{\,\small
Update Typed Graphlets.
Add typed graphlet (with id $\hash$) to $\mathcal{M}_{ij}$ if $\hash \not\in \mathcal{M}_{ij}$ and increment $\vx_{\hash}$ (frequency of that typed graphlet for a given edge).
}
\label{alg:updated-typed-motifs}
\begin{spacing}{1.2}
\fontsize{8}{9}\selectfont
\begin{algorithmic}[1]
\Procedure {Update}{$\vx$, $\mathcal{M}_{ij}$, $\hash = \mathbb{F}\,(g, \Phi_i, \Phi_j, \Phi_k, \Phi_r)$}
\If{$\hash \not\in \mathcal{M}_{ij}$} \label{algline:check-if-typed-motif-already-added}
\State $\mathcal{M}_{ij} \leftarrow \mathcal{M}_{ij} \cup \{\hash\}$
\State $\vx_{\hash} = 0$
\EndIf
\State $\vx_{\hash} = \vx_{\hash} + 1$ \label{algline:update-typed-motif-frequency}

\State {\bf return} updated set of typed graphlets $\mathcal{M}_{ij}$ and their counts $\vx$
\EndProcedure
\smallskip
\end{algorithmic}
\end{spacing}
\vspace{-1.mm}
\end{algorithm}
\end{center}
\vspace{-5.mm}
\end{figure}
}

\subsection{Typed Motif Hash Functions} \label{sec:typed-motif-hash-function}
Given a general heterogeneous graph with $L$ unique types such that $L<10$, then a simple and efficient typed motif hash function $\mathbb{F}$ is defined as follows:
\begin{equation}\label{eq:simple-typed-motif-hash-function}
\mathbb{F}(g,\vt) = g10^4 + t_1 10^3 + t_2 10^2 + t_3 10^1 + t_4
\end{equation}\noindent
where $g$ encodes the $k$-node motif orbit (\eg, 4-path center)
and $t_1$, $t_2$, $t_3$, $t_4$ encode the type of the nodes in $H \in \mathcal{H}$ with type vector $\vt = \big[ t_1 \; t_2 \; t_3 \; t_4 \big]$.
Since the maximum hash value resulting from Eq.~\ref{eq:simple-typed-motif-hash-function} is small (and fixed for any arbitrarily large graph $G$), we can leverage a perfect hash table to allow for fast $o(1)$ constant time lookups to determine if a typed motif was previously found or not as well as updating the typed motif count in $o(1)$ constant time.
For $k$-node motifs where $k<4$, we simply set the last $4-k$ types to $0$.
Note the simple typed motif hash function defined above can be extended trivially to handle graphs with $L \geq 10$ types as follows:
\begin{equation} \label{eq:simple-typed-motif-hash-function-10-or-more}
\mathbb{F}(g,\vt) = g10^8 + t_1 10^6 + t_2 10^4 + t_3 10^2 + t_4
\end{equation}
In general, any non-cryptographic hash function $\mathbb{F}$ can be used (see~\cite{chi2017hashing} for some other possibilities).
Thus, the approach is independent of $\mathbb{F}$ and can always leverage the best known hash function.
The only requirement of the hash function is that it is invertible $\mathbb{F}^{-1}$.

Thus far we have not made any assumption on the ordering of types in $\vt$.
As such, the hash function $\mathbb{F}$ discussed above can be used directly in the framework for counting typed graphlets such that the type structure and position are preserved.
However, since we are interested in counting all typed graphlets $\wrt$ Definition~\ref{def:typed-graphlet-instance}, 
then we map all such orderings of the types in $\vt$ to the same hash value using a precomputed hash table.
This allows us to obtain the unique hash value in $o(1)$ constant time for any ordering of the types in $\vt$.
In our implementation, we compute $s = t_1 10^3 + t_2 10^2 + t_3 10^1 + t_4$ and then use $s$ as an index into the precomputed hash table to obtain the unique hash value $c$ in $o(1)$ constant time.

\begin{figure}[h!]
\centering
\includegraphics[width=0.6\linewidth]{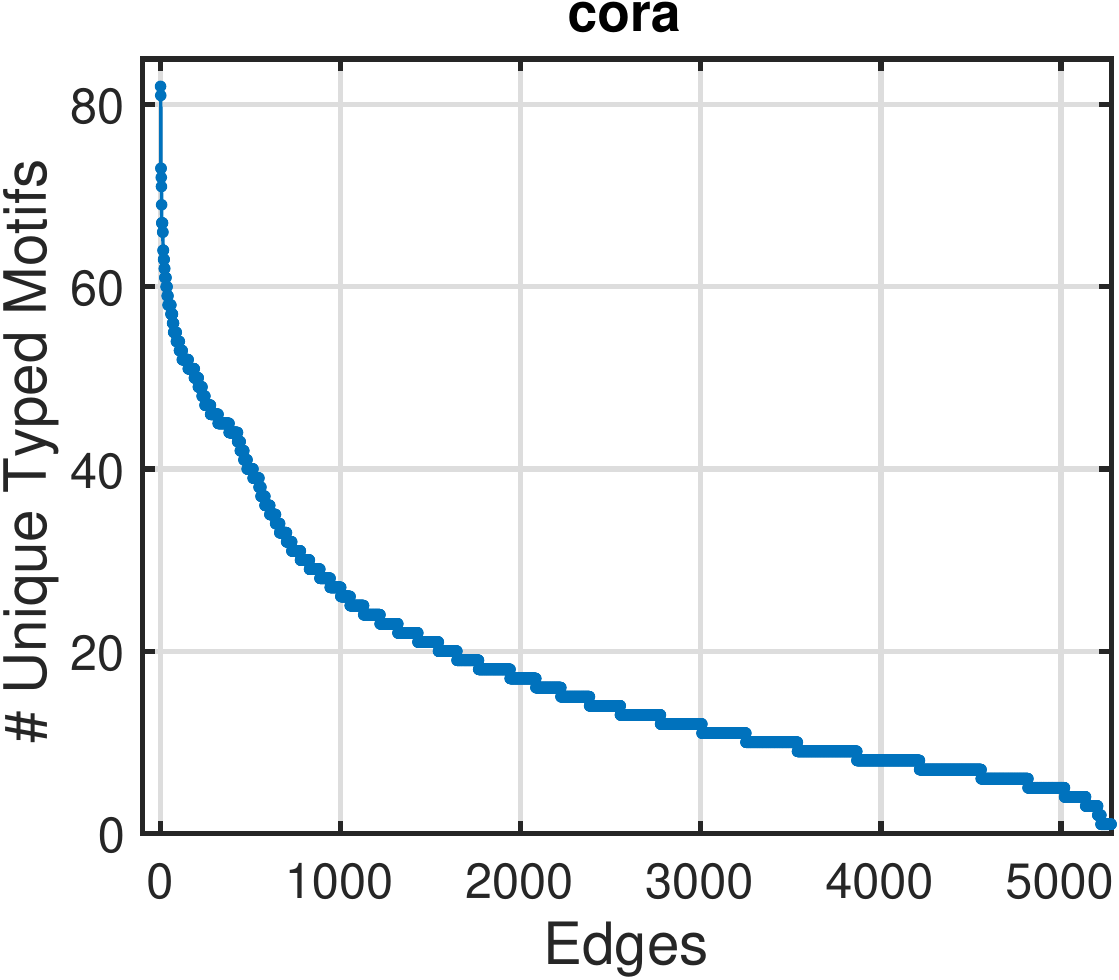}
\caption{
Distribution of unique typed motifs that occur on the edges.
This experiment considers all typed motifs of $\{3,4\}$-nodes.
Among the 1428 possible unique typed motifs that could arise in $G$, there are only 876 unique typed motifs that actually occur (at least once at an edge in $G$).
Even more striking, the maximum unique typed motifs that occur on any edge in $G$ (cora) is only 82.
Overall, the mean number of unique typed motifs over all edges in $G$ is 17,  \ie, only about 1.1\% of the possible typed motifs.
These results indicate the significance of only a few typed motifs as the vast majority of the typed motif counts for any arbitrary edge is zero.
Thus, the space required by the approach is nearly-optimal.
}
\label{fig:unique-typed-motif-edges-cora}
\end{figure}

\subsection{Sparse Typed Motif Format} \label{sec:sparse-typed-motif-format}
This section describes a space-efficient representation for typed motifs based on a key observation.
\begin{Property}\label{prop:small-frac-typed-motifs}
Let $T$ denote the number of \emph{unique typed motifs} that appear in an arbitrary graph $G$ with $L$ types.
Assuming the graph $G$ has a skewed degree distribution,
then most edges in $G$ appear in only a small fraction of the $T$ actual typed motifs that can occur.
\end{Property}\noindent
This property is shown empirically in Figure~\ref{fig:unique-typed-motif-edges-cora}.
Property~\ref{prop:small-frac-typed-motifs} implies that using a $M \times T$ matrix to store the typed graphlet counts is far from optimal in terms of the space required due to most of the $T$ typed motif counts being zero for any given edge. 
Based on this observation, we describe an approach that uses near-optimal space by storing only the typed motifs with nonzero counts for each edge $(i,j) \in E$ in the graph.
Typed motif counts are stored in a sparse format since it would be impractical in large graphs to store all typed motifs as there can easily be hundreds of thousands depending on the number of types in the input graph. 
For each edge, we store only the nonzero typed motif counts along with the unique ids associated with them.
The unique ids allow us to map the nonzero counts to the actual typed motifs.
We also developed a space-efficient format for storing the resulting typed motif counts to disk.
Instead of using the typed motif hash as the unique id, we remap the typed motifs to smaller consecutive ids (starting from $1$) to reduce the space requirement even further.
Finally, we store a typed motif lookup table that maps a given motif id to its description and is useful for looking up the meaning of the typed motifs discovered.

\section{Global Typed Graphlet Counts} \label{sec:global-typed-graphlet-counting}
While Section~\ref{sec:framework} focused on counting typed graphlets locally for each edge $(i,j) \in E$ in $G$, one may also be interested in the total counts of each typed graphlet in $G$.
This is known as the global typed graphlet counting problem.
More formally,
\begin{Definition}[Global Typed Graphlet Counting] \label{def:global-typed-graphlet-counting}
Given a graph $G$ with $L$ types, the global typed graphlet counting problem is to find the set of all typed motifs that occur in $G$ along with their corresponding frequencies.
\end{Definition}

A general equation for solving the above problem for any arbitrary \emph{typed graphlet} $H$ is given below.
Let $H$ denote an arbitrary typed graphlet and $\vx$ be an $M$-dimensional vector of counts of $H$ for every edge $(i,j) \in E$, then the frequency of $H$ in $G$ is:
\begin{equation}\label{eq:global-typed-graphlet-counts}
C \,= \, \frac{1}{|E(H)|}\; \vx^{\!\top}\ve 
\end{equation}\noindent
where $|E(H)|$ is the number of edges in the typed graphlet $H$ and $\ve = [ \, 1\; \cdots \; 1 \,]$ is an $M$-dimensional vector of all 1's.
For instance, suppose $H$ is one of the typed 4-cliques in Figure~\ref{fig:typed-motifs-3colors}, then $|E(H)|=6$.
It is straightforward to see that if we were interested in \emph{only} deriving global counts of the typed graphlets, then we can significantly speedup the proposed algorithms by enforcing constraints to avoid counting a typed graphlet numerous times.
Instead, we can add constraints to count such motifs once per edge by enforcing an ordering.

\section{Parallelization} \label{sec:parallel}
We now describe a parallelization strategy for the proposed typed motif counting approach.
While our implementation uses shared memory, the parallelization is described generally such that it can be used with a distributed-memory architecture as well.
As such, our discussion is on the general scheme.

The parallel constructs we use are a worker task-queue and a global broadcast channel. 
Here, we assume that each worker has a copy of the graph and distribute edges to workers to find the typed graphlet counts that node $i$ and $j$ participate.
At this point, we view the main while loop as a task
generator and farm the current edge out to a worker to find the typed graphlet counts that co-occur between node $i$ and node $j$.
The approach is lock free since each worker uses the same motif hash function to obtain a unique hash value for every typed motif.
Thus, each worker can simply maintain the typed motifs identified and their counts for every edge assigned to it.
In our own shared memory implementation, we avoid some of the communications by using global arrays and avoiding locked updates to them by using a unique edge id.
Counting typed graphlets on the edges as opposed to the nodes also has computational advantages with respect to parallelization and in particular load balancing.
Let $x_i$ and $x_{ij}$ denote the node and edge count of an arbitrary graphlet $H$.
Since $|E| \gg |V|$ and $\sum_{i \in V} x_{i} = \sum_{(i,j) \in E} x_{ij}$, then $\frac{1}{|V|}\sum_{i \in V} x_{i} < \frac{1}{|E|}\sum_{(i,j) \in E} x_{ij}$.
Hence, more work per vertex is required than per edge.
Therefore, counting typed graphlets on the edges is guaranteed to have better load balancing than node-centric algorithms.

\section{Theoretical Analysis} \label{sec:complexity-analysis}
We first show that heterogeneous graphlets contain more information than untyped graphlets.
Next, Section~\ref{sec:time-complexity} provides the worst-case time complexity of the proposed approach whereas Section~\ref{sec:space-complexity} gives the space complexity.
See Table~\ref{table:notation} for a summary of notation.

First, we show the relationship between the count of an untyped motif $H$ in $G$ and the count of all typed motifs in $G$ with induced subgraph $H$.
\begin{Proposition} \label{prop:untyped-graphlet-equal-sum-of-all-typed-graphlets}
Let $\vx$ denote the vector of counts for any untyped network motif $H \in \mathcal{H}$ (\eg, 4-cycle).
Further, let $\mX$ denote a $M \times T_{H}$ matrix of typed motif counts for motif $H$ where $T_{H}$ denotes the number of typed network motifs that arise from $L$ types.
Then the following holds:
\begin{equation}
C = \sum_{i=1}^{M} x_i = \sum_{i=1}^{M} \sum_{j=1}^{T_{H}} X_{ij}
\end{equation}
\end{Proposition}\noindent
Consider the counts of an untyped graphlet for a single edge.
The above demonstrates that these counts are partitioned among the set of typed graphlets that arise from the untyped graphlet when types are considered.

We now show that typed network motifs contain more information than untyped motifs.
Entropy is a measure of the \emph{average information content} of a state~\cite{kullback1997information}.
Let $\vp \in \RR^{T_{H}}$ denote a typed motif probability distribution ($\vp^T \ve = 1$), the entropy 
(average information content) of $\vp$ is
$\mathbb{H}(\vp) = -\sum_{i} \; p_i \log p_i$.
Hence, $\mathbb{H}(\vp)$ quantifies the amount of information in the relative frequencies of the typed network motifs (of a given motif $H \in \mathcal{H}$).
In the case of untyped motifs, the $C = \sum_{i=1}^{M} x_i$ untyped motifs are assumed to belong to a single homogeneous motif where all nodes are of the same type.
This matches exactly the information we have if types are not considered.
\begin{Proposition} \label{prop:info-gain}
Assume $\vp \in \RR^{T_{H}}$ is an arbitrary typed network motif probability distribution such that $p_i<1$, $\forall i$ 
and $\vq$ is the untyped motif distribution where $q_i=1$ and $q_j=0, \forall j\not=i$, then $\mathbb{H}(\vp) > \mathbb{H}(\vq)$.
\end{Proposition}\noindent
This implies the typed motif distribution $\vp$ contains more information than the untyped motif distribution $\vq$.
The proof is straightforward.

\subsection{Time Complexity} \label{sec:time-complexity}
\begin{Property} \label{prop:1}
\begin{equation}
d_i + d_j = 2|T_{ij}| + |S_i| + |S_{j}|
\end{equation}\noindent
where $d_i = |\Gamma_i|$,\; $d_j = |\Gamma_j|$,\; $T_{ij} = \Gamma_i \cap \Gamma_j$,\; $S_i = \Gamma_i \setminus T_{ij}$,\; and $S_j = \Gamma_j \setminus T_{ij}$.
\end{Property}

\begin{Property} \label{prop:2}
The space required to store $T_{ij}$, $S_{i}$, and $S_{j}$ is less than $d_i+d_j$ iff $|T_{ij}|>0$.
\end{Property}
This is straightforward to see since $|S_{i}|+|S_{j}|=|S_{i} \cup S_{j}|$ always holds. 
However, if $|T_{ij}|=0$, then $|S_{i}|+|S_{j}|=d_i+d_j$.
Hence, triangles represent the smallest clique, and as shown in~\cite{rossi2018compressing-graphs-cliques} can be used to compress the graph.
As the density of the graph increases, more triangles are formed, and therefore less space is used.
Notice that the worst case is also unlikely to occur because of this fact.
For instance, suppose $d_i = \Delta$, $d_j = \Delta$, and $\Delta = n$ (worst case), then $|T_{ij}|=d_i=d_j$, and $|S_{i}|=0$, $|S_{j}|=0$.
Furthermore, if $|S_{i}|=n$, then $|S_{j}|=0$ and $|T_{ij}|=0$ must hold.
Obviously, if node $i$ is connected to all $n$ nodes, then any node $k \in \Gamma_{j}$ must form a triangle with $i$ ($k \in T_{ij}$).

\subsubsection{Typed 3-Node Graphlets}
\begin{thm}\label{lem:time-complexity-3-node-graphlets}
The overall worst-case time complexity for counting all 3-node typed graphlets for a given edge $(i,j) \in E$ is:
\begin{equation}
\mathcal{O}(2|\Gamma_i| + |\Gamma_j|) = \mathcal{O}(\Delta)
\end{equation}\noindent
where $|\Gamma_i|$ and $|\Gamma_j|$ denote the number of nodes connected to node $i$ and $j$, respectively.
Further, $\Delta$ is the maximum degree in $G$.
\end{thm} 

\begin{proof}
It takes at most $\mathcal{O}(|\Gamma_i|)$ time for Line~\ref{algline:graphlet-create-hash} and Line~\ref{algline:S_i}-\ref{algline:typed-3-path-centered-i-motif-hash} whereas the time complexity for Line~\ref{algline:triangles-and-wedges}-\ref{algline:typed-3-path-centered-j-motif-hash} is $\mathcal{O}(|\Gamma_j|)$.
Obviously, since $\Gamma_i$ is enumerated twice, then we can always ensure $i$ is the node with smallest degree (\ie, $d_i \leq d_j$).
The worst-case arises when $|\Gamma_i|$ and $|\Gamma_j|$ is the maximum degree of a node in $G$ denoted as $\Delta$.
\end{proof}

\subsubsection{Typed 4-Node Graphlets}
We first provide the time complexity of deriving path-based and triangle-based graphlet orbits in Lemma~\ref{lem:time-complexity-path-based}-\ref{lem:time-complexity-triangle-based}, and then give the total time complexity of all 3 and 4-node typed graphlets in Theorem~\ref{thm:time-complexity-4-node-graphlets} based on these results.
Note that Lemma~\ref{lem:time-complexity-path-based}-\ref{lem:time-complexity-triangle-based} includes the time required to derive all typed 3-node typed graphlets.

\begin{lemma} \label{lem:time-complexity-path-based}
For a single edge $(i,j) \in E$, the worst-case time complexity for deriving all \emph{typed path-based graphlet orbits} is:
\begin{align} \label{eq:time-complexity-path-based}
\mathcal{O}\Big(\Delta \big(|S_i| + |S_j| \big) \Big) 
\end{align}\noindent
\end{lemma}\noindent
Note $|S_{i}|\Delta \geq \sum_{k \in S_{i}} d_k$ and $|S_{j}|\Delta \geq \sum_{k \in S_{j}} d_k$.

\begin{lemma} \label{lem:time-complexity-triangle-based}
For a single edge $(i,j) \in E$, the worst-case time complexity for deriving all \emph{typed triangle-based graphlet orbits} is:
\begin{equation} \label{eq:time-complexity-triangle-based}
\mathcal{O}\big(\Delta |T_{ij}| \big) 
\end{equation}\noindent
\end{lemma}\noindent
Notice $|T_{ij}|\Delta \geq |T_{ij}|\Delta_T \geq \sum_{k \in T_{ij}} d_k$ where $\Delta$ is the maximum degree of a node in $G$ and $\Delta_T$ is the maximum degree of a node in $T_{ij}$.
Thus, $|T_{ij}|\Delta$ only occurs iff $\forall k \in T_{ij}$, $d_k = \Delta$ where $\Delta = $ maximum degree of a node in $G$.
In sparse real-world graphs, $T_{ij}$ is likely to be smaller than $S_i$ and $S_j$ as triangles are typically more rare than 3-node paths.
Conversely, $T_{ij}$ is also more likely to contain high degree nodes, as nodes with larger degrees are obviously more likely to form triangles than those with small degrees.

From Lemma~\ref{lem:time-complexity-path-based}-\ref{lem:time-complexity-triangle-based}, we have the following:
\begin{thm} \label{thm:time-complexity-4-node-graphlets}
For a single edge $(i,j) \in E$, the worst-case time complexity for deriving all 3 and 4-node typed graphlet orbits is:
\begin{equation} \label{eq:time-complexity-fast-alg-overall}
\mathcal{O}\big(\Delta \big(|S_i| + |S_j| + |T_{ij}|\big)\big) 
\end{equation}\noindent
\end{thm}\noindent

\begin{proof}
The time complexity of each step is provided below.
Hashing all neighbors of node $i$ takes $\mathcal{O}(|\Gamma_i|)$.
Recall from Lemma~\ref{lem:time-complexity-3-node-graphlets} that counting all 3-node typed graphlets takes $\mathcal{O}(2|\Gamma_i| + |\Gamma_j|) = \mathcal{O}(\Delta)$ time for an edge $(i,j) \in E$.
This includes the time required to derive the number of typed 3-node stars and typed triangles for all types $t=1,\ldots,L$.
This information is needed to derive the remaining typed graphlet orbit counts in constant time.
Next, Algorithm~\ref{alg:typed-path-based-motifs-exact} is used to derive a few path-based typed graphlet orbit counts taking $\mathcal{O}(\Delta (|S_i| + |S_j|))$ time in the worst-case.
Similarly, Algorithm~\ref{alg:typed-triangle-based-motifs-exact} is used to derive a few triangle-based typed graphlet orbit counts taking in the worst-case $\mathcal{O}(\Delta |T_{ij}|)$ time.
As an aside, updating the count of a typed graphlet count is $o(1)$ (Algorithm~\ref{alg:updated-typed-motifs}).

Now, we derive the remaining typed graphlet orbit counts in constant time (Line~\ref{algline:main-alg-for-type-pair}-\ref{algline:main-alg-derive-remaining-typed-graphlet-orbits-constant-time}).
Since each type pair leads to different typed graphlets, we must iterate over at most $L(L-1)/2+L$ type pairs.
For each pair of types selected, we derive the typed graphlet orbit counts in $o(1)$ constant time via Eq.~\ref{eq:typed-4-path-center-orbit}-\ref{eq:typed-4-chordal-cycle-center-orbit} (See Line~\ref{algline:main-alg-for-type-pair}-\ref{algline:main-alg-derive-remaining-typed-graphlet-orbits-constant-time}).
Furthermore, the term involving $L$ is for the worst-case when there is at least one node in all $L$ sets (\ie, at least one node of every type $L$).
Nevertheless, since $L$ is a small constant, $L(L-1)/2+L$ is negligible.

For a single edge, the worst-case time complexity is $\mathcal{O}(\Delta(|S_i|+|S_j|+|T_{ij}|))$.
Let $\bar{T}$ and $\bar{S}$ denote the average number of triangle and 3-node stars incident to an edge in $G$.
More formally, $\bar{T} = \frac{1}{M}\sum_{(ij) \in E} |T_{ij}|$ and $\bar{S} = \frac{1}{M} \sum_{(ij) \in E} |S_i|+|S_j|$.
The total worst-case time complexity for all $M$ edges is $\mathcal{O}(M\Delta(\bar{S}+\bar{T}))$.
Note that obviously $\bar{S}M = \sum_{(ij) \in E} |S_i|+|S_j|$ and $\bar{T}M = \sum_{(ij) \in E} |T_{ij}|$.
\end{proof}

\begin{cor}\label{lem:time-complexity}
The worst-case time complexity of counting typed graphlets using Algorithm~\ref{alg:typed-motifs-exact} matches the worst-case time complexity of the best known untyped graphlet counting algorithm.
\end{cor}

\begin{proof}
From Theorem~\ref{thm:time-complexity-4-node-graphlets} we have that $\mathcal{O}\big(\Delta \big(|S_i| + |S_j| + |T_{ij}|\big)\big)$, which is exactly the time complexity of the best known untyped (homogeneous) graphlet counting algorithm~\cite{pgd,pgd-kais}.
\end{proof}

\subsection{Space Complexity} \label{sec:space-complexity}
Since our approach generalizes to graphs with an arbitrary number of types $L$, the specific set of typed motifs is unknown.
As demonstrated in Table~\ref{table:typed-graphlets-example}, it is impractical to store the counts of all possible $k$-node typed motifs for any graph of reasonable size as typically done in traditional methods for untyped graphlets~\cite{pgd,pgd-kais,rage}.

Despite this being obviously impractical due to the amount of space that would be required,
the existing state-of-the-art methods such as GC~\cite{gu2018heterAlignment} store counts of all possible typed graphlets, and therefore the space complexity of such methods is:
\begin{equation} \label{eq:space-complexity-of-other-methods}
\mathcal{O}(MT_{\max})
\end{equation}\noindent
where $M=|E|$ is the number of edges in $G$ and $T_{\max}$ is the number of different possible typed graphlets with $L$ types.
Thus, $MT_{\max}$ is the total space to store $M$ vectors of length $T_{\max}$ (\ie, one $T_{\max}$-dimensional vector per edge).
To understand the above space requirements and how it is impractical for any moderately sized graph, 
suppose we have a graph with 10,000,000 ($M$) edges and $L=7$ types.
Counting all 3- and 4-node typed graphlet \emph{orbits} for every edge would require $90.72$ GB of space to store the large $MT_{\max}$ matrix (assuming 4 bytes per count/entry).
This is obviously impractical for any graph of even moderate size.

In contrast, the total space used by our approach for storing the typed graphlet counts is:
\begin{equation} \label{eq:total-actual-space}
\mathcal{O}(M\bar{T}) \; \ll \; \mathcal{O}(MT_{\max})
\end{equation}\noindent
where $\bar{T} = \frac{1}{M} \sum_{(ij) \in E} |\mathcal{X}_{ij}|$ is the average number of typed graphlets with nonzero counts per edge.
Note $|\mathcal{X}_{ij}|$ is the number of typed graphlets with nonzero counts for edge $(ij) \in E$.
Therefore, the total space is only $\mathcal{O}(M\bar{T})$.
The space of all other data structures used in Algorithm~\ref{alg:typed-motifs-exact} is small in comparison, \eg, $\Psi$ takes at most $\mathcal{O}(|V|)$ space, whereas $T_{ij}$, $S_{i}$, and $S_{j}$ take $\mathcal{O}(\Delta)$ space in the worst-case (by Property~\ref{prop:relationship-between-typed-sets-and-untyped-sets}) and can be reused for every edge.
In addition, the size of $\vx$ is independent of the graph size ($|V|+|E|$) and can also be reused.

\begin{table*}[t!]
\centering
\caption{Results comparing the proposed approach to the state-of-the-art methods in terms of runtime performance (seconds).
Since existing state-of-the-art methods are unable to handle large or even medium-sized graphs as shown below, we had to include a number of very small graphs (\eg, cora, citeseer, webkb) in order to compare with existing methods.
Note $\Delta=$ max node degree; $|\mathcal{T}_V|=$ \# of node types; $|\mathcal{T}_E|=$ \# of edge types.
}
\vspace{-3mm}
\label{table:runtime-perf}
\small
\begin{tabularx}{1.0\linewidth}{@{}
r H lllH @{}cc HH
rrr Hr r rrr H}
\toprule
&&&&&&&
&&&
\multicolumn{5}{c}{\sc seconds} &&
\multicolumn{3}{c}{\sc speedup (ours vs.)}
\\
\cmidrule(l{0pt}r{5pt}){11-15}
\cmidrule(l{15pt}r{5pt}){16-19}

& 
& 
$|V|$ & $|E|$ & $\Delta$ & 
&
$|\mathcal{T}_V|$ & 
$|\mathcal{T}_E|$ 
& &&
\;\textbf{GC}\!~\cite{gu2018heterAlignment} &
\;\textbf{ESU}\!~\cite{fanmod} &
\textbf{G-Tries}\!~\cite{ribeiro2014discovering} &
& 
\;\; \textbf{Ours} 
&\;\;&
\textbf{GC}\!~\cite{gu2018heterAlignment} &
\textbf{ESU}\!~\cite{fanmod} &
\textbf{G-Tries}\!~\cite{ribeiro2014discovering} &
\\
\midrule

\textsf{citeseer}
& 
& 3.3k & 4.5k & 99 & 
& 6 & 21
&      && 
46.27 & 
5937.75 &
144.08 & 
& 
\textbf{0.022} && 
2103x & 
269897x  & 
6549x & 
\\

\textsf{cora}
& 
& 2.7k & 5.3k & 168 & 
& 7 & 28
&      && 
467.20 & 
10051.07 & 
351.40 & 
& 
\textbf{0.032} && 
14600x & 
314095x  & 
10981x & 
\\

\textsf{fb-relationship} 
& 
& 7.3k & 44.9k & 106 & 
& 6 & 20
&      && 
1374.60 & 
54,837.69 & 
3789.17 & 
& 
\textbf{0.701} && 
1960x & 
78227x & 
5405x & 
\\

\textsf{web-polblogs}
& 
& 1.2k & 16.7k & 351 & 
& 2 & 1
&      && 
28,986.70 & 
26,577.10 & 
1,563.04 & 
& 
\textbf{1.055} && 
27475x & 
25191x & 
1481x & 
\\

\textsf{ca-DBLP}
& 
& 2.9k & 11.3k & 69 &
& 3 & 3
&      && 
149.20 & 
1,188.11 & 
18.90 & 
& 
\textbf{0.100} && 
1492x & 
11881x & 
189x & 
\\

\textsf{inf-openflights}
& 
& 2.9k & 15.7k & 242 & 
& 2 & 2
&      && 
9262.20 & 
18,839.36 & 
458.01 & 
& 
\textbf{0.578} && 
16024x & 
32594x & 
792x & 
\\

\textsf{soc-wiki-elec}
& 
& 7.1k & 100.8k & 1.1k & 
& 2 & 2
&      && 
ETL & 
ETL & 
26,468.85 & 
& 
\textbf{5.316} && 
$\infty$ & 
$\infty$ & 
45793x & 
\\

\textsf{webkb}
& 
& 262 & 459 & 122 &  
& 5 & 14
&      && 
85.82 & 
7,158.10 & 
187.22 & 
& 
\textbf{0.006} && 
14303x & 
1193016x & 
31203x & 
\\

\textsf{terrorRel}
& 
& 881 & 8.6k & 36 & 
& 2 & 3
&      && 
192.6 & 
3130.7 & 
241.1 & 
& 
\textbf{0.039} && 
4938x & 
80274x & 
6182x & 
\\

\textsf{pol-retweet} 
& 
& 18.5k & 48.1k & 786 & 
& 2 & 3
&      && 
ETL & 
ETL & 
ETL & 
&
\textbf{0.296} && 
$\infty$ & 
$\infty$ & 
$\infty$ & 
\\

\textsf{web-spam}
&
& 9.1k & 465k & 3.9k & 
& 3 & 6
&      && 
ETL & 
ETL & 
ETL & 
&
\textbf{210.97} && 
$\infty$ & 
$\infty$ & 
$\infty$ & 
\\

\textsf{movielens}
& 
& 28.1k & 170.4k & 3.6k & 
& 3 & 3
&      && 
ETL & 
ETL & 
ETL & 
& 
\textbf{5.23} && 
$\infty$ & 
$\infty$ & 
$\infty$ & 
\\

\textsf{citeulike} 
& 
& 907.8k & 1.4M & 11.2k & 
& 3 & 2 
&      && 
ETL & 
ETL & 
ETL & 
& 
\textbf{126.53} &&
$\infty$ & 
$\infty$ & 
$\infty$ & 
\\

\textsf{yahoo-msg} 
& 
& 100.1k & 739.8k & 9.4k & 
& 2 & 2
&      && 
ETL & 
ETL & 
ETL & 
&
\textbf{35.22} && 
$\infty$ & 
$\infty$ & 
$\infty$ & 
\\

\textsf{dbpedia} 
& 
& 495.9k & 921.7k & 24.8k & 
& 4 & 3 
&      && 
ETL & 
ETL & 
ETL & 
& 
\textbf{56.02} &&
$\infty$ & 
$\infty$ & 
$\infty$ & 
\\

\textsf{digg} 
& 
& 217.3k & 477.3k & 219 & 
& 2 & 2 
&      && 
ETL & 
ETL & 
ETL & 
& 
\textbf{5.592} && 
$\infty$ & 
$\infty$ & 
$\infty$ & 
\\

\textsf{bibsonomy} 
& 
& 638.8k & 1.2M & 211 & 
& 3 & 3 
&      && 
ETL & 
ETL & 
ETL & 
& 
\textbf{3.631} && 
$\infty$ & 
$\infty$ & 
$\infty$ & 
\\

\textsf{epinions}
& 
& 658.1k & 2.6M & 775 &
& 2 & 2
&      && 
ETL & 
ETL & 
ETL & & 
\textbf{85.27} &&
$\infty$ & 
$\infty$ & 
$\infty$ & 
\\

\textsf{flickr} 
& 
& 2.3M & 6.8M & 216 & 
& 2 & 2 
&      && 
ETL & 
ETL & 
ETL & 
& 
\textbf{120.79} &&
$\infty$ & 
$\infty$ & 
$\infty$ & 
\\

\textsf{orkut}
& 
& 6M & 37.4M & 166 & 
& 2 & 2
&      && 
ETL & 
ETL & 
ETL & 
& 
\textbf{1241.01} && 
$\infty$ & 
$\infty$ & 
$\infty$ & 
\\

\midrule

\textsf{ER (10K,0.001)}
& 
& 10k & 50.1k & 26 & 
& 5 & 15
&      && 
183.32 & 
5,399.14 & 
241.27 & 
& 
\textbf{0.48} && 
381x & 
11248x & 
502x & 
\\

\textsf{CL (1.8)}
& 
& 9.2k & 44.2k & 218 & 
& 5 & 15
&      && 
31,668 & 
45,399.14 & 
5,241.27 & 
& 
\textbf{1.46} && 
21690x & 
31095x & 
3589x & 
\\

\textsf{KPGM (log 12,14)}
& 
& 3.3k & 43.2k & 1.3k & 
& 5 & 15
&       && 
ETL & 
ETL & 
63,843.86 & 
& 
\textbf{8.94} && 
$\infty$ & 
$\infty$ & 
7141x & 
\\

\textsf{SW (10K,6,0.3)}
& 
& 10k & 30k & 12 & 
& 5 & 15
&    && 
21.48 & 
5,062.67 & 
206.92 & 
& 
\textbf{0.24} && 
89x & 
21094x & 
862x & 
\\

\bottomrule
\multicolumn{7}{l}{\footnotesize $^{*}$ ETL = Exceeded Time Limit (24 hours / 86,400 seconds)} 
\\
\end{tabularx}
\end{table*}

\section{Experiments} \label{sec:exp}
The experiments are designed to investigate the runtime performance (Section~\ref{sec:exp-comparison}), space-efficiency (Section~\ref{sec:exp-space-efficiency}), parallelization (Section~\ref{sec:exp-parallel-scaling}), and scalability (Section~\ref{sec:exp-scalability}) of the proposed approach.
From a computational point of view, any practical approach for this problem must have all four desired properties, \ie, fast, space-efficient, parallel with near-linear scaling, and scalable for large networks. 
To demonstrate the \emph{effectiveness} of the approach, we use a variety of heterogeneous and labeled/attributed graph data from different application domains.
All data can be accessed at NetworkRepository~\cite{nr}.
Unless otherwise mentioned, we use a serial version of the proposed typed graphlet approach for comparison with other methods.
In addition, our approach counts the frequency of all typed $\{2,3,4\}$-node graphlets globally for the entire graph and locally for every edge (\ie, local typed graphlet counting~\cite{rossi18tnnls}).
This is in contrast to the existing methods that focus on counting graphlets for every node as opposed to every edge.
As an aside, existing methods are serial (as they are difficult to parallelize effectively) and typically only solve the global or local subgraph counting problem.
All methods used for comparison are configured to count only 4-node graphlets.

\subsection{Runtime Comparison} \label{sec:exp-comparison}
We first demonstrate how fast the proposed framework is for deriving typed graphlets by comparing the runtime (in seconds) of our approach against ESU (using fanmod)~\cite{fanmod}, G-Tries~\cite{ribeiro2014discovering}, and GC~\cite{gu2018heterAlignment}.
As an aside, the proposed framework is the only approach that derives typed graphlets for every edge (as opposed to counting typed graphlets for every node as done by the existing methods).
Other key differences are summarized in Section~\ref{sec:related-work}.

For comparison, we use a wide variety of real-world graphs from different domains as well as graphs generated from four different synthetic graph models.
In Table~\ref{table:runtime-perf}, we report the time (in seconds) required by each method.
The results reported in Table~\ref{table:runtime-perf} are from a serial version of our approach.
The network statistics and properties of the graphs used for evaluation are also shown in Table~\ref{table:runtime-perf} (columns 2-6).
To be able to compare with the existing methods, we included a variety of very small graphs for which the existing methods could solve in a reasonable amount of time.
Note ETL indicates that a method did not terminate within 24 hours for that graph.
Strikingly, the existing methods are unable to handle medium to large graphs with hundreds of thousands or more nodes and edges as shown in Table~\ref{table:runtime-perf}.
Even small graphs can take hours to finish using existing methods (Table~\ref{table:runtime-perf}).
For instance, the small citeseer graph with only 3.3k nodes and 4.5k edges takes 46.27 seconds using the best existing method whereas ours finishes in a tiny fraction of a second, notably, $\nicefrac{2}{100}$ seconds.
This is about 2,100 times faster than the next best method.
Similarly, on the small cora graph with 2.7K nodes and 5.3K edges, GC takes 467 seconds whereas G-Tries takes 351 seconds.
However, our approach finishes in only 0.03 seconds.
This is 10,000 times faster than the next best method.
Unlike existing methods, our approach is able to handle large-scale graphs.
On flickr, our approach takes about 2 minutes to count the occurrences of all typed graphlets for all 6.8 million edges
Across all graphs, the proposed method achieves significant speedups over the existing state-of-the-art methods as shown in Table~\ref{table:runtime-perf}.
These results demonstrate the effectiveness of our approach for counting typed graphlets in large networks.

\newcommand{\GraphletFigScale}{0.05}
\begin{table}[h!]
\centering
\renewcommand{\arraystretch}{1.1} 
\caption{Comparing the number of unique typed motifs that occur for each induced subgraph (\eg, there are 40 different typed triangles that appear in citeseer where each typed triangle count has a different type/color configuration).
}
\label{table:unique-typed-motif-occur}
\vspace{-3mm}
\fontsize{7.5}{8.5}\selectfont
\begin{tabularx}{1.0\linewidth}{l HHHHHH XX XXXXXX HHH HHHHH@{}
}
\toprule
\textbf{Network data}  &   
$|V|$  &  $|E|$  &  $\Delta$ &&&&
\includegraphics[scale=0.8]{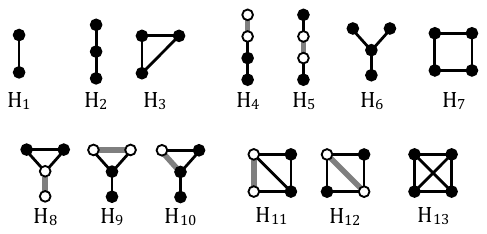} &
\includegraphics[scale=0.8]{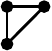} &
\includegraphics[scale=0.15]{fig9.pdf} &
\includegraphics[scale=0.8]{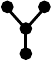} &
\includegraphics[scale=0.8]{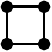} &
\includegraphics[scale=0.15]{fig12.pdf} &
\includegraphics[scale=0.14]{fig13.pdf} &
\includegraphics[scale=0.8]{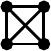} &&
\\
\midrule

\textsf{citeseer} & 3.3k & 4.5k & 99 & &&& 
56 & 40 & 124 & 119 & 66 & 98 & 56 & 19 & \\

\textsf{cora} & 2.7k & 5.3k & 168 & &&&
82 & 49 & 202 & 190 & 76 & 157 & 73 & 19 & \\

\textsf{fb-relationship} & 7.3k & 44.9k & 106 & &&&
50 & 47 & 112 & 109 & 85 & 106 & 89 & 77 & \\

\textsf{web-polblogs} & 1.2k & 16.7k & 351 & &&&
4 & 4 & 5 & 5 & 5 & 5 & 5 & 5 &  \\ 

\textsf{ca-DBLP} & 2.9k & 11.3k & 69 & &&& 
10 & 10 & 15 & 15 & 15 & 15 & 15 & 15 & \\

\textsf{inf-openflights} & 2.9k & 15.7k & 242 & &&&
4 & 4 & 5 & 5 & 5 & 5 & 5 & 5 & \\

\textsf{soc-wiki-elec} & 7.1k & 100.8k & 1.1k & &&&
4 & 4 & 5 & 5 & 5 & 5 & 5 & 5 & \\

\textsf{webkb} & 262 & 459 & 122 & &&&
31 & 21 & 59 & 59 & 23 & 51 & 32 & 8 & \\

\textsf{terrorRel} & 881 & 8.6k & 36 & &&&
4 & 4 & 5 & 0 & 4 & 5 & 5 & 5 & \\

\textsf{pol-retweet} & 18.5k & 48.1k & 786 & &&& 
4 & 4 & 5 & 5 & 5 & 5 & 5 & 4 & \\

\textsf{web-spam} & 9.1k & 465k & 3.9k & &&&
10 & 10 & 15 & 15 & 15 & 15 & 15 & 15 & \\

\textsf{movielens} & 28.1k & 170.4k & 3.6k & &&&
7 & 1 & 6 & 9 & 6 & 3 & 3 & 0 &  \\ 

\textsf{citeulike} & 907.8k & 1.4M & 11.2k & &&&
5 & 0 & 3 & 6 & 3 & 0 & 0 & 0 &  \\

\textsf{yahoo-msg} & 100.1k & 739.8k & 9.4k & &&&
3 & 2 & 3 & 4 & 3 & 3 & 3 & 2 &  \\

\textsf{dbpedia} & 495.9k & 921.7k & 24.8k & &&&
8 & 0 & 6 & 10 & 5 & 0 & 0 & 0 &  \\

\textsf{digg} & 253.2k & 1.2M & 533 & &&&
4 & 3 & 4 & 5 & 4 & 4 & 4 & 2 &  \\

\textsf{bibsonomy} & 638.8k & 1.2M & 211 & &&&
7 & 1 & 6 & 9 & 6 & 3 & 3 & 0 & \\

\textsf{epinions} & 658.1k & 2.6M & 775 & &&&
3 & 2 & 3 & 4 & 3 & 3 & 3 & 2 &  \\

\textsf{flickr} & 2.3M & 6.8M & 216 & &&&
3 & 2 & 3 & 4 & 3 & 3 & 3 & 2 &  \\

\textsf{orkut} & 6M & 37.4M & 166 & &&&
4 & 3 & 4 & 4 & 3 & 4 & 3 & 2 &  \\

\midrule
\textsf{ER (10K,0.001)} & 10k & 50.1k & 26 & &&&
35 & 30 & 70 & 70 & 69 & 66 & 1 & 0 & \\

\textsf{CL (1.8)} & 9.2k & 44.2k & 218 & &&&
35 & 35 & 70 & 70 & 70 & 70 & 70 & 68 & \\

\textsf{KPGM (log 12,14)} & 3.3k & 43.2k & 1.3k & &&&
35 & 35 & 70 & 70 & 70 & 70 & 70 & 70 & \\

\textsf{SW (10K,6,0.3)} & 10k & 30k & 12 & &&&
35 & 35 & 70 & 70 & 70 & 70 & 70 & 69 & \\

\bottomrule
\end{tabularx}
\end{table}

Table~\ref{table:unique-typed-motif-occur} provides the number of unique typed motifs that occur for each induced subgraph.
From these results, we make an important observation.
In real-world graphs we observe that certain typed motifs do not occur at all in the graph.
We define such typed motifs that do not occur in $G$ as \emph{forbidden typed motifs} as their appearance in the future would indicate something strong.
For instance, perhaps an anomaly or malicious activity.
Other interesting insights and applications of typed graphlets are discussed and explored further in Section~\ref{sec:exploratory-analysis}.

In addition to the large collection of real-world networks used for evaluation, we also generated synthetic graphs from 4 different graph models including:
Erd\H{o}s-R\'enyi (ER) model~\cite{erdos1960evolution},
Chung-Lu (CL) graph model~\cite{chung2002connected},
Kronecker Product Graph Model (KPGM)~\cite{leskovec2010kronecker}, and
Watts-Strogatz Small-World (SW) graph model~\cite{watts1998collective}.
Since these graph models do not generate types/labels, we must assign them.
Given $L$ node types, we assign types to nodes uniformly at random such that $\frac{N}{L}$ nodes are assigned to every type. 
In these experiments, we used $L=5$ types.
Results are provided at the bottom of Table~\ref{table:runtime-perf} in the last 4 rows.
Overall, our proposed method is orders of magnitude faster than all other methods.
Compared to G-Tries, our method is between 502x and 7141x faster for these graphs.
Compared to GC, our method is between 89x and 21690x faster.

\begin{table}[t!]
\centering
\setlength{\tabcolsep}{8pt}
\renewcommand{\arraystretch}{1.2} 
\caption{Comparing the \emph{space} used by our approach to the state-of-the-art methods.
Note existing methods derive counts for the nodes, whereas the proposed method derives counts for each edge, and since $|E| > |V|$ then methods that compute counts for nodes should require less space.
}
\vspace{-3mm}
\label{table:space-results}
\small
\begin{tabularx}{1.0\linewidth}{lH @{} rr cc}
\toprule
& &
\multicolumn{1}{c}{\rotatebox{0}{\textsf{citeseer}}} & 
\multicolumn{1}{c}{\rotatebox{0}{\textsf{cora}}} & 
\multicolumn{1}{c}{\rotatebox{0}{\textsf{movielens}}} & 
\multicolumn{1}{c}{\rotatebox{0}{\textsf{web-spam}}}
\\ \midrule

\textbf{GC}\!~\cite{gu2018heterAlignment} & &
30.1MB & 50.4MB & 
ETL & 
ETL 
\\

\textbf{ESU}\!~\cite{fanmod} & & 
13.4MB & 46.2MB & 
ETL &
ETL 
\\

\textbf{G-Tries}\!~\cite{ribeiro2014discovering} & & 
161.9MB & 448.6MB & 
ETL & 
ETL 
\\

\midrule
\textbf{Ours} & &
\textbf{316KB} & \textbf{578KB} &
\textbf{22.5MB} &
\textbf{128.9MB} 
\\

\bottomrule
\multicolumn{5}{l}{\footnotesize $^{*}$ ETL = Exceeded Time Limit (24 hours / 86,400 seconds)} 
\\
\end{tabularx}
\end{table}

\subsection{Space Efficiency Comparison} \label{sec:exp-space-efficiency}
We theoretically showed the space complexity of our approach in Section~\ref{sec:space-complexity}.
In this section, we empirically investigate the space-efficiency of our approach compared to ESU (using fanmod)~\cite{fanmod}, G-Tries~\cite{ribeiro2014discovering}, and GC~\cite{gu2018heterAlignment}.
For this, we again use a variety of real-world networks.
Table~\ref{table:space-results} reports the space used by each method for a variety of real-world graphs.
Strikingly, the proposed approach uses between 42x and 776x less space than existing methods as shown in Table~\ref{table:space-results}.
These results indicate that our approach is space-efficient and practical for large networks.

\subsection{Parallel Speedup} \label{sec:exp-parallel-scaling}
This section evaluates the parallel scaling of the parallel algorithm (Section~\ref{sec:parallel}).
As an aside, this work describes the first parallel typed graphlet algorithm as all existing methods are inherently serial.
In these experiments, we used a two processor, Intel Xeon E5-2686 v4 system with 256 GB of memory.
None of the experiments came close to using all the memory.
Parallel speedup is simply $S_p = \frac{T_1}{T_p}$ where $T_1$ is the execution time of the sequential algorithm, and $T_p$ is the execution time of the parallel algorithm with $p$ processing units (cores).
In Figure~\ref{fig:parallel-scaling}, we observe nearly linear speedup as we increase the number of cores.
These results indicate the effectiveness of the parallel algorithm for counting typed graphlets in general heterogeneous graphs.

In addition, Table~\ref{table:unique-typed-motif-occur-syn-KPGM} demonstrates the parallel scalability of the approach as the number of types $L$ increases.
Notably, the parallel speedup of the approach remains constant indicating that parallel performance is independent of the number of types $L$ assigned to the nodes in the graph.

\begin{figure}[h!]
\centering
\includegraphics[width=0.65\linewidth]{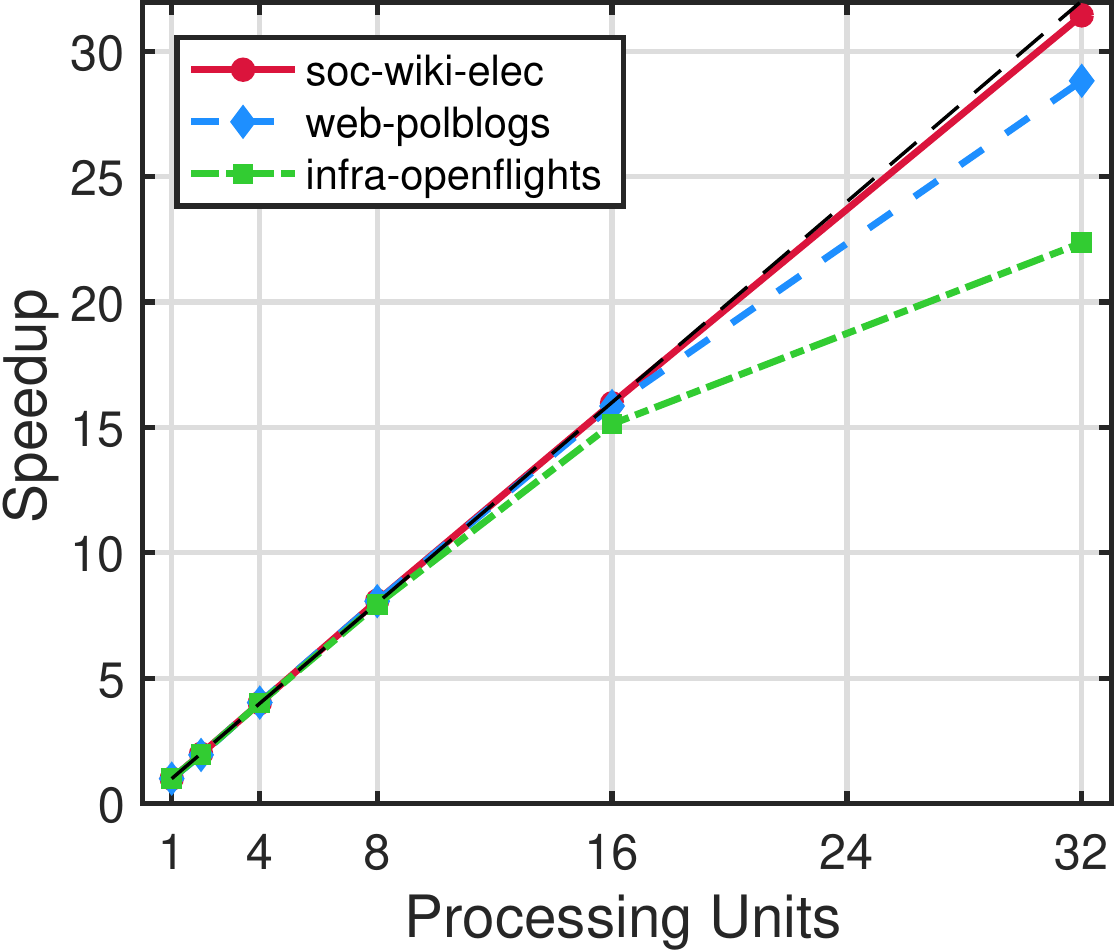}
\caption{Parallel speedup.
}
\label{fig:parallel-scaling}
\end{figure}

\subsection{Scalability} \label{sec:exp-scalability}
To evaluate the scalability of the proposed framework as the size of the graph grows (\ie, number of nodes and edges increase),  
we generate Erd\"{o}s-R\'{e}nyi graphs of increasing size (from 100 to 1 million nodes) such that each graph has an average degree of 10.
In Figure~\ref{fig:exp-runtime-ER}, we observe that our approach scales linearly as the number of nodes and edges grow large.
Furthermore, the scalability of the typed graphlet framework is independent of the number of types as shown in Figure~\ref{fig:exp-runtime-ER}.
As an aside, our approach takes less than 2 minutes to derive all typed $\{2,3,4\}$-node graphlets for a large graph with 1 million nodes and 10 million edges. 
Note that existing methods are not shown in Figure~\ref{fig:exp-runtime-ER} since they are unable to handle medium to large-sized graphs as shown previous in Table~\ref{table:runtime-perf}.

\begin{figure}[h!]
\centering
\includegraphics[width=0.60\linewidth]{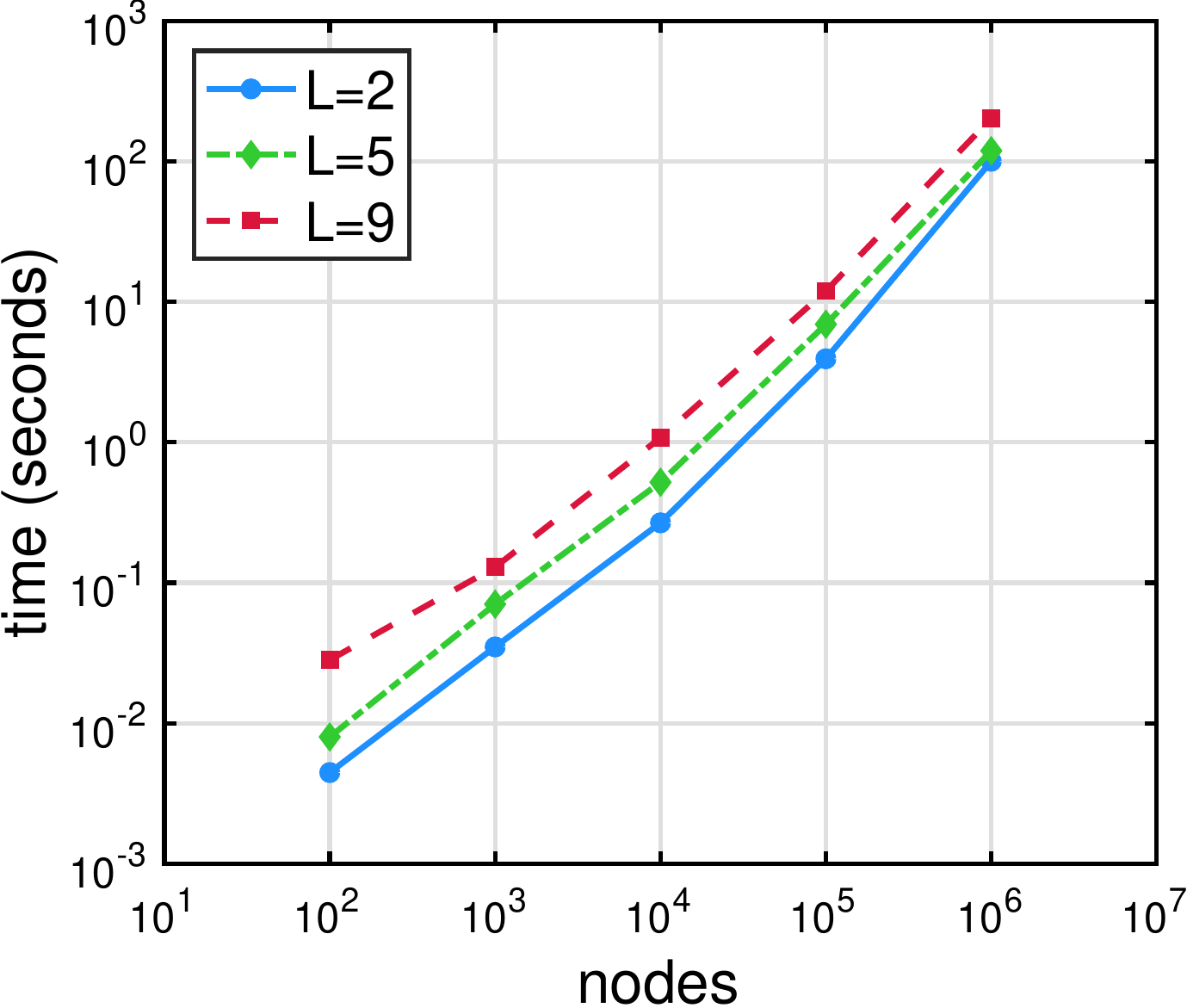}
\caption{
Runtime comparison on  Erd\protect\"{o}s-R\protect\'{e}nyi graphs with an average degree of 10.
Our approach is shown to scale linearly as the size of the graph grows increasingly large.
}
\label{fig:exp-runtime-ER}
\end{figure}

\subsection{Synthetic Graph Experiments} \label{sec:exp-syn-graph-exp}
In these experiments, we generate synthetic graphs.
For each graph, we vary the number of node types $L$ from 2 to 9, and measure the runtime performance as the number of node types increases as well as the impact in terms of space as $L$ increases.
Given $L \in \{2,\ldots,9\}$ node types, we assign types to nodes uniformly at random such that $\frac{N}{L}$ nodes are assigned to every type. 

\begin{table}[b!]
\centering
\renewcommand{\arraystretch}{1.2} 
\setlength{\tabcolsep}{3pt}
\caption{Comparing the number of unique typed motifs that occur for each induced subgraph as we vary the number of types $L$ using a KPGM graph with 3.3k nodes, 43.2k edges, average degree = 26, and max degree = 1.3K.
Speedup is shown in parenthesis.
}
\label{table:unique-typed-motif-occur-syn-KPGM}
\vspace{-3mm}
\fontsize{7.5}{8.5}\selectfont
\begin{tabularx}{1.0\linewidth}{Hc cc cccccc c H cc H 
HHHHHH@{}
}
\toprule

&& && &&&&&&&& \multicolumn{2}{c}{
\vspace{-3mm}
\textbf{time} (sec.)\quad\quad\;\;\;\;
}
\\

& $L$  &
\includegraphics[scale=0.8]{fig7.pdf} &
\includegraphics[scale=0.8]{fig8.pdf} &
\includegraphics[scale=0.15]{fig9.pdf} &
\includegraphics[scale=0.8]{fig10.pdf} &
\includegraphics[scale=0.8]{fig11.pdf} &
\includegraphics[scale=0.15]{fig12.pdf} &
\includegraphics[scale=0.14]{fig13.pdf} &
\includegraphics[scale=0.8]{fig14.pdf} &&&
\textbf{serial} & 
\textbf{parallel} - 4 cores & 
\\
\midrule

\multirow{9}{*}{\rotatebox{90}{\textbf{KPGM}}} & 
\textbf{2} & 4 & 4 & 5 & 5 & 5 & 5 & 5 & 5 &&&
8.11 & 2.08 \;(\emph{3.89x}) \\

& \textbf{5} & 35 & 35 & 70 & 70 & 70 & 70 & 70 & 70 &&&
8.94 & 2.26 \;(\emph{3.95x})  \\

& \textbf{9} & 165 & 165 & 495 & 495 & 495 & 495 & 495 & 495 &&&
10.37 & 2.62 \;(\emph{3.95x})  \\

\bottomrule
\end{tabularx}
\end{table}

\subsubsection{Impact on Performance}
We first investigate the runtime performance of our approach as the number of types $L$ increases from 2 to 9.
We use both a serial and parallel implementation of our method for comparison.
Results are shown in Table~\ref{table:unique-typed-motif-occur-syn-KPGM}.
Notably, the parallel speedup of the parallel algorithm is constant regardless of $L$.
Therefore, it is not impacted by the increase in $L$.
Furthermore, the runtime of both the serial and parallel algorithm increases slightly as $L$ increases. 
Notice the additional work depends on the number of unique typed graphlets (the sum of columns 2-9 in Table~\ref{table:unique-typed-motif-occur-syn-KPGM}) and not directly on $L$ itself. 
The total amount of unique typed graphlets substantially increases as $L$ increases from $2$ to $9$ as shown in Table~\ref{table:unique-typed-motif-occur-syn-KPGM}.
This is primarily due to the random assignment of types to nodes.
However, in sparse real-world graphs the total unique typed graphlets is typically much smaller as shown in Table~\ref{table:unique-typed-motif-occur}.

\begin{figure}[h!]
\includegraphics[width=0.6\linewidth]{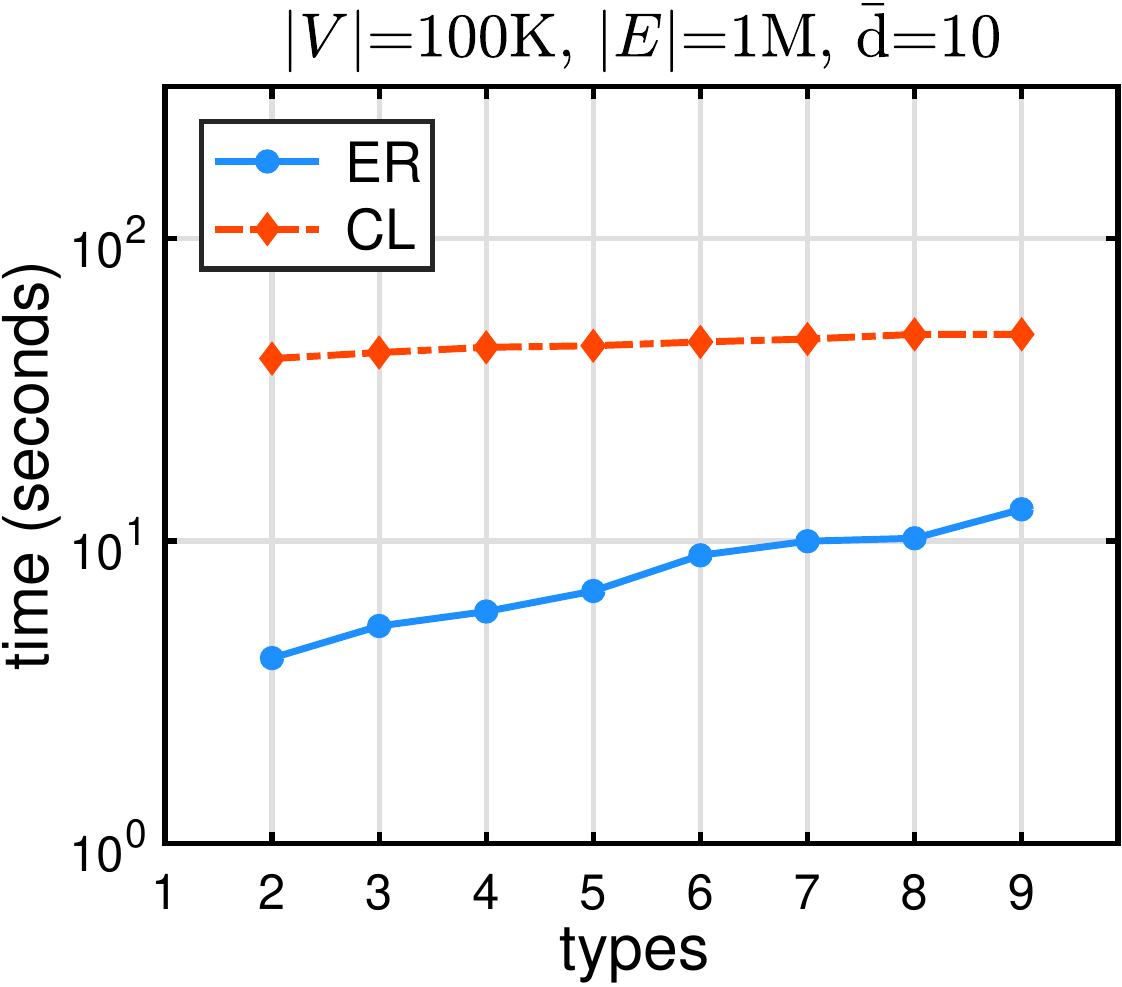}
\caption{Comparing runtime performance as the number of types increases.
In the case of CL graphs with a skewed degree distribution, the runtime is nearly constant as the number of types increases.
Note $\bar{\mathrm{d}} = \frac{1}{n} \sum d_i$ denotes average node degree.
}
\label{fig:runtime-ER-CL-vs-varyingNumTypes}
\end{figure}

To further understand how the structure of the graph impacts runtime, we generate an ER and Chung-Lu (CL) graph with 100K nodes, 1M edges, and average degree 10.
We vary the number of types $L$ and assign types to nodes uniformly at random as discussed previously.
Notice that both the ER and CL graph are generated such that they each have 100K nodes, 1 million edges, with average degree 10.
However, both graphs are structurally very different.
For instance, the degrees among the nodes in the ER graph are more uniform whereas the degree of the nodes in CL are skewed such that a few nodes have very large degree while the others have relatively small degree.
We observe in Figure~\ref{fig:runtime-ER-CL-vs-varyingNumTypes} that for CL graphs with a skewed degree distribution, the runtime of the approach as the number of types increases is essentially constant.
This result is important as most real-world graphs also have a skewed degree distribution (social networks, web graphs, information networks, etc.)~\cite{Girvan2002,Faloutsos1999}.
However, even in the case where the degrees are more uniform across the nodes, our approach still performs well as shown in Figure~\ref{fig:runtime-ER-CL-vs-varyingNumTypes}.

\subsubsection{Impact on Space}
Figure~\ref{fig:space-ER-CL-vs-varyingNumTypes} shows the memory (space) required by our approach as the number of types increases from $L\in \{2,\ldots,9\}$.
Both ER and CL graphs are shown to have similar space requirements. 
This is likely due to the random assignment of types to nodes.
This assignment represents a type of worst case since every edge is likely to have significantly more distinct typed graphlets compared to sparse real-world graphs.
This difference can be seen in Table~\ref{table:unique-typed-motif-occur}.

\begin{figure}[h!]
\includegraphics[width=0.6\linewidth]{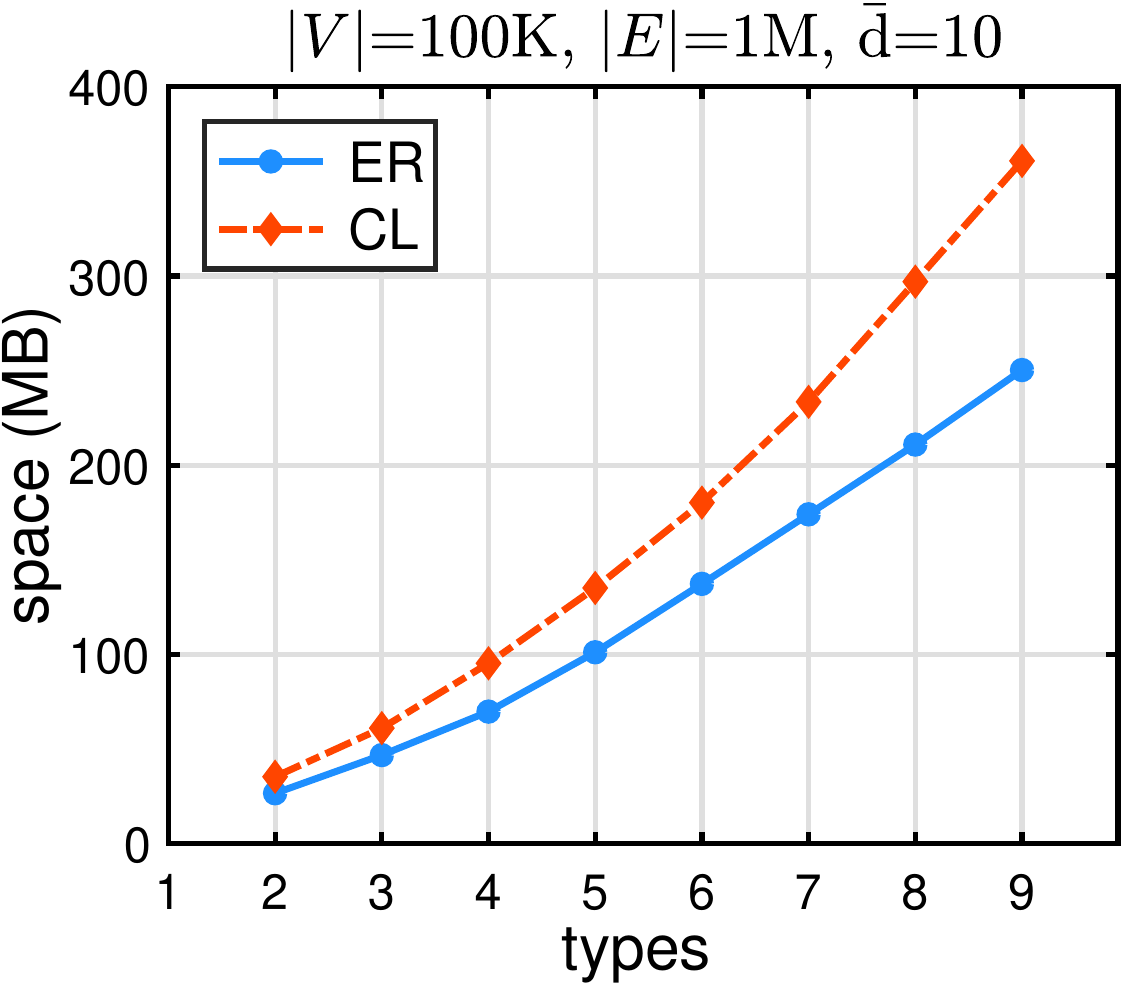}
\caption{Comparing memory (space) required by our approach as the number of types increases.
Note $\bar{\mathrm{d}} = \frac{1}{n} \sum d_i$ denotes average node degree.
}
\label{fig:space-ER-CL-vs-varyingNumTypes}
\end{figure}

\subsection{Exploratory Analysis} \label{sec:exploratory-analysis}
This section demonstrates the use of heterogeneous network motifs for graph mining and exploratory analysis.

\subsubsection{Political retweets}
The political retweets data consists of 18,470 Twitter users.
The graph has 61,157 links representing retweets.
There are 24,815 triangles in the political retweet network.
Triangles in this graph indicate that users retweeted by an individual also retweet each other (\ie, triangle = three users that have all mutually retweeted each other).
Triangles may represent users with similar interests.
However, triangles alone do not reveal any additional information about the users.
To study the (higher-order) structural characteristics of users in this network $\wrt$ their political orientation, we assign types to nodes based on their political leanings (\ie, left, right).
By using the political leanings as the type, we can directly apply the proposed approach to investigate a variety of interesting questions.
Interestingly, the 24,815 (untyped) triangles are distributed as follows:
\begin{center}
\begin{tabular}{l r@{} cccc @{}l r}
&&
\includegraphics[width=5mm]{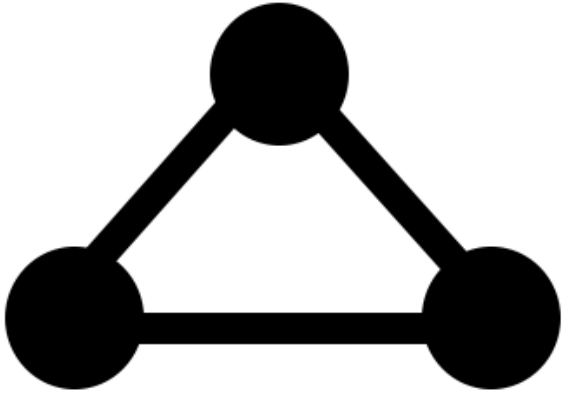} &
\includegraphics[width=5mm]{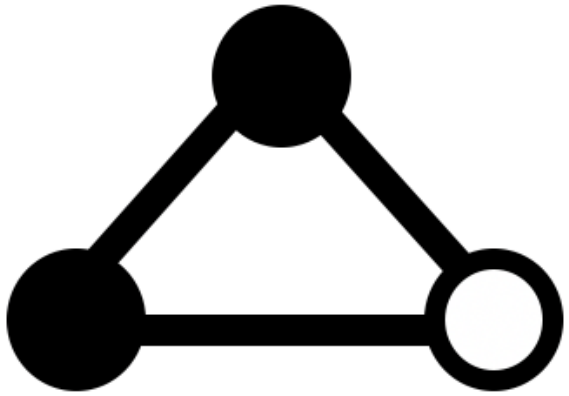} &
\includegraphics[width=5mm]{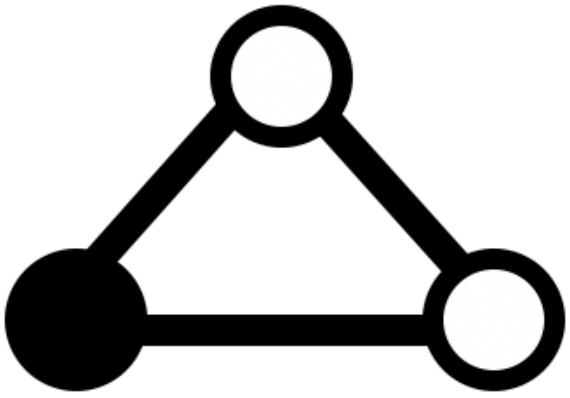} &
\includegraphics[width=5mm]{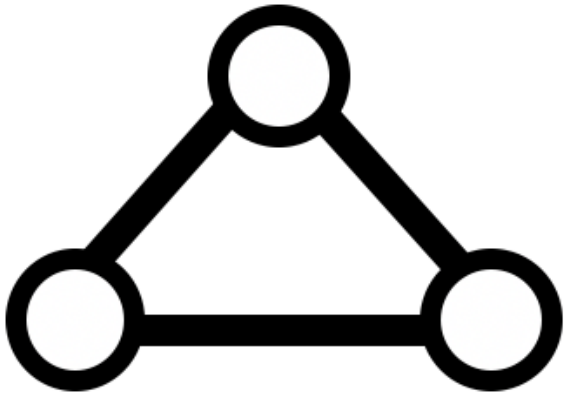} \\

\quad\quad\quad\quad\quad&
$\vp = \mathbf{\big[}\,$ &
0.608 & 0.003 & 0.001 & 0.388 &
$\,\big]$ &
$\quad\quad\quad\quad $
\\
\end{tabular}
\smallskip
\end{center}\noindent
Notably, we observe that $60.86\%$ and $38.79\%$ of the 24,815 triangles are formed among users with the same political leanings.
This implies that three users with the same political leanings are more likely to retweet each other than with users of different political leanings.
These results indicate the presence of homophily~\cite{mcpherson2001homophily} as users tend to retweet similar others.
Furthermore, these homogeneous typed triangles 
(\protect\includegraphics[width=4mm]{fig2}, 
\protect\includegraphics[width=4mm]{fig5}) 
account for $99.65\%$ of the 24,815 triangles.
Intuitively, this implies that the network consists of two tightly-knit communities of users of the same political leanings. The two communities are sparsely connected.
Typed triangles obviously contain significantly more information than untyped triangles.
This includes not only information about the local properties but also about the global structure of the network as shown above.
Obviously, untyped network motifs are unable to provide such insights as they do not encode the types, attribute values, or class labels associated with a network motif.
They only reveal the structural information independent of any important external information associated with the node.

We also investigated typed 4-clique motifs.
Strikingly, only 4 of the 5 typed 4-clique motifs that arise from $2$ types actually occur in the graph.
In particular, the typed 4-clique motif with 2 right users and 2 left users does not even appear in the graph.
This typed motif might indicate collusion between individuals from different political parties or some other extremely rare anomalous activity.
The other typed 4-cliques that are extremely rare are the typed 4-clique motif with 3 right (left) users and a single left (right) user.

\begin{figure}[h!]
\centering
\subfigure{\includegraphics[width=0.48\linewidth]{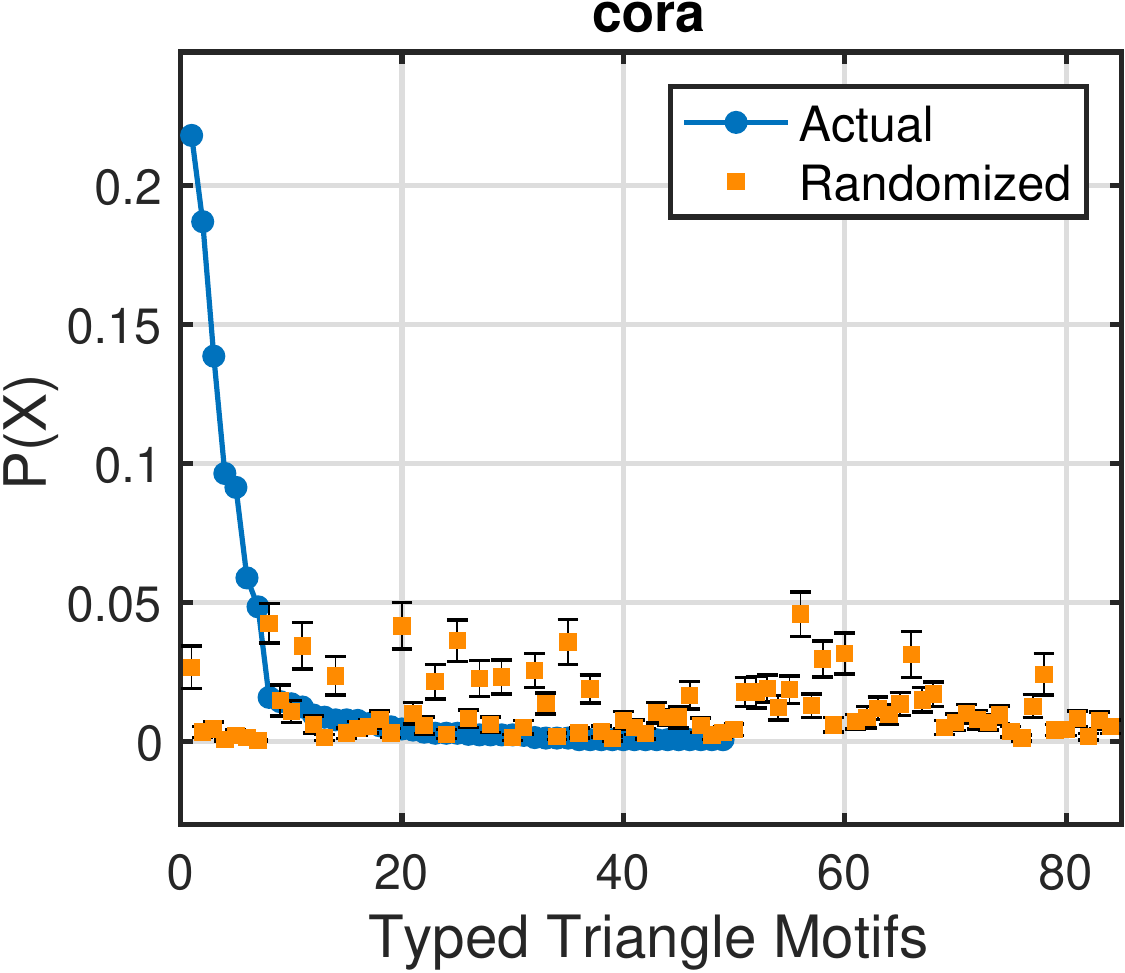}}
\hfill
\subfigure{\includegraphics[width=0.48\linewidth]{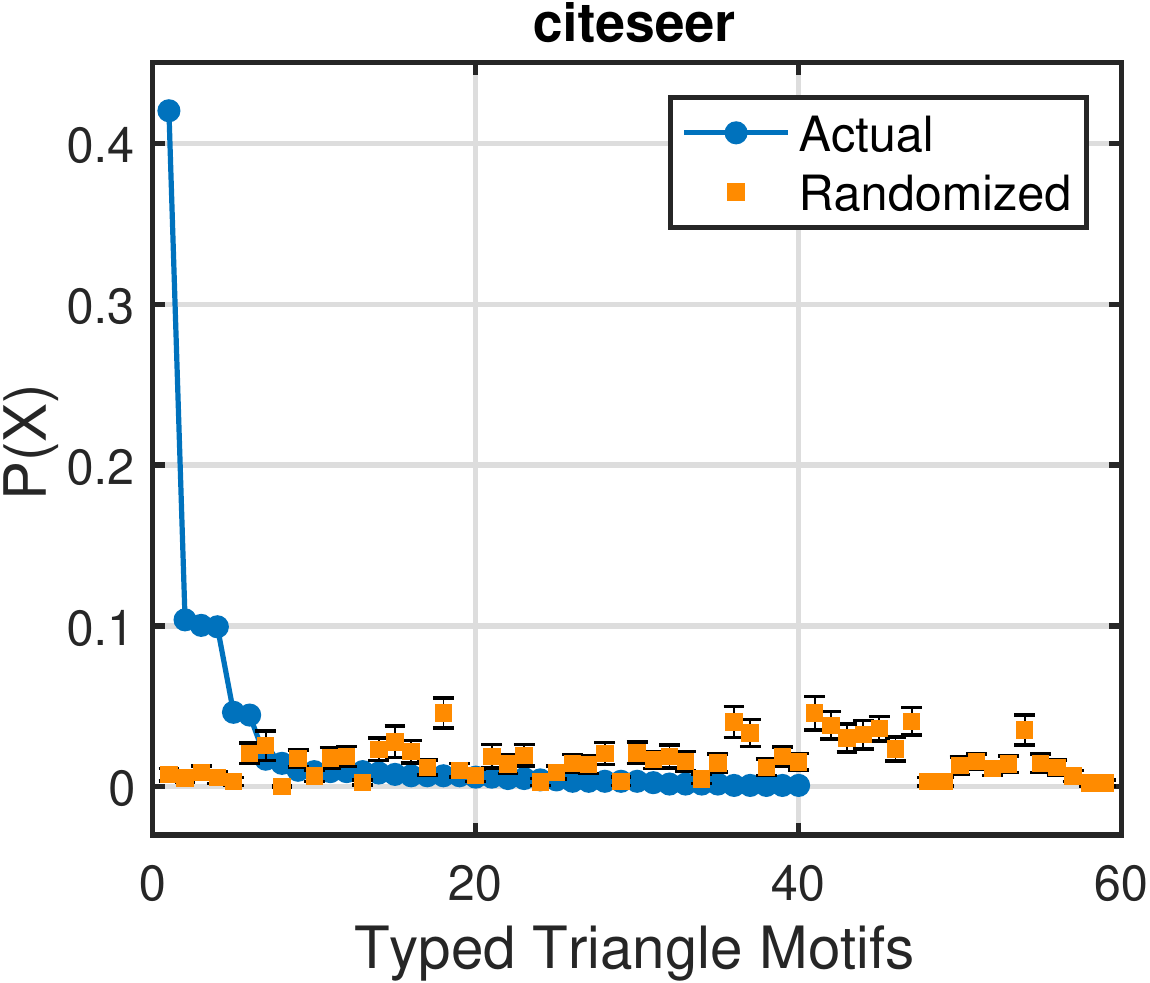}}

\vspace{-3mm}
\caption{
Comparing the actual typed triangle distribution to the randomized typed triangle distribution.
We compute 100 random permutations of the node types and run the approach on each permutation then average the resulting counts to obtain the mean randomized typed triangle distribution.
There are three key findings.
First, we observe a significant difference between the actual and randomized typed triangle distributions.
Second, many of the typed triangles that occur when the types are randomized, do not occur in the actual typed triangle distribution.
Third, we find the typed triangle distribution to be skewed (approximately power-lawed) as a few typed triangles occur very frequently while the vast majority have very few occurrences.
}
\label{fig:typed-tri-prob-dist-cora}
\end{figure}

\begin{figure*}[h!]
\centering
\subfigure{\includegraphics[width=0.32\linewidth]{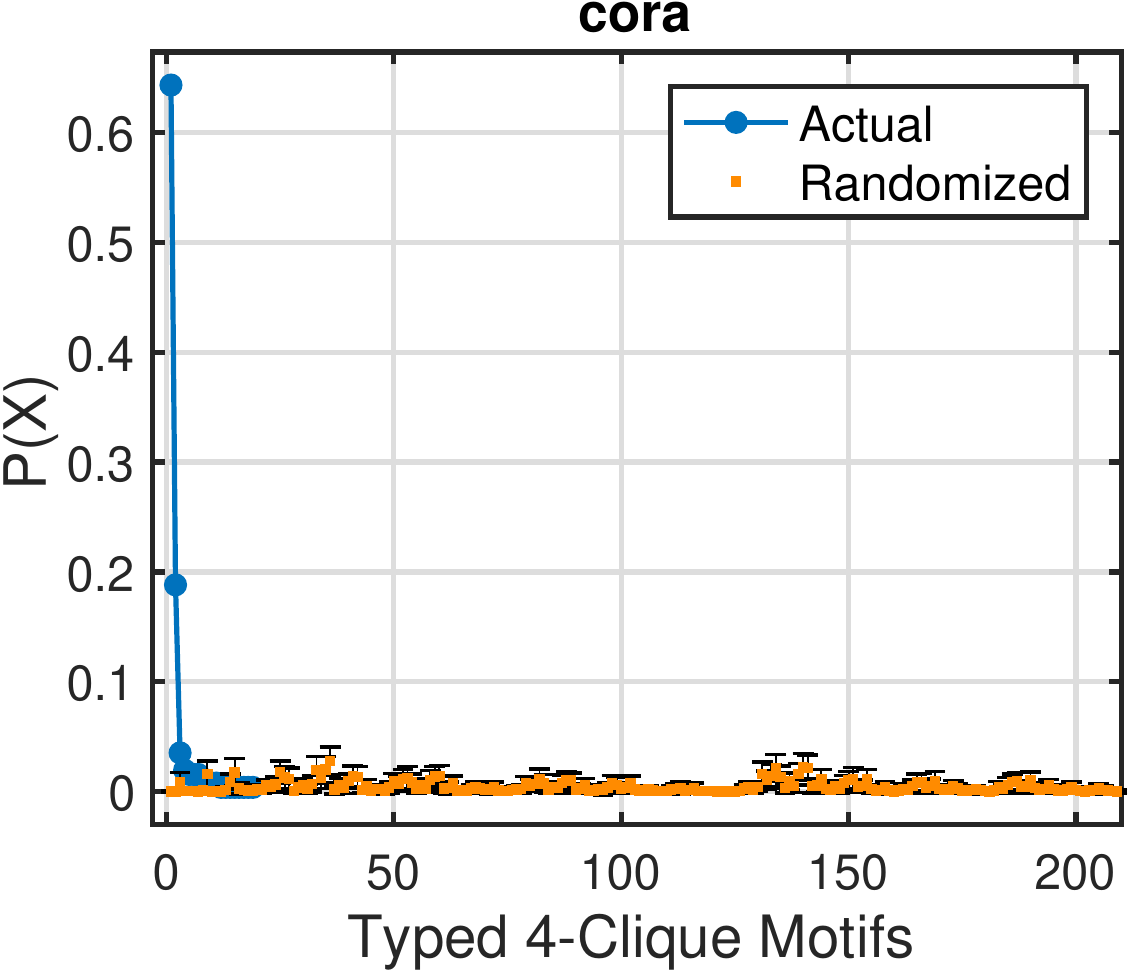}}
\hfill
\subfigure{\includegraphics[width=0.32\linewidth]{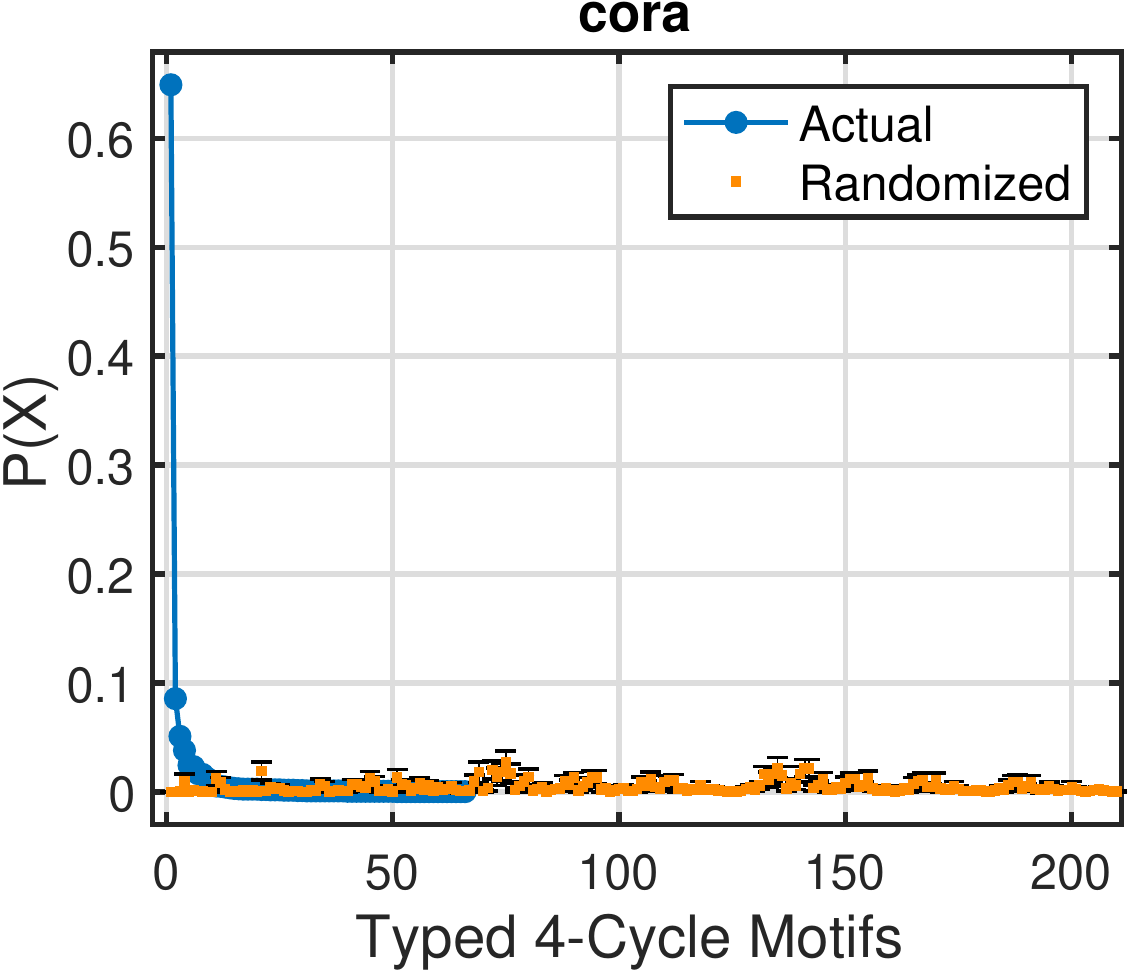}}
\hfill
\subfigure{\includegraphics[width=0.32\linewidth]{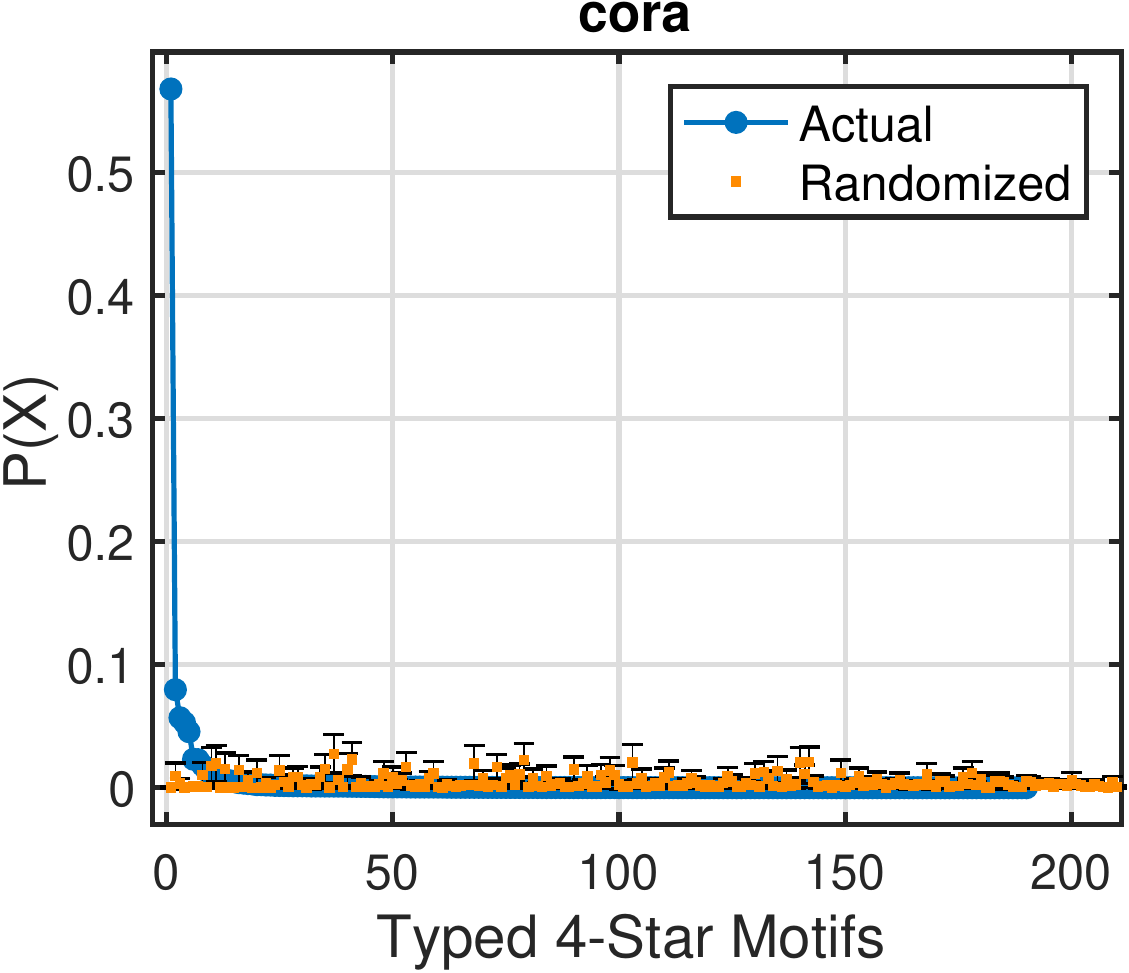}}

\subfigure{\includegraphics[width=0.32\linewidth]{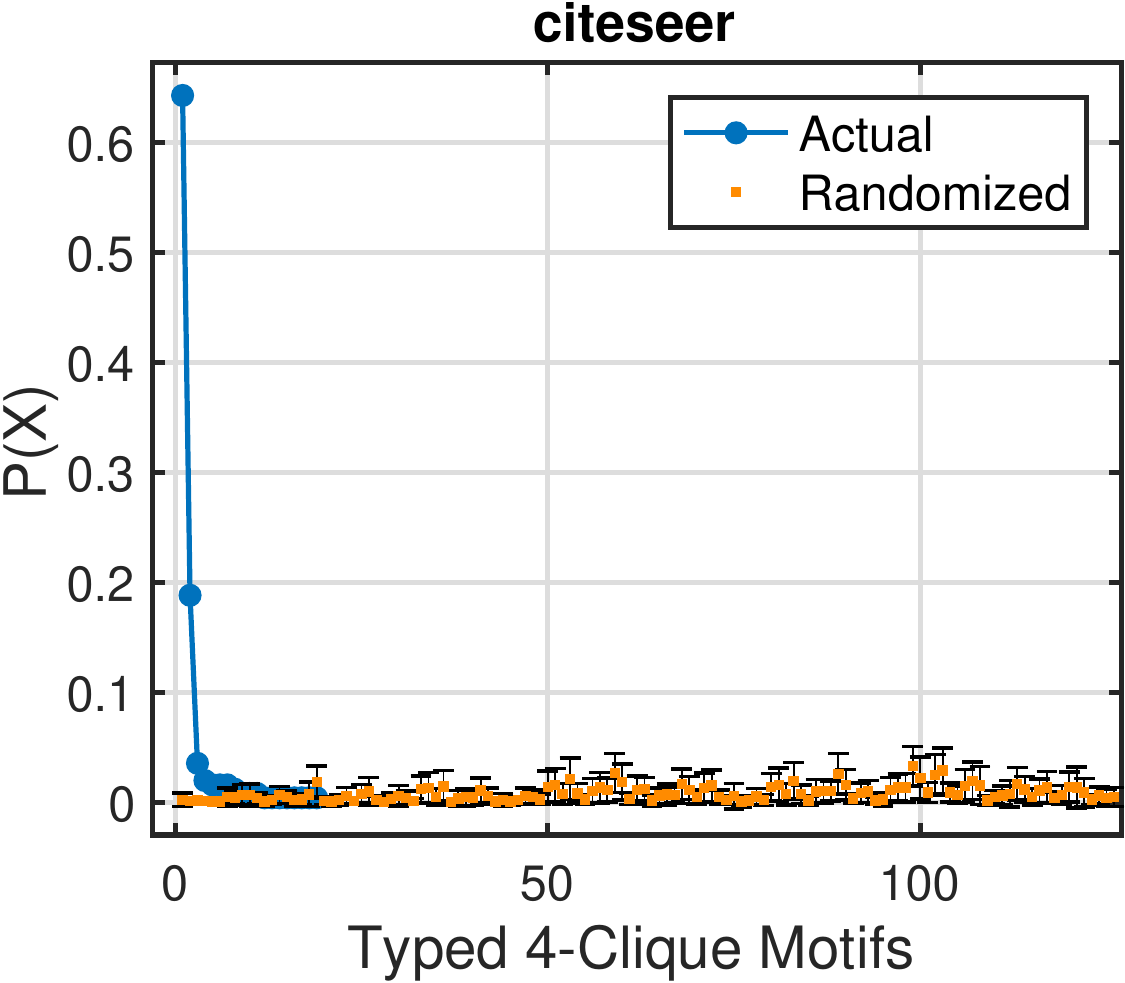}}
\hfill
\subfigure{\includegraphics[width=0.32\linewidth]{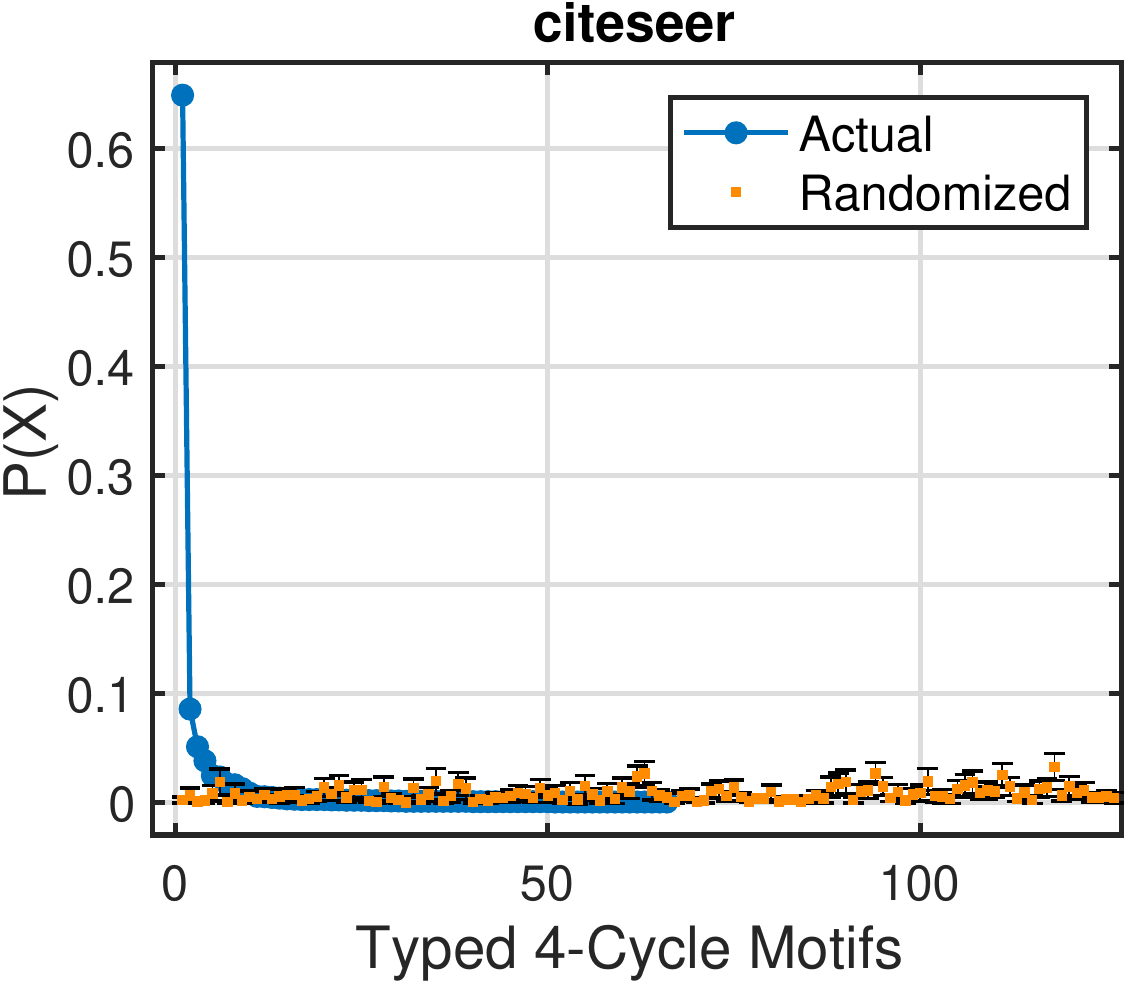}}
\hfill
\subfigure{\includegraphics[width=0.32\linewidth]{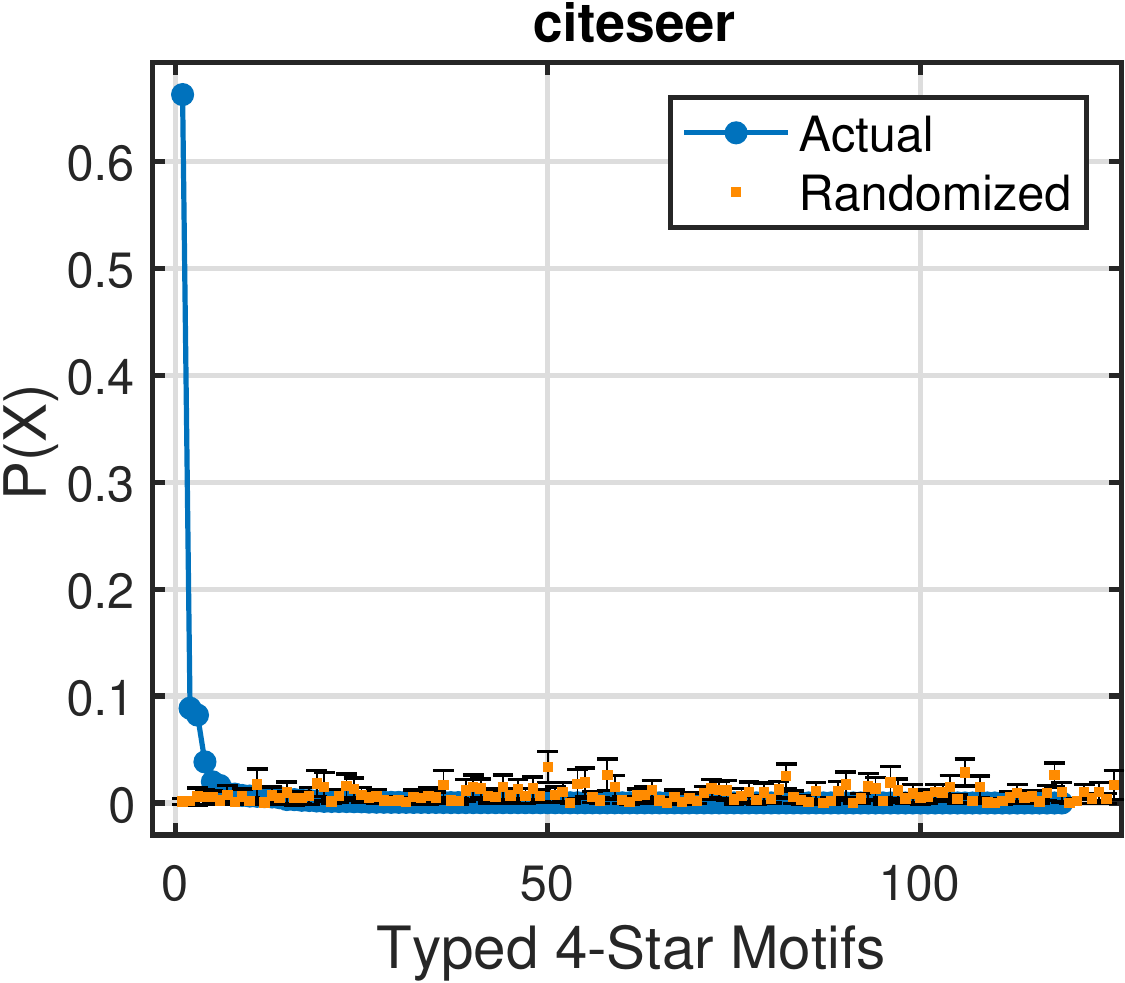}}

\vspace{-3mm}
\caption{
Comparing the actual \emph{typed} 4-clique, 4-cycle and 4-star motif distributions to the randomized typed distributions.
We compute 100 random permutations of the node types and run the approach on each permutation then average the resulting counts to obtain the mean randomized typed motif distribution.
There are three key findings.
First, we observe a significant difference between the actual and randomized typed motif distributions.
Second, many of the typed motifs that occur when the types are randomized, do not occur in the actual typed motif distribution.
Third, we find the typed motif distribution to be extremely skewed as a few typed motifs occur very frequently while the vast majority have very few occurrences (and many typed network motifs are even forbidden, in the sense that they do not occur at all in the graph).
}
\label{fig:typed-4-node-prob-dist-cora}
\end{figure*}

\subsubsection{Cora citation network}
The Cora citation network consists of 2708 scientific publications classified into one of seven types (class labels) that indicate the paper topic.
The citation network consists of 5429 links.
Using the proposed heterogeneous motifs, we find 129 typed 3-node motifs among the 168 possible typed 3-node motifs that could occur.
Notably, we observe the most frequent typed triangle motifs are of a single type.
Indeed, the first 7 typed triangle motifs with largest frequency in Figure~\ref{fig:typed-tri-prob-dist-cora} are of a single type.
Overall, these 7 typed triangle motifs account for 83.86\% of all typed triangle motifs.
This finding indicates \emph{strong homophily} among nodes.
For instance, the co-authors of a researcher are also typically co-authors, but more strongly, the co-authors are also all typically from the same research area.
Homophily has been widely studied in social networks and detecting it remains a challenging problem~\cite{LaFond2010}.
Unlike untyped motifs, typed motifs simultaneously capture the labeling and structural properties that lie at the heart of homophily~\cite{mcpherson2001homophily}.
Therefore, typed motifs provide a principled foundation for studying homophily in social networks~\cite{mcpherson2001homophily}.
In Figure~\ref{fig:typed-tri-prob-dist-cora}, we observe a large gap that clearly separates the 7 single-typed triangle motifs from the other typed triangle motifs with heterogeneous types.
Furthermore, only 49 out of the 84 possible typed triangle motifs (Table~\ref{table:typed-graphlets-example}) actually occur in the graph.
The 25 typed triangle motifs that do not occur are ``forbidden heterogeneous motifs" and can provide additional insights into the network and the processes governing the formation with respect to the types (research areas).

Figure~\ref{fig:typed-4-node-prob-dist-cora} investigates a variety of typed 4-node graphlet distributions (from most dense to least dense).
Strikingly, only 19 of the 210 possible typed 4-clique graphlets actually occur in $G$ when node types are randomly shuffled.
In the case of typed 4-node cycle graphlets, we observe 66 of the actual 210 possible typed 4-node cycle graphlets appear when node types are randomly shuffled.

\section{Related Work} \label{sec:related-work}
While the bulk of work on graphlets (network motifs) have focused on untyped/uncolored graphlets~\cite{rage,milo2002network,pgd,pgd-kais,shervashidze2009efficient11,ahmed16bigdata,zhang2013discovering,orca,hayes2013graphlet,rossi18tnnls}, there have been a few recent works for typed graphlets.\footnote{Typed graphlets are also called colored/heterogeneous graphlets or motifs.}
All of these works focus on the problem of counting typed graphlets for \emph{nodes} whereas this paper focuses on the problem of counting typed graphlets for \emph{edges} (or more generally, between a pair of nodes $i$ and $j$).
The earliest such work by Ribeiro~\etal~\cite{ribeiro2014discovering} proposed an approach called G-Tries for finding frequent and statistically significant colored motifs~\cite{fanmod,ribeiro2014discovering}.
This problem differs from the one we study in this work that focuses on deriving all such colored graphlets.
One recent work by Gu~\etal~\cite{gu2018heterAlignment} used a relaxed definition of colored graphlets for network alignment on very small heterogeneous networks.
This work focused mainly on the problem of network alignment for simple heterogeneous networks and not on the approach for deriving typed graphlets.
Nevertheless, the method GC used in that work differs from our approach in four significant ways.
First, while we leverage combinatorial relationships to derive a number of typed graphlets in $o(1)$ constant time, GC must enumerate all homogeneous graphlets in order to obtain their type/color configuration.
Therefore, our approach is significantly faster than GC as that approach requires a lot of extra work to compute the typed graphlets that we can derive in constant time.
Second, our approach is significantly more space-efficient and stores only the nonzero counts of the typed graphlets discovered at each edge.
Third, we parallelize our approach to handle large sparse networks.
As an aside, Gu~\etal~\cite{gu2018heterAlignment} claim the time complexity of their approach is equivalent to counting homogeneous graphlets.
This is only true if the homogeneous graphlet algorithm used enumerates all such homogeneous graphlets since this is the only way that step 2 of their approach could be performed, which requires them to obtain the colors of the nodes involved in the k-node graphlet found in step 1. 
However, there are much faster algorithms for homogeneous graphlets that avoid explicit enumeration of all such graphlets, \eg,~\cite{pgd,pgd-kais}.
In a similar fashion, our approach avoids enumerating all such graphlets and explicitly obtains the colors of the nodes involved in those graphlets (without actually enumerating or knowing the nodes involved in those graphlets), and therefore the computational complexity is actually equivalent to the best known homogeneous graphlet algorithm.
Other work by Carranza~\etal~\cite{higher-order-clustering-heter} used typed graphlets as a basis for higher-order spectral clustering on heterogeneous networks.

\section{Conclusion} \label{sec:conc}
\noindent
In this work, we generalized the notion of network motif to heterogeneous networks.
We proposed a fast, space-efficient, and parallel framework for counting all $k$-node typed graphlets.
We provide theoretical analysis of combinatorial arguments that capture the relationship between various typed graphlets.
Using these combinatorial relationships, the proposed approach is able to derive many typed graphlet counts directly in $o(1)$ constant time by leveraging the counts of a few typed graphlets.
Thus, the proposed approach avoids explicit enumeration of any nodes involved in those typed graphlets.
For every edge, we count a few typed graphlets and obtain the exact counts of the remaining ones in $o(1)$ constant time.
The time complexity of the proposed approaches matches that of the best untyped graphlet algorithm.
Empirically, our approach is shown to outperform the state-of-the-art in terms of runtime, space-efficiency, and scalability as it is able to handle large networks.
While existing methods take hours on small graphs with thousands of edges, our typed graphlet counting approach takes only seconds on networks with millions of edges.
Finally, unlike other approaches, the proposed approach is able to handle large-scale general heterogeneous networks with an arbitrary number of types and with millions or more edges 
while lending itself to an efficient and highly scalable (asynchronous \& lock-free) parallel implementation.

\balance
\bibliographystyle{ACM-Reference-Format}
\bibliography{paper}

\end{document}